\documentclass[11pt]{article}

\usepackage{amsmath, amsfonts, amssymb}
\usepackage{amsthm}
\usepackage{algorithm}
\usepackage{algorithmic}

\allowdisplaybreaks

\usepackage{microtype}
\usepackage{dsfont}
\usepackage{appendix}
\usepackage{enumitem}

\usepackage{natbib}

\usepackage{chngcntr}

\usepackage{xcolor}
\usepackage{graphicx}
\usepackage{float}
\usepackage{subfigure}
\usepackage{multirow}

\usepackage{geometry}
\usepackage{titlesec}

\titleformat{\paragraph}
{\normalfont\normalsize\bfseries}{\theparagraph}{1em}{}
\titlespacing*{\paragraph}
{0pt}{3.25ex plus 1ex minus .2ex}{1.5ex plus .2ex}

\usepackage{hyperref}
\hypersetup{
            colorlinks={true}, 
            linkcolor={blue}, 
            citecolor=blue,
            pdfencoding=auto, 
            psdextra
            }
            
\usepackage[capitalise,compress]{cleveref}

\numberwithin{equation}{section}

\newtheorem{definition}{Definition}[section]
\newtheorem{assumption}{Assumption}[section]
\newtheorem{lemma}{Lemma}[section]
\newtheorem{theorem}{Theorem}[section]
\newtheorem{proposition}{Proposition}[section]
\newtheorem{remark}{Remark}[section]
\newtheorem{corollary}{Corollary}[section]
\newtheorem{example}{Example}[section]

\Crefname{assumption}{Assumption}{Assumptions}
\Crefname{definition}{Definition}{Definitions}
\Crefname{lemma}{Lemma}{Lemmas}

\newlist{assumpenum}{enumerate}{1} 
\setlist[assumpenum]{label=(\roman*), ref=\theassumption~(\roman*)}
\crefalias{assumpenumi}{assumption} 

\newlist{lemmaenum}{enumerate}{1} 
\setlist[lemmaenum]{label=(\roman*), ref=\thelemma~(\roman*)}
\crefalias{lemmaenumi}{lemma} 

\newlist{propenum}{enumerate}{1} 
\setlist[propenum]{label=(\roman*), ref=\theproposition~(\roman*)}
\crefalias{propenumi}{proposition} 

\newcommand{\indep}{\rotatebox[origin=c]{90}{$\models$}}

\DeclareMathOperator{\argmin}{\mathop{\mathrm{argmin}}}

\makeatother
\usepackage{subfiles} 
\usepackage{import}

\title{A Sensitivity Analysis of the Surrogate Index Approach for Estimating Long-Term Treatment Effects\thanks{An earlier version of this paper: ``Distortion Surrogate Indices: 
 Estimating Long-Term Treatment Effects More Robustly'' was presented at ``Stochastic Dominance and Quantile-based Methods in Financial Econometrics Workshop" in honour of
Prof. Oliver Linton and Prof. Yoon-Jae Whang in June 2025, and ``IAER Econometrics Workshop 2025'' at Dongbei University of Finance and Economics in June 19-20, 2025. The current version of the paper was presented at ``CUHK Workshop: 
Recent Advances in Econometrics'' at Chinese University of Hong Kong in December 19, 2025. We thank the workshop participants for useful feedback. }}

\author{Yanqin Fan\thanks{Department of Economics, University of Washington. Email: fany88@uw.edu}, Carlos A. Manzanares\thanks{Amazon. Email: manzcarl@amazon.com}, 
Hyeonseok Park\thanks{Institute for Advanced Economic Research, Dongbei University of Finance and Economics. Email: hynskpark21@dufe.edu.cn}, and Yuan Qi\thanks{Department of Economics, University of Washington. Email: ayqi@uw.edu}}

\date{\today}

\begin{document}

\maketitle

\begin{abstract}
    This paper develops a sensitivity analysis of the surrogacy assumption for the surrogate index approach in  \cite{athey2025surrogate}. We introduce ``Weighted Surrogate Indices (WSIs)," the analog of the surrogate index under the surrogacy assumption. We show that under comparability, the ATE on WSI identifies the ATE on the long-term outcome when a copula of the treatment and the long-term outcome conditional on baseline covariates and surrogates is known. When the copula is unknown, we establish the identified set of the ATE on the long-term outcome. Furthermore, we construct debiased estimators of the ATE for any given copula and develop asymptotically valid inference in both point-identified and partially identified cases. Using data from a poverty alleviation program in Pakistan, we demonstrate the importance of sensitivity checks as well as the usefulness of our approach.

\vspace{10pt}
\end{abstract}

\newpage
\section{Introduction}

\indent In a variety of applications, researchers are interested in the effect of a treatment on outcomes, where the outcomes take a substantial amount of time to mature. If so, it may be useful to measure the effect of treatment on a set of intermediate or ``surrogate" outcomes, which are predictive of long-term outcomes. For example, when measuring the effect of a job training program on long-term employment and earnings outcomes, short-term employment and earnings may serve as useful surrogates. 

\indent A literature has characterized the assumptions required to identify average treatment effects of interventions through surrogates when researchers have access to two data samples, an experimental and an observational data sample (\cite{athey2025surrogate}). In this setting, the experimental data sample contains observations labeled with treatment assignments and surrogates. The observational data sample contains observations labeled with surrogates and long-term outcomes of interest. Both samples contain information on a set of pre-treatment characteristics. The goal is to identify the average treatment effect (ATE) on long-term outcomes of interest, even though long-term outcomes and treatments are observed in different data samples. 

\indent \cite{athey2025surrogate} introduce the Surrogate Index (SI) and show that under four assumptions, the ATE of the long-term outcome is identified as the ATE of the SI. \textit{The SI is the conditional expectation of the long-term outcome given the surrogates and pre-treatment characteristics in the observational data sample}. A critical assumption in \cite{athey2025surrogate} is the Prentice Criterion (\cite{prentice1989SurrogateEndpoints}) or the surrogacy assumption. The surrogacy assumption requires that, conditional on surrogates and pre-treatment characteristics, treatment assignment is independent of long-term outcomes. If this assumption holds, the effect of the treatment is fully mediated by the surrogates. However, the surrogacy assumption is often not satisfied (\citet{freedman1992}). In many cases, the ATE is only partially mediated by the surrogates. For example, in medical contexts, low-density lipoprotein cholesterol (LDL-c) is often used as a surrogate for long-term cardiovascular disease (CVD), given that low levels of LDL-c are correlated with better CVD outcomes. However, some hormone replacement therapies have been found to lower LDL-c but also increase CVD through other causal pathways (\cite{yetley2017SurrogateDiseaseMarkers}). In education, smaller class size may lead to changes in non-cognitive traits not fully captured by standardized exams ({\cite{heckman2006EffectsCognitiveNoncognitive}}). In online advertising, increasing customer ad-load may generate short-term ad revenue but also frustrate customers, leading to long-term purchase drops not well-proxied by short-term ad-revenue (\cite{hohnhold2015FocusingLongterm}). 

\indent When the surrogacy assumption fails, researchers have two available strategies.\footnote{When treatment is observed in the observational data set, the surrogacy assumption is not needed, see \cite{athey2025ExperimentalSelectionCorrectionEstimator}, \cite{chen_ritzwoller_2023}, 
\cite{obradovic2024TemporalLinks}, and
\cite{park2024DataCombinationUnderDynamicSelection}.} First, they can find different surrogates or add more of them. Better surrogates or adding proxies of confounding surrogates may meet the surrogacy assumption, see \cite{cai2024SurrogatesRepresentationLearning}. However, there are many applications where this strategy fails (\cite{bernard2023EstimatingLongtermTreatment}). Second, they can evaluate worst-case bounds on the ATE leading to the worst-case identified set, see \cite{athey2025surrogate} for binary outcomes. The worst-case identified set shows how much ATEs can vary when the surrogate assumption fails.

\indent This paper contributes a third strategy. This strategy extends the scope of applications of the SI approach in \cite{athey2025surrogate} from full mediation to partial mediation.  This is accomplished by modeling the joint distribution of the long-term outcome and the treatment via a copula, conditional on surrogates and pre-treatment covariates. 
When the copula is the independence copula, the surrogacy assumption holds. Conversely, a general non-independence copula implies failure of the surrogacy assumption. It characterizes partial mediation of the surrogate variables on the effect of treatment on the long-term outcome.

\indent We extend the identification result in \cite{athey2025surrogate} from full mediation (the independence copula) to partial mediation (any non-independence copula). Specifically,  we introduce ``weighted surrogate indices (WSIs)'' and show that the ATE on the long-term outcome is identified as the ``ATE on WSIs'', see \cref{Eq:WSI-ATE}. The WSIs reduce to the SI in \cite{athey2025surrogate} when the copula is the independence copula.   Numerically, we provide evidence of the sensitivity of the ATE to the strength of the global dependence between treatment and long-term outcome (which parameterizes the copula). We also show the robustness of the ATE to the shape of the copula. 
Using the identification result, we establish the sign of the surrogacy bias for stochastically monotone copulas.

\indent Our approach subsumes point identification of the ATE under the surrogacy assumption in  \cite{athey2025surrogate} and the worst case bounds as special cases. For example, at one extreme, when the long-term outcome and treatment are independent (conditional on surrogates and pre-treatment variables), the surrogacy assumption holds and our result reduces to the point identification of the ATE in \cite{athey2025surrogate}. At the other extreme, when the copula varies between the lower and upper bound copulas, our approach leads to the worst-case bounds on the ATE for any outcome type including binary outcomes as studied in \cite{athey2025surrogate}. Researchers can evaluate the risk of the surrogacy assumption failing on ATEs for their application by evaluating the worst-case bounds. 

\indent More importantly, our identification result allows researchers to explicitly model scenarios in-between these two extremes. This is especially helpful when an application allows researchers to rule out certain relationships \textit{a priori}. For example, in an educational setting, it may be reasonable to assume smaller class sizes are non-negatively related to future earnings after conditioning on short-term test score surrogates. These surrogates do not fully control for improvements in non-cognitive measures, which also improve future earnings. This can be achieved by varying the copula from an independence copula to the upper bound copula. Restricting the joint distribution of class size and future earnings in this way (conditional on short-term test scores) reduces the range of relevant ATE bounds. 

We develop a complete set of estimation and inference methods for both the point identified case, when the copula is known, as well as the partially identified case, when the copula is unknown. Specifically, we construct doubly-robust estimators of the ATE for any given copula including the worst-case bounds. We establish the asymptotic normality of our doubly-robust estimator for any given copula. We also establish the joint normality of the doubly-robust estimators of the worst-case bounds. The orthogonal moment functions for  non-independence copulas including the lower and upper Fr\'echet copulas are more complicated than those in \citet{chen_ritzwoller_2023} and \cite{athey2025surrogate}. To verify the general conditions in \citet{chernozhukov2018doubleML}, we adapt the technical proofs in 
\citet{chen_ritzwoller_2023} , \citet{dorn2024dvds}, and \citet{semenova2025generalizedleebounds} to our setting. 

Empirically, using data from a poverty alleviation program in Pakistan (obtained from \cite{banerjee2015MultifacetedProgramCauses}), we demonstrate that relaxing the surrogacy assumption can substantially alter conclusions. In particular, this often reverses the sign of the estimated treatment effect assuming surrogacy, which highlights the importance of these robustness checks. 

The rest of the paper is organized as follows. Section \ref{sec:setup} reviews the setup and identification of the ATE using the surrogacy assumption. Section \ref{sec:iden-copula} introduces weighted surrogate indices, showing the conditions under which the ATE is point identified. It also presents a numerical illustration of how the ATE changes as a function of Kendall's tau, which describes the concordance relationship between the treatment and long-term outcome (conditional on surrogates). Section \ref{sec:PartialId} focuses on the worst-case bounds, their debiased estimation, and inference results for the ATE when the surrogacy assumption fails. Section \ref{sec:est-copula} generalizes results in Section \ref{sec:PartialId} to a general copula and partial identification. Section \ref{sec:empirical} applies the proposed framework to a household poverty alleviation dataset \citep{banerjee2015MultifacetedProgramCauses}, conducting sensitivity analysis and deriving partial identification bounds. Section \ref{sec:conclusion} concludes. The technical details for the main result and the proof of the results in Section \ref{sec:est-copula} are presented in a supplementary appendix.

\section{Setup and Identification Under Surrogacy Assumption}
\label{sec:setup}
To be self-contained, 
this section formally reviews the setup and three critical assumptions used in \cite{athey2025surrogate} to identify the ATE on the primary (long-term) outcome. This identification is achieved using two different datasets 1) an experimental data sample, which contains information on treatment assignment and baseline characteristics (but not the primary outcome), and 2) an observational data sample, which contains information on baseline characteristics and the primary outcome (but not treatment assignment). We also restate the surrogate index representation of the ATE in  \cite{athey2025surrogate}.

Let $W_i$ denote the treatment status (binary) of unit $i$, $Y_i(0),Y_i(1)$ denote the potential primary (long-term) outcomes,  $S_i(0),S_i(1)$ denote the potential short-term outcomes (surrogacy variables), and $X_i$ denote a vector of baseline covariates that are not affected by treatment. Further, let $Y_i:=W_iY_i(1)+(1-W_i)Y_i(0)$ and $S_i:=W_iS_i(1)+(1-W_i)S_i(0)$. 

\begin{assumption}\label{assumption: Random Sample}
We have a single random sample of size $N$ drawn from the joint distribution of $(P_i, X_i, S_i,W_i, Y_i)$, where we observe for each unit in the sample $Z_i$, where 
$Z_i:=(P_i, X_i, S_i,\mathds{1}_{P_i=E}W_i, \mathds{1}_{P_i=O}Y_i)$. 
\end{assumption}

\cref{assumption: Random Sample} states that sample information comes from two data sets: one labeled observational data that contains observations on  $(P_i=O, X_i, S_i, Y_i)$ and the other labeled experimental data that contains observations on $(P_i=E, X_i, S_i,W_i)$. 

\begin{remark}
We point out that \cref{assumption: Random Sample} can be replaced with the following assumption:  $\{X_i, S_i, Y_i\}_{P_i =O}$ and $\{X_i, W_i, S_i\}_{P_i =E}$ are two independent samples and each is a random sample. This makes it clear that treatment status is not found in the observational data sample. 
\end{remark}

The experimental data satisfy the following assumption. 
\begin{assumption}[Unconfounded Treatment Assignment/Strong Ignorability]\label{assumption: Unconfounded}
  (i) \[W_i \ \indep \ (Y_i(0),Y_i(1),S_i(0),S_i(1))|X_i,P_i=E;  \] 
  (ii) $0<\rho(x)<1$ for all $x\in \mathcal{X}$, where $\rho(x):=\mathbb{P}(W_i=1|X_i=x,P_i=E)$ is the propensity score.
\end{assumption}

Let $\tau$ denote the ATE on the primary outcome in the population from which the experimental sample is drawn:
\begin{equation}
    \tau:=\mathbb{E}\left[Y_i(1)-Y_i(0)\mid P_i=E\right].
\end{equation}
We are interested in identifying $\tau$ using the sample information satisfying \cref{assumption: Random Sample}.

\cite{athey2025surrogate} adopts two additional assumptions. 

\begin{assumption}[Surrogacy]\label{assumption: Surrogacy}
    (i) $W_i \ \indep \ Y_i|S_i,X_i, P_i=E;$
    (ii) $ 0<\rho(s,x)<1$ for all $s\in\mathcal{S}, x\in \mathcal{X}$ and $0<\mathbb{P}(P_i=E)<1$, where $\rho(s,x):=\mathbb{P}(W_i=1|S_i=s,X_i=x,P_i=E)$ is the surrogacy score. 
\end{assumption}
\Cref{assumption: Surrogacy} (surrogacy) implies that the conditional distribution of $Y_i$ given $(W_i, S_i,X_i, P_i=E)$ is the same as that of  $Y_i$ given $ (S_i,X_i, P_i=E)$. Thus, given $X_i$, the surrogate variable $S_i$ completely mediates the effect of $W_i$ on the primary outcome $Y_i$ for the experimental sample. For expositional simplicity, we refer to this scenario as full mediation.

\begin{assumption}[Comparability of Samples]\label{assumption:Comparability}
    (i) $P_i \ \indep \ Y_i|S_i,X_i;$
    (ii) $\varphi(s,x)<1$ for all $s\in\mathcal{S}, x\in \mathcal{X}$, where $\varphi(s,x)=\mathbb{P}(P_i=E|S_i=s,X_i=x)$.
\end{assumption}

\begin{definition}[Surrogate Index, \cite{athey2025surrogate}]
The surrogate index (SI) is the conditional expectation of the primary outcome given the surrogate outcomes and the pre-treatment variables, conditional on the sample:
\[\mu(s,x,p)=\mathbb{E}\left[Y_i|S_i=s,X_i=x,P_i=p\right], \mbox{ } p=O,E.\]
\end{definition}

The SI $\mu(S_i,X_i,O)$ is identified from the observational data. 
Theorem 1 in \cite{athey2025surrogate} shows that \cref{assumption: Random Sample}-\cref{assumption:Comparability} imply that $\tau$ is identified as the ATE on the SI, $\mu(S_i,X_i,O)$. We restate this result in the lemma below. The proof relies on the following expressions:
\begin{align}\label{eq: unconfounded}
\mathbb{E}[Y_i(1)|P_i=E]&=\mathbb{E}\left[\frac{1}{\rho(X_i)}\mathbb{E}\left[Y_iW_i|S_i,X_i,P_i=E \right]\mid P_i=E\right],\\
\mathbb{E}[Y_i(0)|P_i=E]
& =\mathbb{E}\left[\frac{1}{1-\rho(X_i)}\mathbb{E}\left[Y_i(1-W_i)|S_i,X_i,P_i=E \right]|P_i=E\right].
\end{align}
\begin{lemma}[\cite{athey2025surrogate}]\label{lemma:IndeC}
    Under \cref{assumption: Random Sample}-\cref{assumption:Comparability}, $\tau$ is identified as 
\begin{align*}
     \tau& =
        \mathbb{E}\left[ \mu(S_i,X_i,O)\frac{W_i }{\rho(X_i)} \mid P_i = E \right] 
  -  \mathbb{E}\left[  \mu(S_i,X_i,O)\left(\frac{1-W_i }{1-\rho(X_i)}\right)\mid P_i = E \right]\\
  &=  \mathbb{E}\left[\frac{\rho(S_i,X_i) }{\rho(X_i)(1-\rho(X_i))} \mu(S_i,X_i,O)  \mid P_i = E \right] 
  -  \mathbb{E}\left[ \frac{ 1 }{1-\rho(X_i)}\mu(S_i,X_i,O)\mid P_i = E \right].
\end{align*}
\end{lemma}

\begin{remark} 
Subsequent work such as \cite{yang2024TargetingLongTermOutcomes} extends the SI approach in \cite{athey2025surrogate} to CATE estimation and policy learning using SI $\mu(S_i,X_i,O)$. Specifically, the CATE denoted as $\tau(x)$ is identified as 
\begin{align*}
     \tau(x)=\mathbb{E}\left[\mu(S_i,x, O)\mid W_i=1,X_i=x, P_i=E \right]-\mathbb{E}\left[\mu(S_i,x, O)\mid W_i=0,X_i=x, P_i=E \right]
\end{align*} 
and the first-best policy is $\mathds{1}(\tau(x)>0)$.
\end{remark}

\section{Point Identification Without Surrogacy Assumption}
\label{sec:iden-copula}

The surrogate index $\mu(S_i,X_i,O)$ identified from the observational data plays a critical role in \cite{athey2025surrogate}: under \cref{assumption: Random Sample}-\cref{assumption:Comparability}, $\tau$ is identified as the ATE on the SI, $\mu(S_i,X_i,O)$.  \cite{athey2025surrogate} discuss the plausibility and implications of violation of either \cref{assumption: Surrogacy} or \cref{assumption:Comparability}. Under \cref{assumption:Comparability}, \cite{athey2025surrogate} state that 
without \cref{assumption: Surrogacy}, the ATE on the SI may not identify the ATE on the primary outcome. Motivated by this and the abundant empirical evidence on the possible failure of \cref{assumption: Surrogacy}, we propose a new framework to study identification of the ATE on the primary outcome relaxing the surrogacy assumption.

By the conditional version of Sklar's Theorem, there exists a copula $C_o(u,v|S_i,X_i, P_i=E)$, $(u,v)\in [0,1]^2$, such that
\begin{align} \label{eq:model}
 (Y_i, W_i)|S_i, X_i, P_i=E&\sim    C_{o}(F_Y(\cdot|S_i,X_i.P_i=E), F_W(\cdot|S_i,X_i,P_i=E)|S_i,X_i.P_i=E).
\end{align}
\begin{assumption} \label{assumption:knowncopula}
Suppose that $C_o$ in \cref{eq:model} is known.
\end{assumption}
 Since $W$ is discrete, $C_o$ is not unique. However, as we show later, for continous outcomes, the identification result depends on the unique sub-copula only. In the special case that 
 the sub-copula is the independence sub-copula, \cref{assumption:knowncopula} is equivalent to \cref{assumption: Surrogacy} and the surrogate $S_i$ fully mediates the effect of $W_i$ on $Y_i$ given $X_i$. 

\begin{example}[Families of Copulas]
    (i) A Gaussian copula with constant correlation $\vartheta$ takes the following form: 
    \[C_o(u,v|S_i,X_i, P_i=E) = \Phi_\vartheta(\Phi^{-1}(u), \Phi^{-1}(v)) = C_o(u,v|P_i=E),\]
    where $\Phi(\cdot)$ is the cumulative distribution function of the standard normal distribution and $\Phi_\vartheta(\cdot, \cdot)$ is the cumulative distribution function of the standard bivariate normal distribution with correlation $\vartheta\in [0,1]$.

    (ii)
Archimedean copulas are a widely used class of copulas defined by a generator function $\phi$:
\[
C(u_1, \dots, u_n) = \phi\left(\phi^{-1}(u_1) + \cdots + \phi^{-1}(u_n)\right),
\]
where $\phi: [0, \infty) \to [0, 1]$ is a strictly decreasing, convex function satisfying $\phi(0) = 1$ and $\phi(\infty) = 0$. Notable examples of Archimedean copulas are Clayton with generator function $\phi(t) = (1 + t)^{-1/\vartheta}$ for $\vartheta\in[-1,\infty)\setminus\{0\}$; Gumbel copula with $\phi(t) = \exp(-t^{1/\vartheta})$ for $\vartheta\in[1,\infty)$; and Frank copula with $\phi(t) = -\frac{1}{\vartheta} \log\left(1 - (1 - e^{-\vartheta}) e^{-t}\right)$ for $\vartheta\in\mathbb{R}\setminus\{0\}$.
    
    (iii)
Let $C_-(u,v):=\max(u+v-1,0)$ and $C_+(u,v):=\min(u,v)$ denote the Fr\'echet-Hoeffding lower and upper bound copulas. When $C_o=C_-$ or $C_o=C_+$, \cref{assumption:knowncopula} implies that $Y_i$ and $W_i$ are comonotonically dependent on each other given $S_i,X_i$ so that $S_i$ does not mediate any effect of $W_i$ on $Y_i$ given $X_i$.
\end{example}

In the rest of this section, we establish point identification of the ATE on the primary outcome under \cref{assumption:knowncopula}, extending the identification under full mediation or \cref{assumption: Surrogacy} in \cite{athey2025surrogate} to known partial mediation or \cref{assumption:knowncopula}. 

For notational convenience, we sometimes ignore the conditional arguments in the copula and write $C_o(u,v)$ instead of $C_o(u,v|S_i,X_i, P_i=E)$ and denote the ATE as $\tau_{C_o}$ to indicate its dependence on the copula $C_o$.

\subsection{Weighted Surrogate Indices}

Our identification strategy relies on the Weighted Surrogate Index (WSI) introduced in the following. 

\begin{definition}[Weighted Surrogate Indices] \label{definition:DSIP}
Let $(U,V)\sim C_o(u,v)$, where $C_o$ is defined in \cref{assumption:knowncopula}. For $\alpha\in (0,1)$, let $C_o(\alpha|u) := \Pr(V \le \alpha |U = u)$.
We define two Weighted Surrogate Indices associated with $w=0,1$ for each $p=E,O$ as
\begin{align*}
   \mu_{C_o, w}(S_i,X_i,p) = \int_{0}^{1} F_Y^{-1}(u|S_i,X_i, P_i=p)\sigma_{C_o, w}(u; 1- \rho(X_i,S_i)) d u, 
\end{align*}
where 
$    \sigma_{C_o, w}(u; \alpha) =\left(w - C_o(\alpha|u)\right)/(w-\alpha).$

\end{definition}
The WSIs in \cref{definition:DSIP} extend the SI in \cite{athey2025surrogate}. We show in \cref{theorem: known copula} that the ATE of the primary outcome is identified as ATE of the WSIs defined in \cref{Eq:WSI-ATE}. 
When the outcomes are continuous and $C_o$ is smooth in $u$,
the  WSIs depend only on the unique sub-copula. To see this, we note that $\mu_{C_o, 1}$ ($\mu_{C_o, 0}$) depends on $C_o(1-\rho(S_i, X_i)|u) = \partial_u C_o(u, 1-\rho(S_i, X_i)$ (or $C_o(u, 1-\rho(S_i, X_i)$) only: $C_o(u, 1-\rho(S_i, X_i)$ is unique because the sup-copula of $Y$ and $W$ given $S_i, X_i, P_i = E$ is unique.
 When the sub-copula is the independence sub-copula, $\sigma_{C_o, w}(u;\alpha)  = 1$  and 
    \[\mu_{C_o, w}(S_i,X_i,O) =\int_{0}^{1} F_Y^{-1}(u|S_i,X_i, P_i=O) d u=\mu(S_i,X_i,O), \mbox{ for } w=0,1\]   
where $\mu(S_i,X_i,O)$ is the SI in \cite{athey2025surrogate}. For a non-independence sub-copula, the weights $\sigma_{C_o, 1}(u;\alpha)$ and $\sigma_{C_o, 0}(u;\alpha)$ are generally non-constant functions of $u\in[0,1]$ and are not equal leading to two different surrogate indices $\mu_{C_o, w}(S_i,X_i,O)$ for $w=0,1$.

\begin{remark}\label{remark:DRM}
It is interesting to observe that WSIs are closely related to two distinct classes of functions, one in finance and risk management and the other in social choice and welfare. Specifically, let 
\begin{align*}
    r := \int_0^1 F_Y^{-1}(u) \sigma(u) du.
\end{align*}
When $\sigma(\cdot)$ is a non-decreasing function, $r$ is a Distortion Risk Measure, where higher ranked $Y$'s are given higher weight, see
 \cite{pichler2015} and \cite{pflug2006SubdifferentialRepresentationsRisk}. When $\sigma(\cdot)$ is a non-increasing function, $r$ is known as Rank-dependent Social Welfare Function, where higher ranked $Y$'s are given lower weight, see \cite{yaari1987DualTheoryChoice}. 
For a general copula $C_o$, the weight $\sigma_{C_o, w}(u;\alpha)$ in WSIs may not be monotone.  
 \end{remark}
 
\begin{example}[AVaR and WSIs for Upper and Lower Bound Copulas] \label{ex:SpecificCopulas}
When $\sigma(u) = \mathds{1}(u \in (\alpha, 1])/(1-\alpha)$ for some $\alpha\in [0,1)$, $r$ is the well-known AVaR of $Y$ at level $\alpha$:
 \[AVaR_{\alpha}(Y)=\frac{1}{1-\alpha}\int_\alpha^1 F_Y^{-1}(u)du.\]
It is insightful to examine the WSIs when the outcome and treatment are perfectly dependent on each other conditional on the surrogates and pre-treatment covariates. 
When $C_o=C_+$, for any $\alpha\in (0,1)$, it holds that 
\begin{align*}
    \sigma_{C_{+}, 1}(u;\alpha) = \frac{\mathds{1}(u \in (\alpha, 1])}{1-\alpha} \mbox{ and }
    \sigma_{C_{+}, 0}(u;\alpha) = \frac{\mathds{1}(u \in [0, \alpha])}{\alpha}.
\end{align*}
 Thus for $w=1$, the top $\rho(X_i,S_i)$ percentile of the conditional distribution of $Y_i$ is given equal positive weight and the bottom $1-\rho(X_i,S_i)$ percentile is given zero weight; for $w=0$, the bottom $1-\rho(X_i,S_i)$ percentile is given equal positive weight and the top $\rho(X_i,S_i)$ is given zero weight. Consequently,  
\begin{align*}
\mu_{C_+, 1}(S_i,X_i,O)
&=\int_{0}^{1} F_Y^{-1}(u|S_i,X_i, P_i=O)\frac{\mathds{1}(u \in (1- \rho(X_i,S_i), 1])}{\rho(X_i,S_i)} d u\\
&=AVaR_{1- \rho(X_i,S_i)}(Y_i\mid S_i,X_i,P_i=O) \mbox{ and}\\
\mu_{C_+, 0}(S_i,X_i,O)&=-AVaR_{\rho(X_i,S_i)}(-Y_i\mid S_i,X_i,P_i=O).
 \end{align*}
 
 Similarly, for $C_o=C_-$, 
\begin{align*}
    \sigma_{C_-, 1}(u;\alpha) = \frac{\mathds{1}(u\in[0,1-\alpha])}{1-\alpha}\mbox{ and }
    \sigma_{C_-, 0}(u;\alpha) = \frac{\mathds{1}(u\in(1-\alpha,1])}{\alpha}.
\end{align*}
As a result, $\mu_{C_-, 1}(S_i,X_i,O)$ is the conditional mean of $Y_i$ for the bottom $\rho(X_i,S_i)$ percentile of the conditional distribution of $Y_i$ and $\mu_{C_-, 0}(S_i,X_i,O)$ is the conditional mean of $Y_i$ for the top $1-\rho(X_i,S_i)$ percentile.
\end{example}

\subsection{Identification of ATE}

Theorem \ref{theorem: known copula} below shows that ``the ATE on WSIs'' defined on the right hand side of \cref{Eq:WSI-ATE} identifies the ATE of the primary outcome thus extending \cref{lemma:IndeC} or Theorem 1 in \cite{athey2025surrogate} under the surrogacy assumption to partial mediation. 

Consider $\mathbb{E}[Y_i(1)|P_i=E]$. Under \Cref{assumption: Unconfounded} (unconfoundedness), \cref{eq: unconfounded} implies that  
\begin{equation*}
\mathbb{E}[Y_i(1)|P_i=E]=\mathbb{E}\left[\frac{\rho(X_i,S_i)}{\rho(X_i)}\mathbb{E}\left[Y_i|W_i=1,S_i,X_i,P_i=E \right]\mid P_i=E\right].
\end{equation*}
Similarly,
\[\mathbb{E}[Y_i(0)|P_i=E]
 =\mathbb{E}\left[\frac{1-\rho(X_i,S_i)}{1-\rho(X_i)}\mathbb{E}\left[Y_i|W_i=0,S_i,X_i,P_i=E \right]|P_i=E\right].\]
Below we extend Proposition 2 (iii) in \cite{athey2025surrogate} under the surrogacy assumption to any copula $C_o$.
\begin{proposition}\label{proposition:DSIPI}
 Under \cref{assumption:Comparability}, it holds that $$ \mu_{C_o, w}(S_i,X_i,O)=\mu_{C_o, w}(S_i,X_i,E)=\mathbb{E}[Y_i | W_i=w, X_i, S_i, P_i=E] \mbox{ for } w = 0, 1.$$ 
\end{proposition}
\cref{proposition:DSIPI} implies that $\mathbb{E}\left[Y_i|W_i=w,S_i,X_i,P_i=E\right]=\mu_{C_o, w}(S_i,X_i,E)$ and under comparability, $\mu_{C_o, w}(S_i,X_i,E)$ is identified as $\mu_{C_o, w}(S_i,X_i,O)$ for any copula $C_o$ and hence \cref{theorem: known copula} holds. 

\begin{theorem}[Known Sub-copula]\label{theorem: known copula}
    Suppose \cref{assumption: Random Sample}, \cref{assumption: Unconfounded}, \cref{assumption:Comparability}, and \cref{assumption:knowncopula} hold. Then $\tau_{C_o}$ is identified as 
\begin{align}\label{Eq:WSI-ATE}
     \tau_{C_o} 
& =\mathbb{E}\left[ \mu_{C_o, 1}(S_i,X_i,O)\frac{W_i }{\rho(X_i)} \mid P_i = E \right] 
  -  \mathbb{E}\left[  \mu_{C_o, 0}(S_i,X_i,O)\left(\frac{1-W_i }{1-\rho(X_i)}\right)\mid P_i = E \right]\\
        &=
        \mathbb{E}\left[ \frac{\rho(X_i,S_i)}{\rho(X_i)[1-\rho(X_i)]}   \mu_{C_o, 1}(S_i,X_i,O)  \mid P_i = E \right]  -  \mathbb{E}\left[ \frac{1}{1-\rho(X_i)} \mu(S_i,X_i,O) \mid P_i = E \right].\nonumber
\end{align}
 When the conditional distribution of $Y_i$ given $S_i,X_i,P_i=O$ is degenerate,  $   \tau_{C_o}=   \tau_{\Pi} $ and is thus identified even if $C_o$ is unknown.
\end{theorem}

When the outcomes are continuous and $C_o$ is smooth in $u$, the WSIs depend on the unique sub-copula only and $\tau_{C_o}$ is identified from the sub-copula. 
Equipped with the WSIs $\left(\mu_{C_o, 1}(S_i,X_i,O), \mu_{C_o, 0}(S_i,X_i,O)\right)$, \cref{theorem: known copula} allows to identify ATE of the primary outcome regardless of full or partial mediation of $S_i$ on the effect of $W_i$ on $Y_i$ as long as $C_o$ is known which includes the lower and upper bound copulas.

\begin{remark} 
Analogously to \cite{yang2024TargetingLongTermOutcomes}, under \cref{assumption: Random Sample}, \cref{assumption: Unconfounded}, \cref{assumption:Comparability}, and \cref{assumption:knowncopula}, the CATE denoted as $\tau_{C_o}(x)$ is identified as
\begin{align*}
     \tau_{C_o}(x)=\mathbb{E}\left[\mu_{C_o,1}(S_i,x, O)\mid W_i=1,X_i=x, P_i=E \right]-\mathbb{E}\left[\mu_{C_o, 0}(S_i,x, O)\mid W_i=0,X_i=x, P_i=E \right]
\end{align*} 
and the first-best policy is $\mathds{1}(\tau_{C_o}(x)>0)$.
\end{remark}

\subsection{Surrogacy Bias---Stochastically Monotone Copulas}
Theorem 4 (ii) in \cite{athey2025surrogate} provides an expression for the surrogacy bias. For a copula $C_o$, it follows directly from \cref{theorem: known copula} that 
\begin{equation}\label{eq:SB}
\tau_{C_o}-\tau_{\Pi}=   \mathbb{E}\left[ \frac{\rho(X_i,S_i)}{\rho(X_i)[1-\rho(X_i)]}    \left\{\mu_{C_o, 1}(S_i,X_i,O)-\mu(S_i,X_i,O)\right\}  \mid P_i = E \right].
\end{equation}

For the class of stochastically monotone copulas, we will establish the sign of the surrogacy bias. 
\begin{definition}\label{definition:Sincreasing}
    Let $(U,V)\sim C_o$. Then $V$ is stochastically increasing (decreasing) in $U$, if $C_o(\alpha|u) = \Pr(V \le \alpha |U = u)$ decreases (increases) in $u$ for every $\alpha$.
\end{definition}
Commonly used copulas are monotone copulas. For example, the Gaussian and Frank copulas are monotonically increasing when $\vartheta > 0$ and monotonically decreasing when $\vartheta < 0$. The Clayton copula is monotonically increasing for $\vartheta > 0$ and monotonically decreasing for $\vartheta \in (-1,0)$. The Gumbel copula is monotonically increasing. \cref{definition:Sincreasing} implies that when $V$ is stochastically increasing (decreasing) in $U$, $\sigma_{C_o, 1}(u; \alpha)$ is an increasing (decreasing) function of $u\in [0,1]$ for every $\alpha$ and $\sigma_{C_o, 0}(u; \alpha)$ is a decreasing (increasing) function of $u\in [0,1]$ for every $\alpha$. As a result, we expect $\tau_{C_o}$ to be larger (smaller) than $\tau_{\Pi}$ when $V$ is stochastically increasing (decreasing) in $U$. We show in the rest of this section that this is indeed the case. 

\begin{definition}[Concordance Order]
    The copula $C_1$ is smaller than the copula $C_2$ in concordance order denoted as $C_1\prec_c C_2$ iff $C_1(u,v)\leq C_2(u,v)$ for all $(u,v)\in [0,1]^2$.
\end{definition}

 It follows from the proof of \cref{proposition:DSIPI} that for any copula $C$, 
 \begin{align*}
\mu_{C, 1}(s,x,O)
=\frac{1}{\rho(s,x) }\int\int ywdC(F_Y(y|s,x, E),F_W(w|s,x, E)|s,x, E).
\end{align*}
This and Theorem 1 of \cite{cambanis_1976_inequalities} imply
the statement for $\mu_{C, 1}$ in \cref{proposition:DSIP} below.
The statement for $\mu_{C, 0}$ follows from that of $\mu_{C, 1}$ and the following relation:  $$(1-\rho(s,x))\mu_{C, 0}(s,x,O)= \mathbb{E}[Y] - \rho(s,x)\mu_{C, 1}(s,x,O).$$

\begin{proposition}\label{proposition:DSIP}
 For any $(s,x)\in \mathcal{S}\otimes\mathcal{X}$, if $C_1(.,.|s,x, E)\prec_c C_2(.,.|s, x, E)$, then $ \mu_{C_1, 1}(s,x,O) \leq  \mu_{C_2, 1}(s,x,O) $
and $ \mu_{C_1, 0}(s,x,O) \geq  \mu_{C_2, 0}(s,x,O) $. 
\end{proposition}

  If $\{C_o(\cdot|u)\}$ is stochastically increasing in $u$, then $\Pi\prec_c C_o$, see Sections 2.8 and 8.3 of \cite{joe2014DependenceModelingCopulas}. The following corollary follows from \cref{proposition:DSIP} and \cref{eq:SB}.

\begin{corollary}[The Sign of the Surrogacy Bias]
      Suppose \cref{assumption: Random Sample}, \cref{assumption: Unconfounded}, and \cref{assumption:Comparability} hold. If $\{C_o(\cdot|u)\}$ is stochastically increasing (decreasing) in $u$, then $\tau_{C_o}-\tau_{\Pi}> 0 \mbox{ }(<0)$. 
\end{corollary}

   Consequently, if $\{C_o(\cdot|u)\}$ is stochastically increasing (decreasing) in $u$, then  the ATE of the SI under-estimates (over-estimates) $\tau_{C_o}$, the ATE of the long term outcome. 

\subsection{A Numerical Illustration}
\label{sec:numerical}
In this section, we use several parametric families of copulas to gauge the sensitivity of $\tau_{C_o}$ on the shape of $C_o$ and the strength of global dependence by varying their parameters.

We consider the Gaussian copula and several copulas from the Archimedean family. For simplicity, we assume that the conditional copula is the same as the unconditional copula. 
For common (bivariate) copulas such as the Gaussian, Clayton, Gumbel, and Frank copulas, the dependency parameter can be expressed in terms of Kendall's tau $\varrho_K$, a rank-based measure of monotonic dependence. 
This parameterization makes dependency strength interpretable and comparable across different copula families. 
Appendix \ref{sec:numerical details} in the supplementary appendix collects the mapping between $\varrho_K$ and $\vartheta$ for the aforementioned copulas.

For illustration, consider the following data generating process for both the experimental and observational data samples:
    \[X_i\sim U[0,1], \quad W_i\sim Bernoulli(\rho), \quad S_i = X_i + W_i + \eta_i^S, \text{ where } \eta_i^S\sim\mathcal{N}(0,1),\]
    \[Y_i = S_i + 0.5X_i + \eta_i^Y, \text{ where } \eta_i^Y\sim\mathcal{N}(0,1).\]
We investigate $\rho\in\{0.1, 0.5, 0.9\}$. Following Assumption \ref{eq:model}, $W_i|S_i,X_i =I\{\epsilon_i > 1 -  \rho(S_i, X_i)\}$, where $\epsilon_i | S_i, X_i, P_i=E \sim U[0, 1]$, and the copula structure is for $\epsilon_i$ and $\eta_i^Y$, independent of $S_i, X_i$. More specifically, we set $C_o(u,v;\vartheta|S_i, X_i, P_i=E) = C_o(u,v;\vartheta|P_i=E)$, for $C_o(u,v;\vartheta|P_i=E)$ being the Gaussian, Clayton, Gumbel, and Frank copulas. For each copula family, we solve for $\vartheta$ for each $\varrho_K$ in the corresponding grid of appropriate $\varrho_K$ values. 
Note that the Clayton copula is more commonly used to model positive dependence, while the Gumbel copula exclusively models positive dependence. Therefore, the range of $\varrho_K$ is adjusted to ensure the corresponding $\vartheta$ falls within its valid range. 

In Figure \ref{Figure: illustrative Archimedean conditional}, we plot how 
$\tau_{C_o}$ in Theorem \ref{theorem: known copula} changes with $\varrho_K$, with computational details in Appendix \ref{sec:numerical details}. Consistent with the DGP, Figure \ref{Figure: illustrative Archimedean conditional} shows that for any $\rho\in(0,1)$, if $C_o(u, v; \vartheta \mid P_i = E)$ is the independence copula ($\varrho_K=0$), the true long-term ATE is $\tau_{\Pi}=1$. The surrogacy bias $\tau_{C_o}-\tau_{\Pi}$ increases in magnitude as $|\varrho_K|$ increases. When $\rho=0.5$, the threshold for $\tau$ to change sign is $\varrho_K\approx-0.55$, which is consistent across the copulas considered here that allow for negative dependence (Gaussian and Frank). In other words, if the negative dependence between $W$ and $Y$ is strong enough with Kendall's tau smaller than $-0.55$, assuming an independence copula will produce the wrong sign for the long-term ATE. For $\rho=0.1$ or $0.9$, the same threshold for the Gaussian copula is around $-0.41$. Note that the form of the copula is not crucial for $\rho = 0.5$. This is intuitive because values of $\rho(s, x)$ tend to be around $0.5$ as well, making the conditional dependence $C_o(1-\rho(s,x)|u)$ in $\sigma_{C_o,1}(\cdot;\cdot)$ relatively insensitive to the copula family. In contrast, when $\rho$ takes more extreme values like $0.1$ or $0.9$, differences in tail dependence across copula families lead to more pronounced variation in conditional behavior, thereby affecting $\tau_{C_o}$. In general, in a randomized controlled trial with relatively balanced treatment and control groups, if $S$ depends only weakly on $W$ (i.e., $S$ carries limited information about $W$), the choice of copula serves mainly as a functional tool for achieving the desired dependence level, while $\varrho_K$ captures the essential dependence information relevant for $\tau_{C_o}$.
\begin{figure}[H]
    \centering
    \subfigure{
        \includegraphics[width=0.318\textwidth]{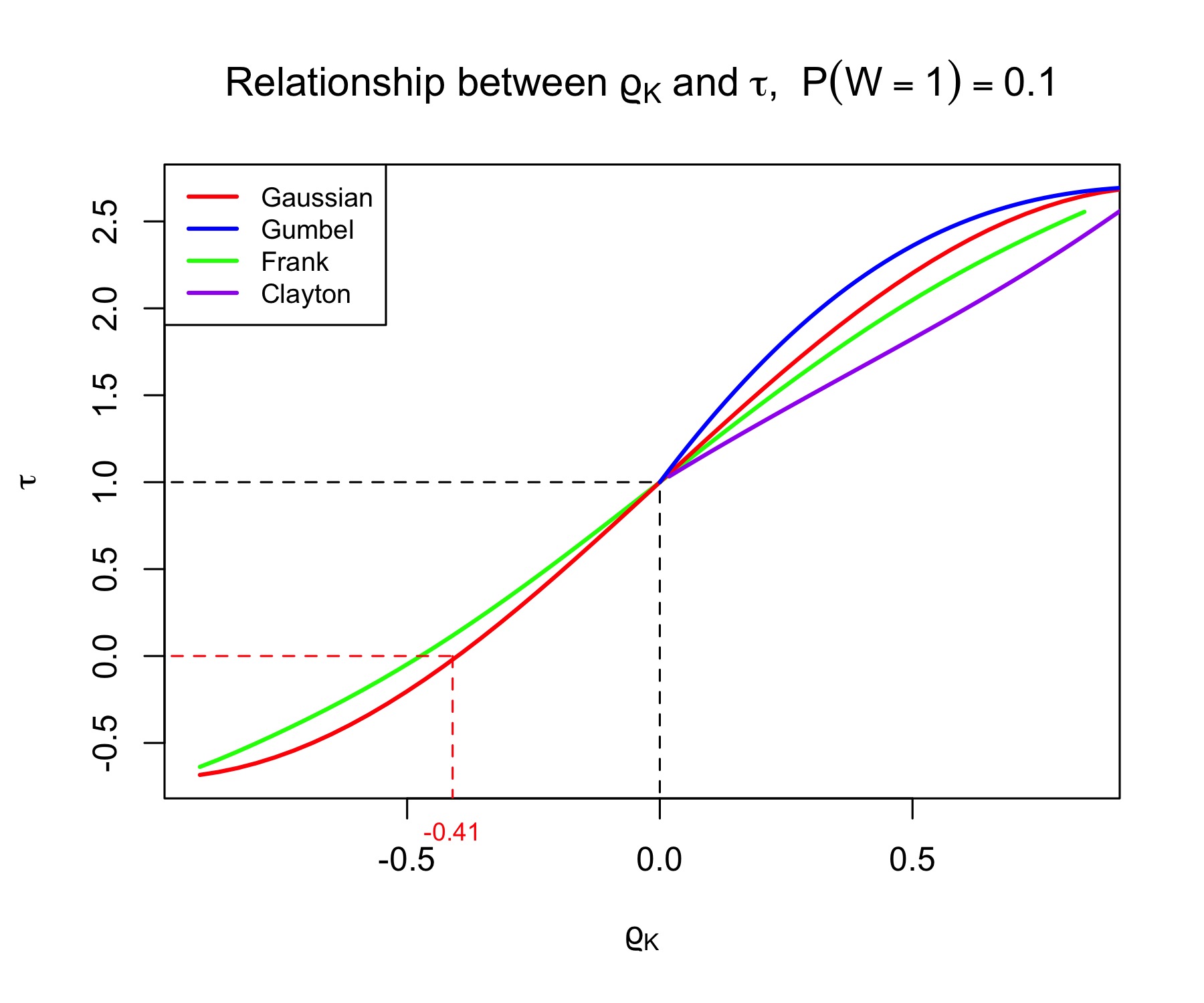}}
    \subfigure{
        \includegraphics[width=0.318\textwidth]{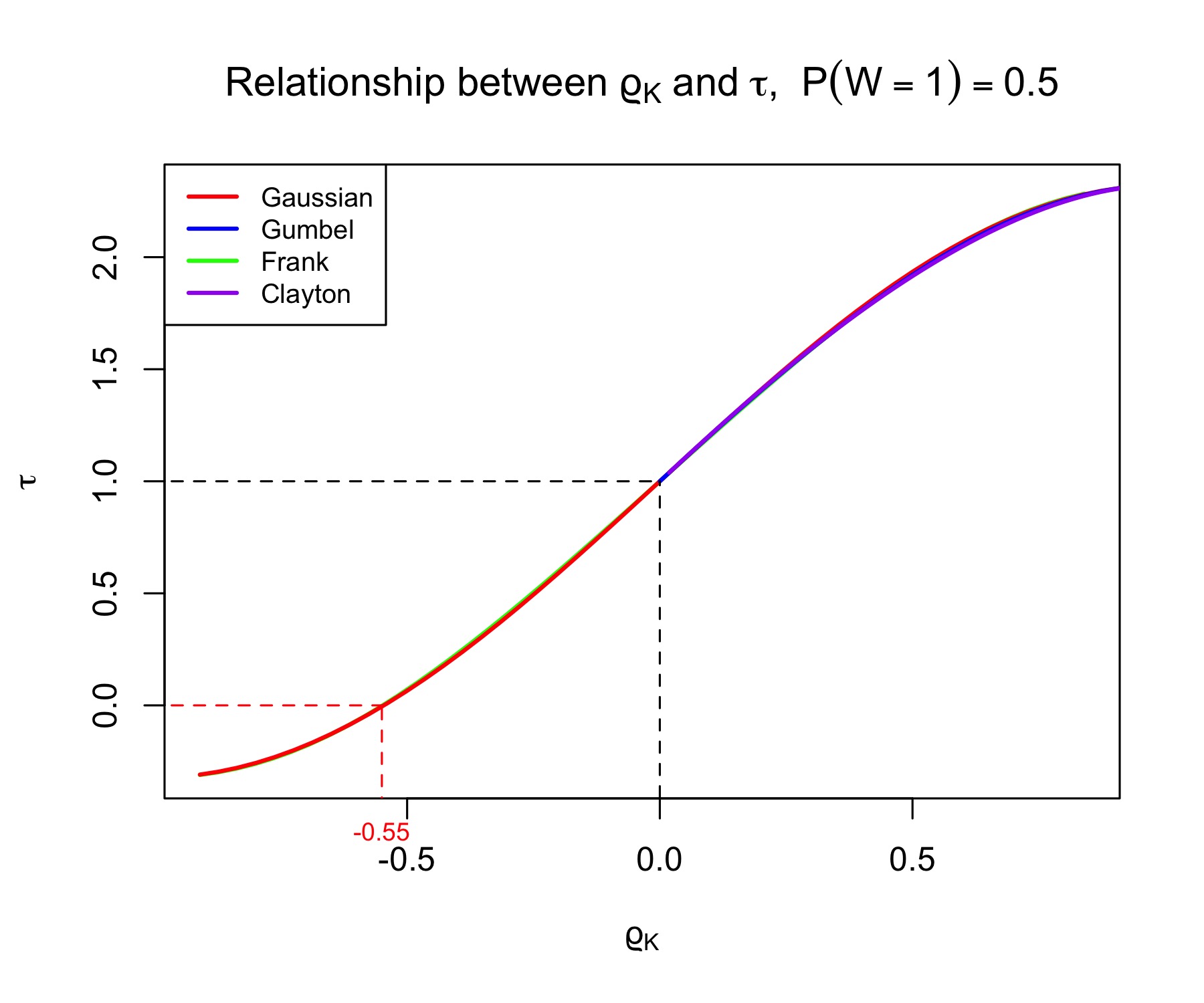}}
        \subfigure{
        \includegraphics[width=0.318\textwidth]{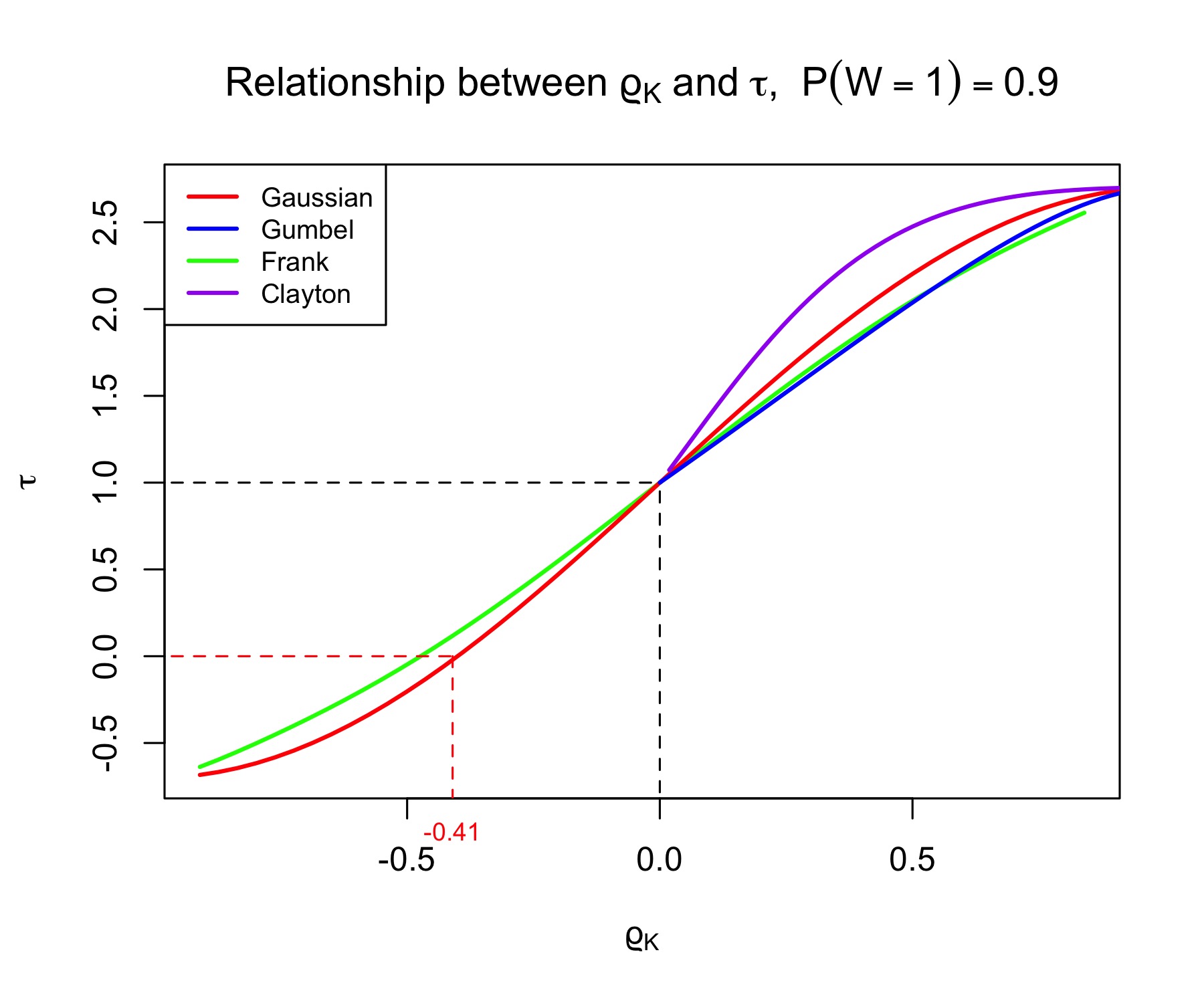}}
        \caption{Relationship between Kendall's tau $(\varrho_K)$ and $\tau$ given $\rho\in\{0.1,0.5,0.9\}$ for the Gaussian, Clayton, Gumbel, and Frank copulas.}
   \label{Figure: illustrative Archimedean conditional}
\end{figure}

\section{Worst-Case Bounds, Debiased Estimation, and Inference}
\label{sec:PartialId}

In practice, the copula $C_o$ is rarely known. \cref{sec:numerical} provides a numerical illustration of the application of \cref{theorem: known copula} to check sensitivity of $\tau_{\Pi}$ to the violation of the surrogacy assumption by letting the copula $C_o$ deviate from the independence copula. \cref{theorem: known copula} can also be used to establish sharp bounds on $\tau_{C_o}$ when $C_o$ is unknown but lies between two known copulas. 

 The following corollary follows immediately from \cref{theorem: known copula} and \cref{proposition:DSIP}.

\begin{corollary}[Identified Set]\label{Coro: identifiedSets}
      Suppose \cref{assumption: Random Sample}, \cref{assumption: Unconfounded}, and \cref{assumption:Comparability} hold. Furthermore, suppose the copula $C_o$ satisfies $C_L(\cdot, \cdot|s,x,E)\prec_c C_o(\cdot, \cdot|s,x,E)\prec_c  C_U(\cdot, \cdot|s,x,E)$ for almost all $(s,x)\in\mathcal{S}\otimes\mathcal{X}$, where  $C_L(\cdot, \cdot |s,x,E)$ and $C_U(\cdot, \cdot | s,x,E)$ are  two known copula functions.
Then $\tau_{C_o}$ is partially identified with the identified set $ [\tau_{C_L}, \tau_{C_U}]$. 
When the conditional distribution of $Y_i$ given $S_i,X_i,P_i=O$ is degenerate,  the identified set is singleton: $\tau_{C_L}=\tau_{C_U}=\tau_{\Pi}$.  
\end{corollary}

\cref{Coro: identifiedSets} implies that the worst case bounds on $\tau_{C_o}$ are $\tau_{C_-}$ and $ \tau_{C_+}$. As a result, $\tau_{C_-}$ can be interpreted as the smallest ATE when the surrogacy assumption fails. Conversely, $\tau_{C_+}$ is the largest ATE when the surrogacy assumption fails. Another implication is that the ATE under the surrogacy assumption in \cite{athey2025surrogate} is the least ATE among copulas dominating the independence copula in concordance order. It is also the greatest ATE among copulas dominated by the independence copula in concordance order.

\subsection{Comparison with Lemma 1 in \cite{athey2025surrogate}}

 We restate the expressions for 
$\tau_{C_-}, \tau_{C_+}$ in the following proposition which also shows that they are the same as  Lemma 1 (ii) in \cite{athey2025surrogate} for binary outcomes.

\begin{proposition}\label{prop:WCB}
          Suppose \cref{assumption: Random Sample}, \cref{assumption: Unconfounded}, and \cref{assumption:Comparability} hold. 
Then the identified set of $\tau_{C_o}$ is $[\tau_{C_-}, \tau_{C_+}]$, where $\tau_{C_-},\tau_{C_+}\in\mathbb{R}$ are given  by
\begin{align*}
     \tau_{C_-} 
& =\mathbb{E}\left[ \mu_{C_-, 1}(S_i,X_i,O)\frac{W_i }{\rho(X_i)} \mid P_i = E \right] 
  -  \mathbb{E}\left[  \mu_{C_-, 0}(S_i,X_i,O)\left(\frac{1-W_i }{1-\rho(X_i)}\right)\mid P_i = E \right],\\
    \tau_{C_+} 
& =\mathbb{E}\left[ \mu_{C_+, 1}(S_i,X_i,O)\frac{W_i }{\rho(X_i)} \mid P_i = E \right] 
  -  \mathbb{E}\left[  \mu_{C_+, 0}(S_i,X_i,O)\left(\frac{1-W_i }{1-\rho(X_i)}\right)\mid P_i = E \right], 
\end{align*}
in which
 \begin{align*}
 \mu_{C_-,0}(S_i,X_i,O)&=AVaR_{\rho(X_i,S_i)}(Y_i\mid S_i,X_i,P_i=O),\\
 \mu_{C_-,1}(S_i,X_i,O)&=-AVaR_{1-\rho(X_i,S_i)}(-Y_i\mid S_i,X_i,P_i=O),\\
\mu_{C_+,1}(S_i,X_i,O)&=AVaR_{1-\rho(X_i,S_i)}(Y_i\mid S_i,X_i,P_i=O), \mbox{ and }\\
 \mu_{C_+,0}(S_i,X_i,O)&=-AVaR_{\rho(X_i,S_i)}(-Y_i\mid S_i,X_i,P_i=O).
\end{align*}
When the outcome is binary, the identified set $[\tau_{C_-}, \tau_{C_+}]$ is the same as  Lemma 1 (ii) in \cite{athey2025surrogate}.
\end{proposition}

\cref{prop:WCB} implies that the ATE is partially identified regardless of the outcome type/range and for binary outcomes. The identified interval in \cref{prop:WCB} is the same as that in Lemma 1 (ii) in Section 5.2 of \cite{athey2025surrogate}.  
However, it differs from Lemma 1 (i) in Section 5.2 of \cite{athey2025surrogate} which states that ``If the outcome can take on values on the whole real line, then there is no value for the average treatment effect $\tau$ that can be ruled out.'' Their proof seems to have ignored the fact that  $F_Y(\cdot\mid s,x,E)$ is identified from $F_Y(\cdot\mid s,x,O)$ under \cref{assumption:Comparability}. In contrast, our proof makes use of the identified  $F_Y(\cdot\mid s,x,E)$ and 
the following expression for $\mu_{C_o,1}(S_i,X_i,O)$ in \cref{eq:SB}:
 \begin{align*}
\mu_{C_o, 1}(s,x,O)
=\frac{1}{\rho(s,x) }\int\int ywdC_o(F_Y(y|s,x, E),F_W(w|s,x, E)|s,x, E).
\end{align*}
Since both $F_Y(y|s,x, E)$ and $F_W(w|s,x, E)$ are point identified, Theorem 1 of \cite{cambanis_1976_inequalities} implies that $\mu_{C_o, 1}(s,x,O)$ is partially identified. 

\subsection{ Debiased Estimation}
\label{sec:estimation}
\cite{athey2025surrogate} and \cite{chen_ritzwoller_2023} develop debiased estimation of $\tau_{\Pi}$ for which WSIs reduce to the SI $\mu(S_i,X_i,O)=\mathbb{E}\left(Y_i\mid S_i,X_i,P_i=O\right)$. For the worst-case bounds, the WSIs are related to conditional AVaR of $Y_i$ instead of the conditional mean of $Y_i$ and also depend on the surrogacy score $\rho(S_i,X_i)$. As a result, it is more challenging to derive the orthogonal moment functions for the worst-case bounds $\tau_{C_+}$ and $\tau_{C_-}$. 

To proceed, we make use of the dual representations of the worst-case WSIs in terms of conditional means of the following newly defined functions:
\begin{align*}
 H_U(Y_i, s, \alpha) := s + \frac{1}{\alpha} \left(Y_i - s \right)_{+} \mbox{ and }
    H_L(Y_i, s, \alpha) := s - \frac{1}{\alpha} \left(Y_i - s \right)_{-}.
\end{align*}
Specifically, from the dual form of AVaR  (c.f., \cite{rockafellar2002ConditionalValueatriskGeneral} and \cite{acerbi2002es}), it follows that 
    \begin{align*}
 \mu_{C_-,0}(S_i,X_i,O)&=\mathbb{E}[H_U(Y_i, q_{C_-}(S_i,X_i,O), 1-\rho(S_i, X_i))\mid S_i,X_i,P_i=O],\\
 \mu_{C_-,1}(S_i,X_i,O)&=\mathbb{E}[H_L(Y_i, q_{C_-}(S_i,X_i,O), \rho(S_i, X_i))\mid S_i,X_i,P_i=O],\\
 \mu_{C_+,0}(S_i,X_i,O)&=\mathbb{E}[H_L(Y_i, q_{C_+}(S_i,X_i,O), 1-\rho(S_i, X_i))\mid S_i,X_i,P_i=O], \\
 \mu_{C_+,1}(S_i,X_i,O)&=\mathbb{E}[H_U(Y_i, q_{C_+}(S_i,X_i,O), \rho(S_i, X_i))\mid S_i,X_i,P_i=O],
\end{align*}
where 
\begin{align*}
q_{C_+}(S_i, X_i, O) &:= F_Y^{-1}(1-\rho(X_i, S_i)|S_i, X_i, O)\mbox{ and}\\
    q_{C_-}(S_i, X_i, O) &:= F_Y^{-1}(\rho(X_i, S_i)|S_i, X_i, O).
    \end{align*} 

Let $\varphi(x) := \mathbb{P}(P_i=E | X_i = x)$, $
    \varphi := \mathbb{P}(P_i = E)$, and  for $w = 0, 1$,
    \begin{align*}
        \bar{\mu}_{C_+, w} &:= \mathbb{E}[\mu_{C_{+}, w}(S_i, X_i, O)|W_i=w, X_i, P_i=E] \mbox{ and}\\
        \bar{\mu}_{C_-, w} & := \mathbb{E}[\mu_{C_{-}, w}(S_i, X_i, O)|W_i=w, X_i, P_i=E].
    \end{align*}
Finally, let $\eta$ denote the collection of all the nuisance functions, i.e., 
\begin{align*}
    \eta & = (\mu_{C_{-}, 0}, \mu_{C_{-}, 1}, \mu_{C_{+}, 0}, \mu_{C_{+}, 1}, 
            \bar{\mu}_{C_-, 0}, \bar{\mu}_{C_-, 1}, \bar{\mu}_{C_+, 0}, \bar{\mu}_{C_+, 1}, q_{C_+}, q_{C_-}, \rho(s, x), \rho(x), \varphi(x), \varphi(s, x), \varphi).
\end{align*}
Then $\tau_{C_+}$ satisfies the  moment condition: $ \mathbb{E}[m_{C_+}(Z_i, \tau_{C_+}, \eta)] = 0,$
where 
\begin{align*}
    & m_{C_+}(Z_i, \tau_{C_+}, \eta) \\
    & = \frac{\mathds{1}_{P_i = E}}{\varphi} \Bigg[\frac{W_i}{\rho(X_i)} (\mu_{C_+, 1}(S_i, X_i, O) - \bar{\mu}_{C_+, 1}(1, X_i)) - \frac{1-W_i}{1-\rho(X_i)} (\mu_{C_+, 0}(S_i, X_i, O) - \bar{\mu}_{C_+, 0}(0, X_i))  \Bigg] \\
    & \quad + \frac{\mathds{1}_{P_i = E}}{\varphi} (\bar{\mu}_{C_+, 1}(1, X_i) - \bar{\mu}_{C_+, 0}(0, X_i) - \tau_{C_+}) \\
    & \quad +
    \frac{\mathds{1}_{P_i = O}}{\varphi} \frac{\varphi(S_i, X_i)}{1-\varphi(S_i, X_i)} \Bigg[ \frac{\rho(S_i, X_i)}{\rho(X_i)} (H_U(Y_i, q_{C_+}(S_i, X_i, O), \rho(S_i, X_i)) -\mu_{C_+, 1}(S_i, X_i, O)) \\
    & \quad \quad \quad -
    \frac{1-\rho(S_i, X_i)}{1-\rho(X_i)} (H_L(Y_i, q_{C_+}(S_i, X_i, O), 1 - \rho(S_i, X_i)) - \mu_{C_+, 0}(S_i, X_i, O)) \Bigg] \\
    & \quad +
     \frac{\mathds{1}_{P_i = E}}{\varphi} \frac{1}{\rho(X_i)} \left[q_{C_+}(S_i, X_i, O) - \mu_{C_+, 1}(S_i, X_i, O)\right] \left(W_i - \rho(S_i, X_i)\right) \\
    & \quad +
    \frac{\mathds{1}_{P_i = E}}{\varphi} \frac{1}{1-\rho(X_i)} \left[q_{C_+}(S_i, X_i, O) - \mu_{C_+, 0}(S_i, X_i, O)\right] \left(W_i - \rho(S_i, X_i)\right).
\end{align*}
The first three terms in the orthogonal moment function $ m_{C_+}$ are analogous to those in Equation (4.4) of \cite{athey2025surrogate} and Theorem 3.1 of \cite{chen_ritzwoller_2023}, and the last two terms are new and correct for the effect of estimating $\rho(S_i, X_i)$ in the dual representations of the WSIs. 

Similarly,  $\tau_{C_-}$ satisfies: $ \mathbb{E}[m_{C_-}(Z_i, \tau_{C_-}, \eta)] = 0,$
where
\begin{align*}
    & m_{C_-}(Z_i, \tau_{C_-}, \eta) \\
    & = \frac{\mathds{1}_{P_i = E}}{\varphi} \Bigg[\frac{W_i}{\rho(X_i)} (\mu_{C_-, 1}(S_i, X_i, O) - \bar{\mu}_{C_-, 1}(1, X_i)) - \frac{1-W_i}{1-\rho(X_i)} (\mu_{C_-, 0}(S_i, X_i, O) - \bar{\mu}_{C_-, 0}(0, X_i))  \Bigg] \\
    & \quad + \frac{\mathds{1}_{P_i = E}}{\varphi} (\bar{\mu}_{C_-, 1}(1, X_i) - \bar{\mu}_{C_-, 0}(0, X_i) - \tau_{C_-}) \\
    & \quad +
    \frac{\mathds{1}_{P_i = O}}{\varphi} \frac{\varphi(S_i, X_i)}{1-\varphi(S_i, X_i)} \Bigg[ 
    \frac{\rho(S_i, X_i)}{\rho(X_i)} (H_L(Y_i, q_{C_-}(S_i, X_i, O), \rho(S_i, X_i)) - \mu_{C_-, 1}(S_i, X_i, O)) \\
    & \quad \quad \quad -
    \frac{1-\rho(S_i, X_i)}{1-\rho(X_i)} (H_U(Y_i, q_{C_-}(S_i, X_i, O), 1-\rho(S_i, X_i)) - \mu_{C_-, 0}(S_i, X_i, O)) \Bigg] \\
    & \quad +
     \frac{\mathds{1}_{P_i = E}}{\varphi} \frac{1}{\rho(X_i)} \left[q_{C_-}(S_i, X_i, O) - \mu_{C_-, 1}(S_i, X_i, O)\right] \left(W_i - \rho(S_i, X_i)\right) \\
    & \quad +
     \frac{\mathds{1}_{P_i = E}}{\varphi} \frac{1}{1-\rho(X_i)} \left[q_{C_-}(S_i, X_i, O) - \mu_{C_-, 0}(S_i, X_i, O)\right] \left(W_i - \rho(S_i, X_i)\right).
\end{align*}

Our debiased estimators $\widehat{\tau}_{C_-}$ and $\widehat{\tau}_{C_+}$ are defined as the solutions to $\frac{1}{n}\sum_{i=1}^n m_{C_-}(Z_i, \tau, \widehat{\eta})=0$ and $\frac{1}{n}\sum_{i=1}^n m_{C_+}(Z_i, \tau, \widehat{\eta})=0$, respectively, where $\widehat{\eta}$ is estimated by cross-fitting over $K$ even folds, which ensures  that the nuisance estimators remain independent of the samples to which they are applied \citep{chernozhukov2018doubleML}. 

Our orthogonal moment conditions in debiased estimation allow for greater flexibility in estimating nuisance parameters, enabling the use of both parametric regressions and machine learning methods. For example, the estimation of $\rho(x)$, $\rho(s, x)$, $\varphi(x)$, and $\varphi(s, x)$ readily accommodates methods such as Lasso, random forests, gradient boosting, and neural networks. In contrast, estimating the worst-case WSIs ($\mu_{C_{-}, 0}$, $\mu_{C_{-}, 1}$, $\mu_{C_{+}, 0}$, and $\mu_{C_{+}, 1}$) relies on previously obtained cross-fitted estimates of $\rho(X_i, S_i)$ and is more involved due to the computation of conditional AVaRs. We follow the two-stage, locally robust approach of \cite{olma2021} with our chosen model specifications: the first stage requires estimating a conditional quantile, which we do nonparametrically using quantile forests; the second stage involves fitting a linear sieve model to a generated outcome variable whose derivative with respect to the conditional quantile, evaluated at the truth, is zero. With estimates of the worst-case WSIs as pseudo-outcomes, $\bar{\mu}_{C-, 0}$, $\bar{\mu}_{C-, 1}$, $\bar{\mu}_{C+, 0}$, and $\bar{\mu}_{C+, 1}$ can then be estimated either parametrically or nonparametrically. Finally, $q_{C_+}$ and $q_{C_-}$ are again estimated using quantile forests with cross-fitted $\widehat{\rho}(X_i,S_i)$. More details can be found in Algorithm \ref{alg:debiased} of Appendix \ref{sec:numerical details} in the supplementary materials.

\subsection{Asymptotic Theory}

We adopt conditions similar to those of \cite{chen_ritzwoller_2023}, \cite{dorn2024dvds}, and \cite{semenova2025generalizedleebounds} to establish the asymptotic joint normality of $\widehat{\tau}_{C_-}$ and $\widehat{\tau}_{C_+}$.

\begin{assumption}[Regularity Conditions]
\label{assumption:regularity-wc}
(i) For a continuous outcome, we assume that it's distribution is absolutely continuous with bounded support and conditional density function $f(y|s, x, O)$ that is continuous with respect to y for each $s$ and $x$, and is uniformly bounded above and below by positive absolute constants; For a binary outcome, we assume that the random variable $\left[\mu(S_i, X_i, O) - \rho(S_i, X_i)\right]$ has bounded density.

(ii) There exists some absolute constant $t$ such that for each $w \in \{0, 1\}$ and $b \in \{C_-, C_+\}$, 
either (1) or (2-1, 2-2) holds with probability one:
\begin{gather*}
    \text{(1) } \mathbb{E}[(\bar{\mu}_{b, 1}(1, X_i) - \bar{\mu}_{b, 0}(0, X_i) - \tau_{b})^2|X_i, P_i = E] > t \text{ or }  \\
    \text{(2-1) } \mathrm{Var}\left(\frac{1}{\rho(X_i)}(Y_i - q_{C_+}(S_i, X_i, O))_{+} + \frac{1}{1-\rho(X_i)}(Y_i - q_{C_+}(S_i, X_i, O))_{-}  | S_i, X_i, O)\right) > t \text{ and } \\
     \text{(2-2) } \mathrm{Var}\left(\frac{1}{1-\rho(X_i)}(Y_i - q_{C_-}(S_i, X_i, O))_{+} + \frac{1}{\rho(X_i)}(Y_i - q_{C_-}(S_i, X_i, O)_{-}  | S_i, X_i, O)\right) > t.
    \end{gather*}
\end{assumption}
The assumption of bounded support of $Y$ and boundedness of the conditional density function in  \cref{assumption:regularity-wc} (i) are similar to \cite{semenova2025generalizedleebounds}. Assumption (1) or (2-1, 2-2) in \cref{assumption:regularity-wc} (ii) ensures that the asymptotic variances of $\tau_{C_+}$ and $\tau_{C_-}$ are positive.

\begin{assumption}[Realization Set] \label{assumption:realization-set-wc}
Let $q>2$ be a positive constant and $\epsilon$ be a constant such that $0 < \epsilon < 1/2$, $\omega = (\omega_{\rho}, \omega_{\varphi})$ with $\omega_{\rho} =(\rho(S_i, X_i), \rho(X_i)), \mbox{ }
\omega_{\varphi}=(\varphi(S_i, X_i), \varphi )$, 
and 
\begin{align*}
    \kappa & = (\mu_{C_+, 1}(S_i, X_i, O), \mu_{C_+, 0}(S_i, X_i, O), \mu_{C_-, 1}(S_i, X_i, O), \mu_{C_-, 0}(S_i, X_i, O), 
     \\
     & \quad \quad \bar{\mu}_{C_+, 1}(1, X_i), \bar{\mu}_{C_+, 0}(0, X_i), \bar{\mu}_{C_-, 1}(1, X_i), \bar{\mu}_{C_-, 0}(0, X_i),  q_{C_+}(S_i, X_i, O), q_{C_-}(S_i, X_i, O) ).
\end{align*}
For all probability measures $P$ satisfying Assumptions \ref{assumption: Random Sample}, \ref{assumption: Unconfounded}, and \ref{assumption:Comparability}, 
the following condition holds; for some sequences $\Delta_n \to 0$ and $\delta_n \to 0$ with $\delta_n \ge n^{-1/2}$ with probability $1 - \Delta_n$, the estimator of nuisance parameter belongs to the realization set $R_n$ which contains $\widetilde{\eta}$ such that 
\begin{gather}
    \| \widetilde{\eta} - \eta \|_{P, q} \le C, \mbox{ }
    \mathbb{P}(\epsilon \le \widetilde{\rho}(S_i, X_i) \le 1-\epsilon) = 1, \notag \\
    \mathbb{P}(\epsilon \le \widetilde{\varphi}(S_i, X_i) \le 1-\epsilon) = 1, \mbox{ }
    \|  \widetilde{\eta} - \eta \|_{P, 2} \le \delta_n, \notag \\
    \| \widetilde{\omega} - \omega \|_{P, 2} \times \| \widetilde{\kappa} - \kappa \|_{P, 2} \le \delta_n n^{-1/2}, \notag \\
    \| \widetilde{\omega}_{\rho} - \omega_{\rho} \|_{P, 2} \times \| \widetilde{\omega}_{\varphi} - \omega_{\varphi} \|_{P, 2} \le \delta_n n^{-1/2}, \label{eq:rate-wc-cross-term-rho-varphi}    
\end{gather}
and  for continuous outcomes, 
    \begin{align}
         \| \widetilde{q}_{C_+} - q_{C_+} \|_{P, 2}^2 \le \delta_n n^{-1/2} \label{eq:rate-quantile-L2-sq-wc} ;
    \end{align}
for binary outcomes, 
    \begin{align}
         \| \widetilde{\mu}(S_i, X_i, O) - \mu(S_i, X_i, O) \|_{\infty} \le \delta_n, 
\| \widetilde{\mu}(S_i, X_i, O) - \mu(S_i, X_i, 0) \|_{\infty}^2 \le \delta_n n^{-1/2}. \notag    
    \end{align}

\end{assumption}
Most of the conditions in \cref{assumption:realization-set-wc} are similar to those in \cite{chen_ritzwoller_2023}. Since our moment function has additional correction terms for $\rho$, we impose additional conditions on the rate of the cross-product in \cref{eq:rate-wc-cross-term-rho-varphi}. 
Condition \eqref{eq:rate-quantile-L2-sq-wc} and conditions for the binary case are similar to \cite{dorn2024dvds}.

\begin{theorem} \label{thm:CLT-for-wc}
Suppose Assumptions \ref{assumption: Random Sample}, \ref{assumption: Unconfounded}, \ref{assumption:Comparability}, \ref{assumption:regularity-wc}, and  \ref{assumption:realization-set-wc} hold. Then 
    \begin{align*}
       \sqrt{n}
       \begin{pmatrix}
           \widehat{\tau}_{C_+} - \tau_{C_+} \\
           \widehat{\tau}_{C_-} - \tau_{C_-} 
       \end{pmatrix}
       \to 
       N
       \left(\begin{pmatrix}
           0 \\ 0 
       \end{pmatrix}, 
       \begin{pmatrix}
           \sigma_{C_+}^2 & \rho \sigma_{C_+} \sigma_{C_-} \\
           \rho \sigma_{C_+} \sigma_{C_-} & \sigma_{C_-}^2 
       \end{pmatrix}
       \right),
    \end{align*}
    where $\sigma_{C_+}^2  = \mathbb{E}[m(Z_i, \tau_{C_+}, \eta_{0})^2] > 0$, $\sigma_{C_-}^2 = \mathbb{E}[m(Z_i, \tau_{C_-}, \eta_{0})^2] > 0$, and $-1 \le \rho \le 1$.
\end{theorem}

Our proof strategy builds on \cite{chen_ritzwoller_2023} and \cite{dorn2024dvds}. Departing from \cite{chen_ritzwoller_2023} which checks the orthogonality condition in Assumption 3.1 and the statistical rate for second-order derivative of moment condition $r_N'$ in Assumption 3.2 (c) in \cite{chernozhukov2018doubleML}, similar to \cite{dorn2024dvds},
 we directly verified  the following high-level condition in the proof of Theorem 3.1 of \cite{chernozhukov2018doubleML}:
    \begin{align*}
        \sqrt{n} \| \mathbb{E}[m(W_i, \tau_{C_+}, \hat{\eta})|\mathcal{I}_{-k}] -  \mathbb{E}[m(W_i, \tau_{C_+}, \eta)]\| = o_p(1).
    \end{align*}
Note that this high-level condition is satisfied when 
both the orthogonality condition in Assumption 3.1 in \cite{chernozhukov2018doubleML} and the statistical rate for second-order derivative of moment condition $r_N'$ in Assumption 3.2 (c) in \cite{chernozhukov2018doubleML} hold. 

Let $V$ denote the asymptotic variance-covariance matrix in Theorem \ref{thm:CLT-for-wc}. That is,
\begin{align*}
    V = \mathbb{E}[(m_{C_+}(Z_i, \tau_{C_+}, \eta_0), m_{C_-}(Z_i, \tau_{C_-}, \eta_0)) [(m_{C_+}(Z_i, \tau_{C_+}, \eta_0), m_{C_-}(Z_i, \tau_{C_-}, \eta_0))]^{\top}].
\end{align*}
We provide a consistent estimator of $V$ by following Theorem 3.2 in \cite{chernozhukov2018doubleML}. 
\begin{theorem}
     Suppose Assumptions \ref{assumption: Random Sample}, \ref{assumption: Unconfounded}, \ref{assumption:Comparability}, \ref{assumption:regularity-wc}, and  \ref{assumption:realization-set-wc} hold. In addition, we assume that $\delta_n \ge n^{-[((1-2/q) \wedge 1/2]}$ for all $n \ge 1$. 
     Then, $V$ can be consistently estimated by
    \begin{align*}
        \frac{1}{K} \sum_{k=1}^K \mathbb{E}_{r, k} [(m_{C_+}(Z_i, \hat{\tau}_{C_+}, \hat{\eta}_{k}), m_{C_-}(Z_i, \hat{\tau}_{C_-}, \hat{\eta}_{k})) (m_{C_+}(Z_i, \hat{\tau}_{C_+}, \hat{\eta}_{k}), m_{C_-}(Z_i, \hat{\tau}_{C_-}, \hat{\eta}_{k}))^{\top}],
    \end{align*}
    where $K$ is the number of folds in K-fold cross-fitting, and $r := [n/K]$ is the number of observations in each fold, and $E_{r, k}$ is operator for sample expectation from empirical data in $k$-the fold, i.e., $E_{r, k}[g(X_i)] = r^{-1} \sum_{i \in \mathcal{F}_k} g(X_i)$ where $\mathcal{F}_k$ is the $k$-th fold of indices $\{1, \dotsc, n\}$.
\end{theorem}

Wald inference on each bound is straightfoward and inference on the true ATE can be carried out by applying the misspecification-adaptive confidence interval in \cite{stoye2020SimpleShortNeverEmpty} which allows for the covariance matrix to be degenerate.

\section{Debiased Estimation and Inference---General Copula}

\label{sec:est-copula}

The debiased estimators of $\tau_{C_+}$ and $\tau_{C_-}$ constructed in \cref{sec:estimation} rely critically on the dual representations of the WSIs (associated with $C_+$ and $C_-$) obtained from the existing dual representation of AVaR. 

For a general copula $C_o$, we establish dual representation for $r$ introduced in \cref{remark:DRM} when $\sigma$ is of bounded variation in \cref{lemma:dual-SpectralRiskMeasure} below. Our proof builds on the proof of 
\cite{pichler2015} and Section 2.4.2. in \cite{pflug2007ModelingMeasuringManagingRIsk}  for non-decreasing function $\sigma(\cdot)$ and the dual representation for AVaR. 

\begin{lemma}
\label{lemma:dual-SpectralRiskMeasure}
Assume that $Y$ is bounded and $\sigma(u)$ is of bounded variation on [0, 1]. Then,
\begin{align}
        r 
        & = \sigma(0) \mathbb{E}[Y] + \int_0^1 \left((1-u) F_Y^{-1}(u) + \mathbb{E}[[Y-F_Y^{-1}(u)]_+]\right) d \sigma(u) \label{eq:dual-SpectralRiskMeasure-sig0-ver2} \\
        & = \sigma(1) \mathbb{E}[Y] - \int_0^1 \left(u F_Y^{-1}(u) - \mathbb{E}[[Y-F_Y^{-1}(u)]_-]\right) d \sigma(u) \label{eq:dual-SpectralRiskMeasure-sig1-ver2}.
    \end{align}
\end{lemma}
 \cref{lemma:dual-SpectralRiskMeasure} extends the dual representation for AVaR to $r$ with a general weight function $\sigma$. However, in contrast to AVaR for which     $F_Y^{-1}(u) \in \argmin_G \{(1-u) G(u) + \mathbb{E}[[Y-G(u)]_+\}$, $F_Y^{-1}$ may not be an argmin of the following minimization problem:
    \begin{align}
        \min_{G}\int_0^1 \left((1-u) G(u) + \mathbb{E}[[Y-G(u)]_+]\right) d \sigma(u). \label{eq:min-dual-copula}
    \end{align}
This is because the sign of the second-order derivative may not be positive in the minimization problem \eqref{eq:min-dual-copula}.

\subsection{Orthogonal Moment Function and Debiased Estimator}

We construct a debiased estimator of $\tau_{C_o}$ from the dual representation in \cref{lemma:dual-SpectralRiskMeasure} under the following assumption.

\begin{assumption}\label{assumption:boundedvariation}
    The outcome variable is continuous variable, and it has a bounded support. In addition, $C_o(\alpha|\cdot)$ is of bounded variation.
\end{assumption}
\cref{lemma:dual-SpectralRiskMeasure} implies that under \cref{assumption:boundedvariation}, 
\begin{align} \label{eq:dual}
    \mu_{C_o, 1}(S_i, X_i, O)
    &= \mathbb{E}[h_{C_o, 1}(Y_i; F_Y^{-1}(\cdot | S_i, X_i, O), 1-\rho(S_i, X_i))|S_i, X_i, P_i=O] \mbox{ and} \nonumber\\
    \mu_{C_o, 0}(S_i, X_i, O)
    & = \mathbb{E}[h_{C_o, 0}(Y_i; F_Y^{-1}(\cdot | S_i, X_i, O), 1-\rho(S_i, X_i))|S_i, X_i, P_i=O],
\end{align}
where 
\begin{align*}
    & h_{C_o, 1}(Y_i; F_Y^{-1}(\cdot | S_i, X_i, O), 1-\rho(S_i, X_i)) \\
    & = 
    \sigma_{C_o, 1}(0; 1- \rho(S_i, X_i)) Y_i\\
     &+ \int_0^1 \left((1-u) F_Y^{-1}(u|S_i, X_i, O) + \left[Y_i - F_Y^{-1}(u|S_i, X_i, O)\right]_{+}\right) d \sigma_{C_o, 1}(u; 1 - \rho(S_i, X_i)) \mbox{ and}\\
    & h_{C_o, 0}(Y_i; F_Y^{-1}(\cdot | S_i, X_i, O), 1-\rho(S_i, X_i)) \\
    & = 
    \sigma_{C_o, 0}(1; 1 - \rho(S_i, X_i)) Y_i\\
     &- \int_0^1 \left(u F_Y^{-1}(u|S_i, X_i, O) - \left[Y - F_Y^{-1}(u|S_i, X_i, O)\right]_{-}\right) d \sigma_{C_o, 0}(u; 1 - \rho(S_i, X_i)).
\end{align*}

Similar to orthogonal moment functions in the worst case bounds in \Cref{sec:estimation},  we construct the following orthogonal moment function for a general $C_o$ using the dual representation of the WSIs in \cref{eq:dual}:
\begin{align}
    & m_{C_o}(Z_i, \tau, \eta) \\
         & = \frac{\mathds{1}_{P_i = E}}{\varphi} \Bigg[\frac{W_i}{\rho(X_i)} (\mu_{C_o, 1}(S_i, X_i, O) - \bar{\mu}_{C_o, 1}(1, X_i))  - \frac{1-W_i}{1-\rho(X_i)} (\mu_{C_o, 0}(S_i, X_i, O) - \bar{\mu}_{C_o, 0}(0, X_i))  \Bigg] \nonumber \\
    & \quad + \frac{\mathds{1}_{P_i = E}}{\varphi} (\bar{\mu}_{C_o, 1}(1, X_i) - \bar{\mu}_{C_o, 0}(0, X_i) - \tau) \nonumber \\
    & \quad +
     \frac{\mathds{1}_{P_i = O}}{\varphi} \frac{\varphi(S_i, X_i)}{1-\varphi(S_i, X_i)} \Bigg[\frac{\rho(S_i, X_i)}{\rho(X_i)} (h_{C_o, 1}(Y_i, F_{Y}^{-1}(\cdot|S_i, X_i, O), 1 - \rho(S_i, X_i)) -\mu_{C_o, 1}(S_i, X_i, O)) \nonumber \\
    & \quad \quad \quad -
    \frac{1-\rho(S_i, X_i)}{1-\rho(X_i)} (h_{C_o, 0}(Y_i, F_Y^{-1}(\cdot|S_i, X_i, O), 1 - \rho(S_i, X_i)) - \mu_{C_o, 0}(S_i, X_i, O)) \Bigg] \nonumber \\
     & \quad +
   \frac{\mathds{1}_{P_i = E}}{\varphi} \frac{1}{\rho(X_i)}  (d_{C_o}(S_i, X_i) - \mu_{C_o, 1}(S_i, X_i, O)) \left(W_i - \rho(S_i, X_i)\right) \nonumber \\
    & \quad +
    \frac{\mathds{1}_{P_i = E}}{\varphi} \frac{1}{1-\rho(X_i)}  (d_{C_o}(S_i, X_i) - \mu_{C_o, 0}(S_i, X_i, O))
    \left(W_i - \rho(S_i, X_i)\right),
    \label{eq:orthogonal moment}
\end{align}
where 
\begin{align*}
    \eta & = (F_Y^{-1}(\cdot|S_i, X_i, O), \rho(S_i, X_i), \rho(X_i), \varphi, \varphi(S_i, X_i), \\
    & \mu_{C_o, 1}(S_i, X_i, O),  \bar{\mu}_{C_o, 1}(1, X_i),  \mu_{C_o, 0}(S_i, X_i, O), \bar{\mu}_{C_o, 0}(0, X_i), d_{C_o}).
\end{align*}
in which
\begin{align*}
    \bar{\mu}_{C_o, w}(w, X_i)& = \mathbb{E}[\mu_{C_o, w}(S_i, X_i, O)|W_i = w, X_i, P_i=E] \mbox{ and}\\
    d_{C_o}(S_i, X_i)& = \int_0^1 F_Y^{-1}(u|S_i, X_i, O) c_o(1-\rho(S_i, X_i)|u) du.
\end{align*}

\subsection{Asymptotic Theory}

\begin{assumption}[Regularity Conditions] \label{assumption:regularity-copula}
(i) There exists some absolute constant $t$ such that with probability 1, either (1) or (2) holds: 
{\footnotesize
\begin{gather*}
    \textrm{(1) } \mathbb{E}[(\bar{\mu}_{C_o, 1}(1, X_i) - \bar{\mu}_{C_o, 0}(0, X_i) - \tau_{C_o})^2|X_i, P_i = E] > t; \\
    \textrm{(2) } \mathrm{Var}\left(\frac
    {\rho(S_i, X_i)}{\rho(X_i)} h_{C_o, 1}(Y_i, F_Y^{-1}, 1-\rho(S_i, X_i)) - \frac
    {1-\rho(S_i, X_i)}{1-\rho(X_i)}  h_{C_o, 0}(Y_i, F_Y^{-1}, 1-\rho(S_i, X_i)\Big|S_i, X_i, O) \right) > t. 
\end{gather*}
}
(ii) $C_o(\alpha|u)$ is continuously twice differentiable with respect to $u, \alpha$, and 
\begin{gather*}
    \sup_{\alpha \in [\epsilon, 1-\epsilon], u \in [0, 1]}\left|\frac{\partial C_o(\alpha|u)}{\partial u}\right|, 
    \sup_{\alpha \in [\epsilon, 1-\epsilon], u \in [0, 1]}\left|\frac{\partial c_o(\alpha|u)}{\partial u}\right|, 
    \sup_{\alpha \in [\epsilon, 1-\epsilon], u \in [0, 1]}\left|\frac{\partial^2 c_o(\alpha|u)}{\partial \alpha \partial u}\right|, 
\end{gather*}
$\sup_{\alpha \in [\epsilon, 1-\epsilon], u \in [0, 1]} |c_o(\alpha|u)|$, and $\sup_{\alpha \in [\epsilon, 1-\epsilon], u \in [0, 1]}\left|\frac{\partial c_o(\alpha|u)}{\partial \alpha}\right|$
are all bounded from above by absolute positive constant for some $\epsilon > 0$. 
\end{assumption}
The conditions in \Cref{assumption:regularity-copula} (i) are identical to those in \Cref{assumption:regularity-wc} (ii): Condition (2) in \Cref{assumption:regularity-copula} (ii) is reduced to Condition (2) in \Cref{assumption:regularity-wc} (ii) when $C_o$ is $C_+$ or $C_-$. 
The conditions in \Cref{assumption:regularity-copula} (iii) imply that $C_o(\alpha|u)$ and $c_o(\alpha|u)$ are absolute continuous with respect to $u$ so \Cref{lemma:dual-SpectralRiskMeasure} is applicable. 

\begin{assumption}[Realization Set] \label{assumption:realization-set-copula}
Let $q>2$ be a constant and $\epsilon$ be a constant such that $0 < \epsilon < 1/2$, $\omega = (\omega_{\rho}, \omega_{\varphi})$ in which $\omega_{\rho} =(\rho(S_i, X_i), \rho(X_i)), \omega_{\varphi}=(\varphi(S_i, X_i), \varphi )$, 
and
\begin{align*}
    \kappa = (\mu_{C_o, 1}, \bar{\mu}_{C_o, 1}, \mu_{C_o, 0}, \bar{\mu}_{C_o, 0}, d_{C_o}(S_i, X_i, O), F_Y^{-1}(\cdot|S_i, X_i, O)).
\end{align*}
and $\eta = (\omega, \kappa)$. 
For all $P$ satisfying Assumptions \ref{assumption: Random Sample}, \ref{assumption: Unconfounded}, and \ref{assumption:Comparability}, 
the following condition holds: for some sequences $\Delta_n \to 0$ and $\delta_n \to 0$ with $\delta_n \ge n^{-1/2}$ with probability $1 - \Delta_n$, the estimator of nuisance parameter belongs to the realization set $R_n$ which contains $\widetilde{\eta}$ such that 
\begin{gather*}
    \| \widetilde{\eta} - \eta \|_{P, q} \le C, 
    \mathbb{P}(\epsilon \le \widetilde{\rho}(S_i, X_i) \le 1-\epsilon) = 1,
    \mathbb{P}(\epsilon \le \widetilde{\varphi}(S_i, X_i) \le 1-\epsilon) = 1, \\
    \| \widetilde{\omega}_{\rho} - \omega_{\rho} \|_{P, 2} \times \| \widetilde{\omega}_{\varphi} - \omega_{\varphi} \|_{P, 2} \le \delta_n n^{-1/2}, 
    \| \hat{\omega} - \omega \|_{P, 2} \times \| \hat{\kappa} - \kappa \|_{P, 2} \le \delta_n n^{-1/2}, \\
 \|\widetilde{\rho}(S_i, X_i) - \rho(S_i, X_i)\|_{P, 2}^2 \le \delta_n n^{-1/2}, \\
 \|\widetilde{F}_Y^{-1}(\cdot|S_i, X_i, O) - F_Y^{-1}(\cdot|S_i, X_i, O) \|_{P, 2}^2 \le \delta_n n^{-1/2}.
\end{gather*}
Here, with abuse of notation, we denote the $L_q$ norm of a function of the form of $G(u|S_i, X_i)$ where $u \in [0, 1]$ by $\| G(\cdot|S_i, X_i) \|_{P, q} = \left(\mathbb{E}[\int_0^1 G^q(u|S_i, X_i) du]\right)^{1/q}$.
\end{assumption}

We verify conditions related to Theorem 3.1 in \cite{chernozhukov2018doubleML}. In particular, we verify \cref{lemma:DML-asym-normality} in the appendix.

\begin{theorem}
\label{thm:CLT-for-copula} 
Under Assumptions \ref{assumption: Random Sample}, \ref{assumption: Unconfounded}, \ref{assumption:Comparability}, \ref{assumption:knowncopula}, \ref{assumption:boundedvariation}, \Cref{assumption:regularity-wc} (i), \ref{assumption:regularity-copula}, and  \ref{assumption:realization-set-copula},  
$\sqrt{n} (\widehat{\tau}_{C_o} - \tau_{C_o}) \to  N(0, \sigma_{C_o}^2)$,
    where $\sigma_{C_o}^2  = \mathbb{E}[m(Z_i, \tau_{C_0}, \eta)^2]$.
\end{theorem}

\begin{remark}
(i) When $C_o$ is known, Wald inference can be constructed from \cref{thm:CLT-for-copula}.

(ii)
      Suppose $C_o$ is unknown and satisfies $C_L(\cdot, \cdot|s,x,E)\prec_c C_o(\cdot, \cdot|s,x,E)\prec_c  C_U(\cdot, \cdot|s,x,E)$ for almost all $(s,x)\in\mathcal{S}\otimes\mathcal{X}$, where  $C_L(\cdot, \cdot |s,x,E)$ and $C_U(\cdot, \cdot | s,x,E)$ are  two known copula functions.
Then \cref{Coro: identifiedSets} implies that $\tau_{C_o}$ is partially identified with the identified set $ [\tau_{C_L}, \tau_{C_U}]$. \cref{thm:CLT-for-copula} can be extended to the joint asymptotic normality of $(\widehat{\tau}_{C_L},\widehat{\tau}_{C_U})$ and inference for $\tau_{C_o}$ can be done in the same way as in Section \ref{sec:PartialId}.
\end{remark}
\section{Empirical Application}
\label{sec:empirical}

In the same spirit as \cref{sec:numerical}, this section presents results from a sensitivity analysis using the household poverty alleviation dataset used in \cite{banerjee2015MultifacetedProgramCauses}. In this data set, the treatment program allocated productive assets to randomly selected households. The baseline variables are welfare-related measurements taken before treatment. The short-term outcomes consist of the same set of measurements taken two years after treatment, while the long-term outcome is one of these measurements recorded three years after treatment. We take the data from Pakistan, which has 446 treated units and 408 control units. For illustrative purposes, consider $X$ and $S$ to include the following five welfare indicators: per capita consumption, the food security index, the household asset index, the total amount borrowed, and agricultural income. The long-term outcome, $Y$, represents the household asset value in dollars three years post-treatment. 

The poverty alleviation data is both experimental and observational in that $W_i$ and $Y_i$ are observed for all households in the sample. To illustrate our method, we randomly and evenly split the Pakistan data into two parts, removing $Y_i$ from the experimental sample and $W_i$ from the observational sample. $\tau_{C_{\vartheta}}$ is estimated given $C_{\vartheta}$, where $C_{\vartheta}$ is a Frank copula whose dependence parameter $\vartheta$ is calibrated to match values of Kendall’s tau in the set $\{-0.9, -0.75, -0.5, -0.25, -0.1, 0, 0.1, 0.25, 0.5, 0.75, 0.9\}$. To estimate the nuisance components in the orthogonal moment function (\ref{eq:orthogonal moment}), we employ quantile forests (\texttt{quantile\_forest} from the \texttt{grf} package) to estimate the conditional quantile function $F_{Y}^{-1}(\cdot \mid S_i, X_i, O)$; use Lasso regressions (\texttt{cv.glmnet}) to estimate the conditional expectation of the WSI $\bar{\mu}_{C_o, w}(w, X_i)$; and apply logistic Lasso regressions to estimate the propensity and surrogacy scores $\rho(X_i)$, $\rho(S_i, X_i)$, as well as the selection probability $\varphi(S_i, X_i)$. The nuisance functions are cross-fitted with three data folds. Figure \ref{Figure: empirical 11 frank} shows how $\hat{\tau}_{C_\vartheta}$ varies with Kendall's tau. The estimated worst-case bounds are [$-\$131.8$, $\$148.8$], with 95\% confidence intervals of [$-\$158.9$, $-\$104.8$] for the lower bound and [$\$124.3$, $\$173.4$] for the upper bound.

\begin{figure}[H]
\centering
   \includegraphics[width=0.8\textwidth]{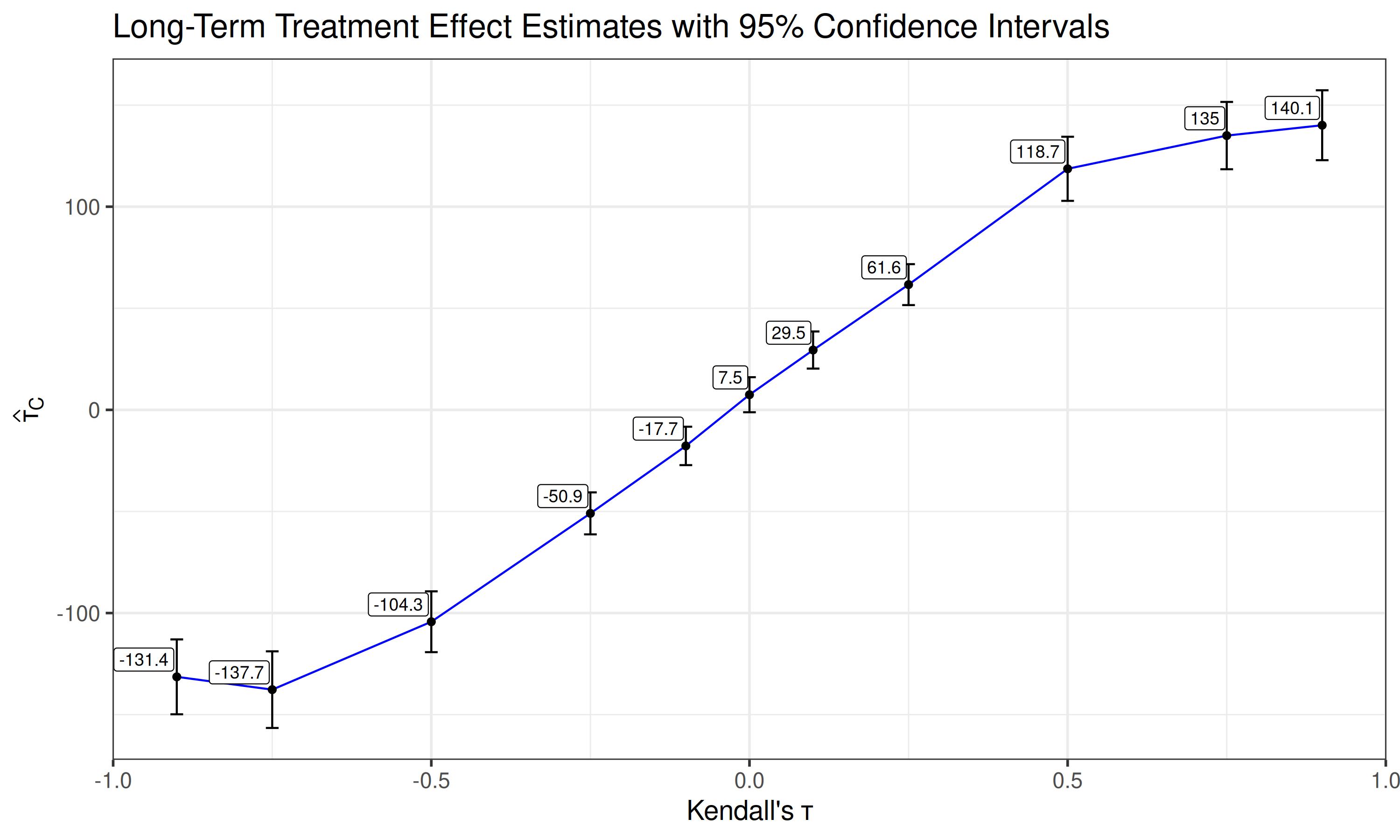}
   \caption{Relationship between $\hat{\tau}_{C_\vartheta}$ and Kendall's tau using the Frank copula for Pakistan.}
   \label{Figure: empirical 11 frank}
\end{figure}

We then repeat the analysis using the Plackett copula. The results, reported in Figure \ref{Figure: empirical 11 plackett}, indicate that the estimated treatment effects are largely insensitive to the choice of copula family: for a given Kendall's tau, the point estimates under the Plackett copula closely align with those obtained using the Frank copula. This robustness is consistent with the structure of the empirical application. In our data, both the estimated propensity score $\widehat\rho(X_i)$ and the estimated surrogacy score $\widehat\rho(S_i,X_i)$ are tightly centered around $0.5$ (with means $\approx0.54$), implying that most observations enter the conditional copula $C_o(1-\widehat\rho(s,x)\mid u)$ in regions far from the tails where copula families differ most sharply in their dependence behavior. Hence, the copula family itself has little influence on $\widehat{\tau}_{C_\vartheta}$. What matters most for the sensitivity analysis is the overall dependence level captured by Kendall's tau. The wide identified interval underscores the need for caution when interpreting the results under the surrogacy assumption. Even moderate departures from the assumed dependence structure can lead to substantial variation in the long-term treatment effect.

\begin{figure}[H]
\centering
   \includegraphics[width=0.8\textwidth]{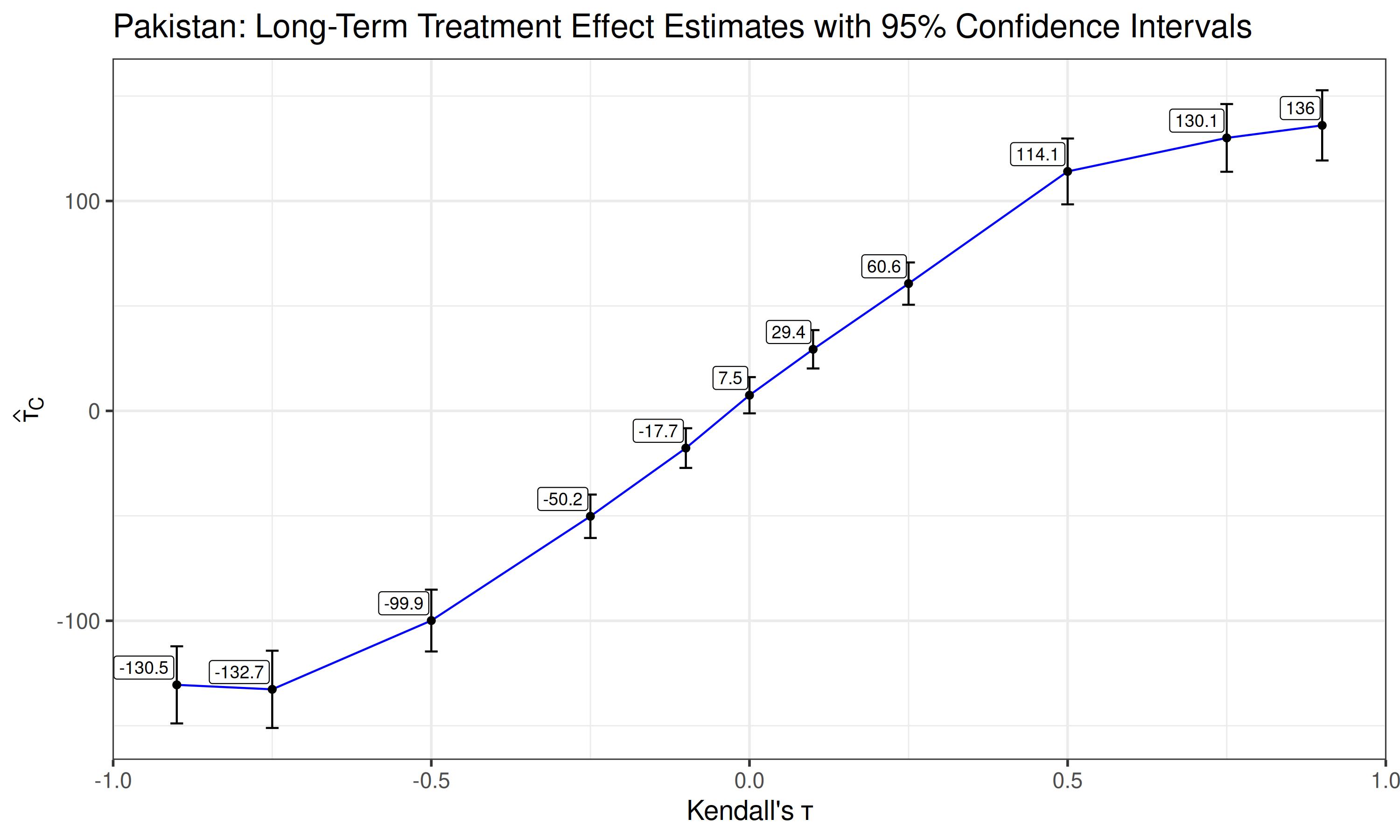}
   \caption{Relationship between $\hat{\tau}_{C_\vartheta}$ and Kendall's tau using the Plackett copula for Pakistan.}
   \label{Figure: empirical 11 plackett}
\end{figure}

To further examine the behavior of the welfare estimate near the surrogacy benchmark, Figure \ref{Figure: pak small} takes a closer look at small positive values of Kendall's tau near zero. Specifically, we conduct a local sensitivity analysis over a fine grid of small positive dependence levels, re-estimating the long-term treatment effect at each value. This zoomed-in analysis allows us to pinpoint the minimum degree of dependence at which the welfare estimate becomes statistically distinguishable from zero. For Pakistan, once Kendall's tau exceeds approximately 0.032, the confidence intervals no longer include zero, and the estimated effect remains significant thereafter. We interpret this Kendall's tau as a practical breakpoint, capturing the minimum strength of conditional dependence between treatment and outcome required for the welfare effect to turn significant.

\begin{figure}[H]
\centering
   \includegraphics[width=0.825\textwidth]{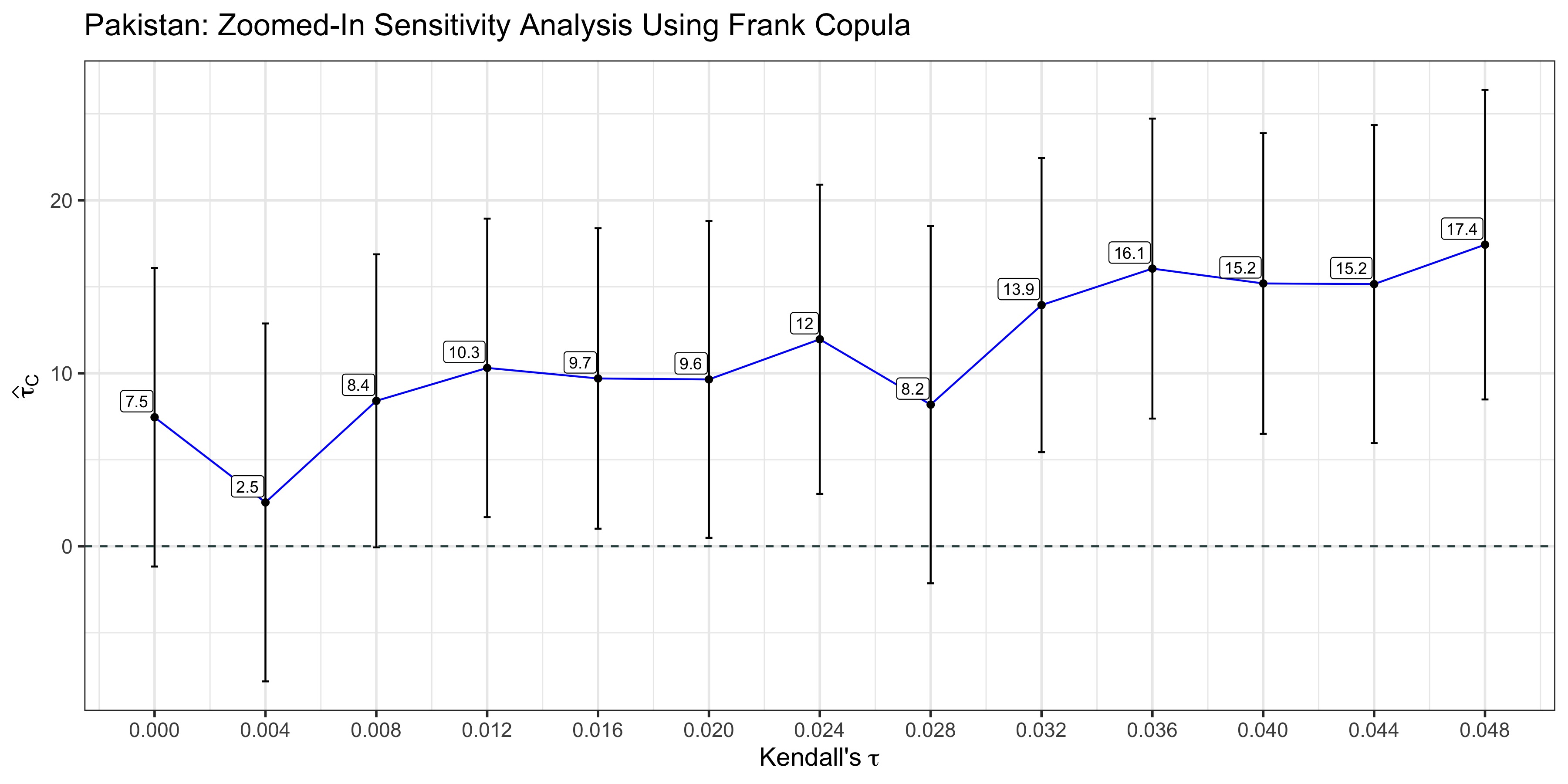}
   \caption{Local sensitivity analysis near the surrogacy benchmark (Kendall's tau $=0$).}
   \label{Figure: pak small}
\end{figure}

\section{Concluding Remarks}
\label{sec:conclusion}
In this paper, we have extended the SI approach for identifying and estimating long-term treatment effects in \citet{athey2025surrogate} from full mediation to partial mediation, substantially broadening the scope of application of the SI approach. Specifically, we develop two methodologies based on our identification result for a known copula: sensitivity analysis and partial identification analysis. The usefulness of both is illustrated via synthetic and real data. Complementing \citet{athey2025surrogate}, we determine the sign of the surrogacy bias for stochastically monotone copulas and establish the worst-case bounds on the true ATE regardless of the type of the primary outcome. Our partial identification result applies to any copula bounds. When applied to copulas that dominate the independence copula in concordance order, the lower bound is the ATE under the surrogacy assumption in \citet{athey2025surrogate}. This gives an alternative interpretation of ATE under the surrogacy assumption as the minimum ATE among all copulas that dominate the independence copula and thus are robust to such deviations from the surrogacy assumption. 

Several extensions are worthwhile and are currently under investigation. First, in a companion article, the authors develop a sensitivity analysis to  the comparability assumption.
Second, the worst-case bounds are often wide, suggesting caution in making the surrogacy assumption. In addition to exploiting prior knowledge on the range of copulas to shrink the identified set, in specific applications, side information such as  exclusion restrictions, monotone IV, monotone treatment response may be available. It is worthwhile exploring the possibility of tightening the worst-case bounds by exploring such information. Third, extensions to multi-valued treatments and 
continuous treatments would broaden the applicability of the SI approach further. Finally, on the technical side, 
it would be worthwhile exploring the possibility of extending the inference results to sub-copulas. 
 
\bibliographystyle{plainnat}
\bibliography{surrogacy}

\appendix

\section{Proofs in \cref{sec:iden-copula}}
\label{sec:proofs-iden-copula}

\subsection{Proof of \cref{theorem: known copula}}

\begin{proof} It follows from the proof of \cref{lemma:IndeC} that 
\begin{align*}
     \tau&=\mathbb{E}\left[Y_i(1)-Y_i(0)\mid P_i=E\right]\\
        &=
        \mathbb{E}\left[ \mathbb{E}\left[Y_iW_i|S_i=s,X_i=x,P_i=E \right] \left(\frac{1}{\rho(X_i)} + \frac{1}{1-\rho(X_i)} \right)  \mid P_i = E \right] \\
        & \quad -  \mathbb{E}\left[ \frac{1}{1-\rho(X_i)} \mathbb{E}\left[Y_i|S_i=s,X_i=x,P_i=E \right]  \mid P_i = E \right]\\
         &=
        \mathbb{E}\left[ \rho(S_i,X_i)\mathbb{E}\left[Y_i|W_i=1,S_i=s,X_i=x,P_i=E \right] \left(\frac{1}{\rho(X_i)} + \frac{1}{1-\rho(X_i)} \right)  \mid P_i = E \right] \\
        & \quad -  \mathbb{E}\left[ \frac{1}{1-\rho(X_i)} \mathbb{E}\left[Y_i|S_i=s,X_i=x,P_i=E \right]  \mid P_i = E \right].
\end{align*}
As a result of \cref{proposition:DSIPI}, we get 
\begin{align*}
     \tau_{C_o}
        &=
        \mathbb{E}\left[ \frac{\rho(X_i,S_i)}{\rho(X_i)[1-\rho(X_i)]}   \mu_{C_o, 1}(Y\mid S_i,X_i)  \mid P_i = E \right] -  \mathbb{E}\left[ \frac{1}{1-\rho(X_i)} \mu(S_i,X_i,O) \mid P_i = E \right]\\
          &=\mathbb{E}\left[ \mu_{C_o, 1}(Y\mid S_i,X_i)\frac{W_i }{\rho(X_i)} \mid P_i = E \right] 
  -  \mathbb{E}\left[  \mu_{C_o, 0}(Y\mid S_i,X_i)\left(\frac{1-W_i }{1-\rho(X_i)}\right)\mid P_i = E \right],
\end{align*}
where the second equality follows from the definition of $\mu_{C_o, w}$.

In conclusion, $\tau_{C_o}$ is identified from the sample information, because 
    (i) $ F_Y^{-1}(u|S_i,X_i, P_i=O)$ is identified from the observational data and 
    (ii) $\rho(X_i,S_i)$ is identified from the experimental data.
\end{proof}

\subsection{Proof of \cref{proposition:DSIPI}}

\begin{proof}
Note that we have
\begin{align*}
    \Pr(Y_i \le y, W_i \le w) = C_o(F_Y(y | S_i, X_i, P_i=E), F_W(w|S_i, X_i, P_i=E)| S_i, X_i, P_i = E).
\end{align*}
Consequently, we get 
\begin{align*}
&\mathbb{E}\left[Y_iW_i|S_i,X_i,P_i=E\right]\\
&=\int\int ywdC_o(F_Y(y|S_i,X_i, P_i=E),F_W(w|S_i,X_i, P_i=E)|S_i,X_i, P_i=E)\\
&=\int_0^1\int_0^1 F_Y^{-1}(u|S_i,X_i, P_i=O)F_W^{-1}(v|S_i,X_i, P_i=E)|S_i,X_i, P_i=E)dC_o(u,v|S_i,X_i, P_i=E)\\
&=\int_{(u,v)\in[0,1]\times(1-\rho(X_i,S_i),1]} F_Y^{-1}(u|S_i,X_i, P_i=O)dC_o(u,v|S_i,X_i, P_i=E)
\\
&=\int_0^1 F_Y^{-1}(u|S_i,X_i, P_i=O) \left(\int_{1-\rho(X_i,S_i)}^1 dC_o(v | u, S_i,X_i, P_i=E)\right) d u \\
&=\int_0^1 F_Y^{-1}(u|S_i,X_i, P_i=O) (1 - C_o(1-\rho(X_i,S_i)|u, S_i,X_i, P_i=E)) d u \\
&=\rho(X_i,S_i) \int_{0}^{1} F_Y^{-1}(u|S_i,X_i, P_i=O)\sigma_{C_o, 1}(u; 1-\rho(X_i,S_i)) d u\\
&= \rho(X_i,S_i) \mu_{C_o, 1}(Y\mid S_i,X_i),
\end{align*}
where we have used \Cref{assumption:Comparability} (comparability). 
Similarly, 
\begin{align*}
&\mathbb{E}\left[Y_i(1-W_i)|S_i,X_i,P_i=E\right]\\
& = \mathbb{E}\left[Y_i|S_i,X_i,P_i=E\right] - \mathbb{E}\left[Y_i W_i|S_i,X_i,P_i=E\right] \\
& = \int_{0}^{1} F_Y^{-1}(u|S_i,X_i, P_i=O)d u - \rho(X_i,S_i) \int_{0}^{1} F_Y^{-1}(u|S_i,X_i, P_i=O)\sigma_{C_o, 1}(u; 1-\rho(X_i,S_i)) d u,\\
& = (1 - \rho(X_i,S_i))
\int_{0}^{1} F_Y^{-1}(u|S_i,X_i, P_i=O)\sigma_{C_o, 0}(u; 1-\rho(X_i,S_i)) d u\\
&=(1- \rho(X_i,S_i)) \mu_{C_o, 0}(Y\mid S_i,X_i).
\end{align*}
\end{proof}

\section{Proofs for \cref{sec:estimation}}

\label{sec:proofs-estimation}

\subsection{Proof of \cref{prop:WCB}}

\begin{proof}[Proof of \cref{prop:WCB}] 

The proof of \cref{prop:WCB} consists of two parts. We will discuss the worst-case bound in Part 1 below and the binary case in Part 2 below.

\textbf{Part 1.} It follows from the Fr\'echet-Hoeffding inequality and \cref{Coro: identifiedSets} that the identified set of $\tau_{C_o}$ is $[\tau_{C_-}, \tau_{C_+}]$ and the conclusion follows from \cref{ex:SpecificCopulas}.
It is easy to see that 
\begin{align*}
    \sigma_{C_{+}, 1}(u;\alpha) = \frac{\mathds{1}(u \in (\alpha, 1])}{1-\alpha} \mbox{ and }
    \sigma_{C_{+}, 0}(u;\alpha) = \frac{\mathds{1}(u \in [0, \alpha])}{\alpha}
\end{align*}
for any $\alpha\in [0,1)$. Thus 
\begin{align*}
\mu_{C_+, 1}(S_i,X_i,O)&=\int_{0}^{1} F_Y^{-1}(u|S_i,X_i, P_i=O)\sigma_{C_+, 1}(u; 1- \rho(X_i,S_i)) d u\\
&=\int_{0}^{1} F_Y^{-1}(u|S_i,X_i, P_i=O)\frac{\mathds{1}(u \in (1- \rho(X_i,S_i), 1])}{\rho(X_i,S_i)} d u\\
&=AVaR_{1- \rho(X_i,S_i)}(Y_i\mid S_i,X_i,P_i=O) \mbox{ and}\\
\mu_{C_+, 0}(S_i,X_i,O)&=-AVaR_{\rho(X_i,S_i)}(-Y_i\mid S_i,X_i,P_i=O).
 \end{align*}
 Similarly, 
\begin{align*}
    \sigma_{C_-, 1}(u;\alpha) = \frac{\mathds{1}(u\in[0,1-\alpha])}{1-\alpha}\mbox{ and }
    \sigma_{C_-, 0}(u;\alpha) = \frac{\mathds{1}(u\in(1-\alpha,1])}{\alpha}.
\end{align*}
As a result, 
\begin{align*}
    \mu_{C_-, 1}(S_i,X_i,O)&=-AVaR_{1-\rho(X_i,S_i)}(-Y_i\mid S_i,X_i,P_i=O) \mbox{ and }\\
    \mu_{C_-, 0}(S_i,X_i,O)&=AVaR_{\rho(X_i,S_i)}(Y_i\mid S_i,X_i,P_i=O).
\end{align*}

\textbf{Part 2.} Note that
 \begin{align*}
     &\mathbb{E}\left[\mu_{C_+,1}(S_i,X_i,O)\frac{\rho(S_i, X_i)}{\rho(X_i)} | P_i = E\right]\\
& =
\mathbb{E}\left[\mathbb{E}\left[\mu_{C_+,1}(S_i,X_i,O)\frac{\rho(S_i, X_i)}{\rho(X_i)} | S_i, X_i, P_i = E\right]  | P_i = E\right] \\
& = 
\mathbb{E}\left[\mathbb{E}\left[\mu_{C_+,1}(S_i,X_i,O)\frac{W_i}{\rho(X_i)} | S_i, X_i, P_i = E\right]  | P_i = E\right] \\
& = 
\mathbb{E}\left[\mathbb{E}\left[\mu_{C_+,1}(S_i,X_i,O)\frac{W_i}{\rho(X_i)} | X_i, P_i = E\right]  | P_i = E\right] \\
& = 
\mathbb{E}\left[\mathbb{E}\left[\mu_{C_+,1}(S_i,X_i,O) | W_i = 1, X_i, P_i = E\right]  | P_i = E\right].
 \end{align*}
Similarly, we have
\begin{align*}
      \mathbb{E}\left[\mu_{C_+,0}(S_i,X_i,O)\frac{1-\rho(S_i, X_i)}{1-\rho(X_i)} | P_i = E\right] = \mathbb{E}\left[\mathbb{E}\left[\mu_{C_+,0}(S_i,X_i,O)| W_i = 0, X_i, P_i = E\right]  | P_i = E\right].
\end{align*}
Then, according to Theorem 4 (ii) in \cite{athey2025surrogate} and its proof on p. 68, we have
\begin{align}\label{eq:WUB}
\tau_{C_+}&=\mathbb{E}\left[\mu_{C_+,1}(S_i,X_i,O)\frac{\rho(S_i, X_i)}{\rho(X_i)}-\mu_{C_+,0}(S_i,X_i,O)\frac{1-\rho(S_i, X_i)}{1-\rho(X_i)}|P_i=E\right] \nonumber\\
& = 
\mathbb{E}[\mu(S_i(1), X_i, O)|P_i=E] - 
\mathbb{E}[\mu(S_i(0), X_i, O)|P_i=E] \nonumber\\
& \quad +
\mathbb{E}\left[
(\mu_{C_+,1}(S_i,X_i,O) - \mu_{C_+,0}(S_i,X_i,O)) 
\left\{
\frac{\rho(S_i, X_i) (1 - \rho(S_i, X_i))}{\rho(X_i) (1 - \rho(X_i))}
\right\}
|P_i=E\right].
\end{align}
When $Y$ is binary, 
 \begin{align*}
     \mu_{C_+,1}(S_i,X_i,O)&=AVaR_{1-\rho(S_i,X_i)}(Y_i\mid S_i,X_i,P_i=O)\\
     &=\frac{1}{\rho(S_i,X_i)}\int_{1-\rho(S_i,X_i)}^1 F_Y^{-1}(u)du\\
     &=\frac{\mu(S_i,X_i,O)}{\rho(S_i,X_i)}\mathds{1}(\mu(S_i,X_i,O)<\rho(S_i,X_i))+\mathds{1}(\mu(S_i,X_i,O)\geq\rho(S_i,X_i)), \mbox{ and}\\
     \mu_{C_+,0}(S_i,X_i,O)&=-AVaR_{\rho(S_i,X_i)}(-Y_i\mid S_i,X_i,P_i=O)\\
     &=\frac{1}{1-\rho(S_i,X_i)}\int_{0}^{1-\rho(S_i,X_i)} F_Y^{-1}(u)du\\
     &=\frac{\mu(S_i,X_i,O)-\rho(S_i,X_i)}{1-\rho(S_i,X_i)}\mathds{1}(\mu(S_i,X_i,O)\geq\rho(S_i,X_i)) \\
     &=- \frac{1-\mu(S_i,X_i,O)}{1-\rho(S_i,X_i)}\mathds{1}(\mu(S_i,X_i,O)\geq\rho(S_i,X_i)) + \mathds{1}(\mu(S_i,X_i,O)\geq\rho(S_i,X_i)). \\
 \end{align*}
Consequently, we obtain that 
  \begin{align*}
     &\quad\ \mu_{C_+,1}(S_i,X_i,O)-\mu_{C_+,0}(S_i,X_i,O)\\
     &=\frac{\mu(S_i,X_i,O)}{\rho(S_i,X_i)}\mathds{1}(\mu(S_i,X_i,O)<\rho(S_i,X_i))+\frac{1-\mu(S_i,X_i,O)}{1-\rho(S_i,X_i)}\mathds{1}(\mu(S_i,X_i,O)\geq\rho(S_i,X_i))\\
     &=\min\left(\frac{\mu(S_i,X_i,O)}{\rho(S_i,X_i)}, \frac{1-\mu(S_i,X_i,O)}{1-\rho(S_i,X_i)}\right).\text{ This equals $\Delta_S^U$ on page 72 of \cite{athey2025surrogate}.}
 \end{align*}
 \end{proof}

\subsection{Proof of \Cref{thm:CLT-for-wc}}

\subsubsection{Introduction and Technical Lemma}

The proof of \Cref{thm:CLT-for-wc} verifies the conditions for Theorem 3.1 of \cite{chernozhukov2018doubleML}. In particular, we would verify the modified version of Assumption 3.2 in \cite{chernozhukov2018doubleML} because the moment function is not differentiable for the binary case. 

Consider the moment function with the following affine form.
\begin{align*}
    m(Z_i, \tau, \eta) = m^a(Z_i, \eta) \tau + m^b(Z_i, \eta).
\end{align*}
We denote by $\{\delta_n\}_{n=1}^{\infty}$ and $\{\Delta_n\}_{n=1}^{\infty}$ the sequences of positive constants converge to zero such that $\delta_n \ge n^{-1/2}$. We restate asymptotic theory for DML estimators in \cite{chernozhukov2018doubleML} with preliminary conditions.

\begin{lemma}[c.f., Theorem 3.1 of \cite{chernozhukov2018doubleML}] \label{lemma:DML-asym-normality}.
For $n \ge 3$ and $P \in \mathcal{P}_N$, 
the true parameter $\tau$ and true nuisance parameter $\eta$ satisfy the moment condition $\mathbb{E}_P[m(W_i, \tau, \eta)]$, and 
the nuisance parameter $\widehat{\eta}_n := \widehat{\eta}_n(\mathcal{F}_{i}^c)$ belongs to the realization set $R_n$ with probability $1 - \Delta_n$, where $R_n$ contains true parameter $\eta$ and is governed by the following conditions.

\begin{enumerate}[label=(\alph*)]
    \item (Identification condition) The singular values of the matrix
    $
        \mathbb{E}_P[m^a(Z_i; \eta)]
    $
    are between positive constants $t$ and $T$.

    \item (Moment conditions) For some $q > 2$, 
    \begin{gather}
        \sup_{\tilde{\eta} \in \mathbb{R}_n} \| m^a(Z_i, \widetilde{\eta}) \|_{P, q} \le T,  \label{eq:Lq-m_a} \\
        \sup_{\tilde{\eta} \in \mathbb{R}_n} \| m(Z_i, \tau, \widetilde{\eta}) \|_{P, q} \le T. \label{eq:Lq-m} 
    \end{gather}
    by some absolute positive constant $T$.

    \item (The statistical rates)
    \begin{gather}
    \sup_{\widetilde{\eta} \in R_n} \| \mathbb{E}_P[m^a(Z_i, \widetilde{\eta})]- \mathbb{E}_P[m^a(Z_i, \eta)] \| \le \delta_n,  \label{eq:rate-diff-m_a} \\
    \sup_{\widetilde{\eta} \in R_n} \| m(Z_i, \tau, \widetilde{\eta})- m(Z_i, \tau, \eta) \|_{P, 2} \le \delta_n, \label{eq:rate-diff-L2}  \\
    \lambda_n := \sup_{\widetilde{\eta} \in R_n} \| \mathbb{E}_P[m(Z_i, \tau, \widetilde{\eta}) - \mathbb{E}_P[m(Z_i, \tau, \eta)] \| \le n^{-1/2} \delta_n \label{eq:rate-diff-m}
\end{gather}
\end{enumerate}

The DML estimator follows the asymptotic normality with variance 
\begin{align*}
    \mathbb{E}_P[m(Z_i, \tau, \eta)m(Z_i, \tau, \eta)^T].
\end{align*}
The asymptotic variance is positive definite when its singular values is bounded below by positive constant.
\end{lemma}
Note that \Cref{eq:rate-diff-m} is a high-level condition that is verified in the proof of Theorem 3.1 of \cite{chernozhukov2018doubleML}. \Cref{eq:rate-diff-m} is satisfied when the near-orthogonality condition and the statistical rate of the second-order derivative in Assumptions 3.1 and 3.2 of \cite{chernozhukov2018doubleML} are satisfied.

\begin{proof}
     It comes directly from the proof of Theorem 3.1 in \cite{chernozhukov2018doubleML}. All conditions except  Condition \eqref{eq:rate-diff-m} and positive-definiteness of variance matrix are the same as Assumption 3.2 in \cite{chernozhukov2018doubleML}. Therefore, all steps except Step 3 and Step 5 in the proof of Theorem 3.1 in \cite{chernozhukov2018doubleML} works. Also, the proofs in Step 2 holds except for $\mathcal{I}_{4, k}$. Therefore, it is enough to discuss $\mathcal{I}_{4, k}$ in the proof of Theorem 3.1.
     
     Condition \eqref{eq:rate-diff-m} directly verified that $\mathcal{I}_{4, k}$ in Equation (A.16) of \cite{chernozhukov2018doubleML} is less than $\delta_n$.
     \begin{align*}
         \mathcal{I}_{4, k} := \sqrt{|\mathcal{F}_k|} \| \mathbb{E}[m(Z_i, \tau, \widehat{\eta}) - \mathbb{E}[m(Z_i, \tau, \eta)] \| = O_{P_n}(\sqrt{n} \lambda_n).
     \end{align*}
     Therefore, the asymptotic normality of DML estimators holds by the central limit theorem that allows degenerate case. (e.g., Theorem 10.2 (b) in \cite{potscher1997DynamicNonlinearEconometric}.) The asymptotic variance is positive definite when its singular values are bounded below by a positive constant.
\end{proof}

\subsubsection{Main Proof}

From \Cref{lemma:DML-asym-normality}, it is sufficient to verify that conditions in \Cref{lemma:DML-asym-normality} holds for $m(W_i, \tau, \eta) = [m_{C_+}(Z_i, \tau_{C_+}, \eta), m_{C_-}(Z_i, \tau_{C_-}, \eta)]^T$, where $\tau = (\tau_{C_+}, \tau_{C_-})^T$.  Note that we can write $m(Z_i, \tau, \eta)$ in the following affine form.
\begin{align*}
    m(Z_i, \tau, \eta) = m^a(Z_i, \eta) \tau + m^b(Z_i, \eta), \text{ where }
    m^a(Z_i, \eta) = \begin{pmatrix}
    \mathds{1}_{P_i = E} / \varphi & 0 \\
    0 & \mathds{1}_{P_i = E} / \varphi
\end{pmatrix}.
\end{align*}
\begin{lemma} \label{lemma:DML-verification-wc}
    Under Assumptions \ref{assumption: Random Sample}, \ref{assumption: Unconfounded}, \ref{assumption:Comparability}, \ref{assumption:regularity-wc}, and  \ref{assumption:realization-set-wc}, Conditions in \Cref{lemma:DML-asym-normality} hold. In addition, \begin{align}
        \mathbb{E}[m_{C_+}^2(Z_i, \tau_{C_+}, \eta)] > c_0 \text{ and } \mathbb{E}[m_{C_-}^2(Z_i, \tau_{C_-}, \eta)] > c_0 \label{eq:pd-variance}
    \end{align} 
    for some positive constant $c_0$ which only depends on $c$ and $\epsilon$.
\end{lemma}

\begin{proof}[Proof of \Cref{lemma:DML-verification-wc} ]
    The proof is motivated by the proofs of Theorem 4.1 of \cite{chen_ritzwoller_2023} and Theorem 2 of \cite{dorn2024dvds}. In this proof, we focus on $m_{C_+}(Z_i, \tau_{C_+}, \eta)$ except for the verification of Condition (a). This is because we can have the similar result for $m_{C_-}(Z_i, \tau_{C_-}, \eta)$, and we can use $(a^2 + b^2)^{1/2} \le (|a| + |b|)$ to conclude results for $m(Z_i, \tau, \eta)$.
    
    We introduce some notation for the reader's convenience. Let $m_{C_+}^a(Z_i, \eta) = \mathds{1}_{P_i = E} / \varphi$, and
    {\footnotesize
\begin{align*}
    m_1(Z_i, \tau_{C_+}, \eta) & = \frac{\mathds{1}_{P_i = E}}{\varphi} \frac{W_i}{\rho(X_i)} (\mu_{C_+, 1}(S_i, X_i, O) - \bar{\mu}_{C_+, 1}(1, X_i)), \\
    m_2(Z_i, \tau_{C_+}, \eta) & = - \frac{\mathds{1}_{P_i = E}}{\varphi} \frac{1-W_i}{1-\rho(X_i)} (\mu_{C_+, 0}(S_i, X_i, O) - \bar{\mu}_{C_+, 0}(0, X_i)), \\
    m_3(Z_i, \tau_{C_+}, \eta) & =  \frac{\mathds{1}_{P_i = E}}{\varphi} (\bar{\mu}_{C_+, 1}(1, X_i) - \bar{\mu}_{C_+, 0}(0, X_i) - \tau_{C_+}), \\
    m_4(Z_i, \tau_{C_+}, \eta) & = \frac{\mathds{1}_{P_i = O}}{\varphi} \frac{\varphi(S_i, X_i)}{1-\varphi(S_i, X_i)}\frac{\rho(S_i, X_i)}{\rho(X_i)} (H_U(Y_i, q_{C_+}(S_i, X_i, O), \rho(S_i, X_i)) -\mu_{C_+, 1}(S_i, X_i, O)), \\
    m_5(Z_i, \tau_{C_+}, \eta) & = - \frac{\mathds{1}_{P_i = O}}{\varphi} \frac{\varphi(S_i, X_i)}{1-\varphi(S_i, X_i)} 
    \frac{1-\rho(S_i, X_i)}{1-\rho(X_i)} (H_L(Y_i, q_{C_+}(S_i, X_i, O), 1 - \rho(S_i, X_i)) - \mu_{C_+, 0}(S_i, X_i, O)), \\
    m_6(Z_i, \tau_{C_+}, \eta) & = \frac{\mathds{1}_{P_i = E}}{\varphi} \frac{1}{\rho(X_i)} \left[q_{C_+}(S_i, X_i, O) - \mu_{C_+, 1}(S_i, X_i, O)\right] \left(W_i - \rho(S_i, X_i)\right), \\
    m_7(Z_i, \tau_{C_+}, \eta) & =
    \frac{\mathds{1}_{P_i = E}}{\varphi} \frac{1}{1-\rho(X_i)} \left[q_{C_+}(S_i, X_i, O) - \mu_{C_+, 0}(S_i, X_i, O)\right] \left(W_i - \rho(S_i, X_i)\right).
\end{align*}
}
Note that 
\begin{align*}
    m_{C_+}(Z_i, \tau_{C_+}, \eta) = \sum_{j=1}^7 m_j(Z_i, \tau_{C_+}, \eta).
\end{align*}
Parts 1 to 4 below verify the conditions in \Cref{lemma:DML-asym-normality} 

\textbf{Part 1: Verification of Condition (a).} It holds because $\mathbb{E}_P[m^a(Z_i, \eta)] = I_2$.

\textbf{Part 2: Verification of Condition (b).} We verify conditions \eqref{eq:Lq-m_a} and \eqref{eq:Lq-m} in Parts 2-1 and 2-2 below.

\textbf{Part 2-1: Verification of Condition \eqref{eq:Lq-m_a}.} Condition \eqref{eq:Lq-m_a} holds because
\begin{align*}
    \sup_{\widetilde{\eta} \in \mathbb{R}_n} \| m_{C_+}^a(Z_i, \widetilde{\eta}) \| =   \sup_{\widetilde{\eta} \in \mathbb{R}_n} \left|\frac{\varphi}{\widetilde{\varphi}}\right| \le \frac{1-\epsilon}{\epsilon}.
\end{align*}

\textbf{Part 2-2: Verification of Condition \eqref{eq:Lq-m}.} 
We will look at $\|m_j(Z_i, \tau_{C_+}, \widetilde{\eta})\|_{P, q}$ for $j=1, \dotsc, 7$. Note that $|Y| < T$ for some absolute constant $T$ under boundedness of $Y$.
\begin{align*}
    \|m_1(Z_i, \tau_{C_+}, \widetilde{\eta})\|_{P, q} 
    & \le
    \epsilon^{-2}
    \left\|\widetilde{\mu}_{C_+, 1}(S_i, X_i, O) - \widetilde{\bar{\mu}}_{C_+, 1}(1, X_i)\right\|_{P, q} \\
    & \le
     \epsilon^{-2}
    \left\|\widetilde{\mu}_{C_+, 1}(S_i, X_i, O) - \mu_{C_+, 1}(S_i, X_i, O) \right\|_{P, q} \\
    & \quad +  \epsilon^{-2} \left\|\widetilde{\bar{\mu}}_{C_+, 1}(1, X_i) - \bar{\mu}_{C_+, 1}(1, X_i)\right\|_{P, q} \\
    & \quad +  \epsilon^{-2} \left\|\mu_{C_+, 1}(S_i, X_i, O) \right\|_{P, q} + \epsilon^{-2} \left\|\bar{\mu}_{C_+, 1}(1, X_i)\right\|_{P, q} \\
    & \le
    T_{\epsilon}
\end{align*}
for some positive constant $T_{\epsilon}$ which only depends on $T$ and $\epsilon$ because of Assumptions \ref{assumption:realization-set-wc} and the boundedness of $Y$, which implies that $\mu_{C_+, 1}(S_i, X_i, O)$ and $\bar{\mu}_{C_+, 1}(1, X_i)$ are bounded. Similar calculation shows that $\|m_j(Z_i, \tau_{C_+}, \widetilde{\eta})\|_{P, q} \le T_{\epsilon}$ by some constant $T_{\epsilon}$ which only depends on $T$ and $\epsilon$ for $j = 2, 3, 6, 7$.
{\footnotesize
\begin{align*}
    \|m_4(Z_i, \tau_{C_+}, \widetilde{\eta})\|_{P, q} 
    & \le
    \frac{(1-\epsilon)}{\epsilon^3}
    \left\|\widetilde{\rho}(S_i, X_i) \widetilde{q}_{C_+}(S_i, X_i, O) + [Y_i - \widetilde{q}_{C_+}(S_i, X_i, O)]_+ -\widetilde{\rho}(S_i, X_i) \widetilde{\mu}_{C_+, 1}(S_i, X_i, O)\right\|_{P, q} \\
    & \le 
    \frac{(1-\epsilon)^2}{\epsilon^3}
    \left\|\widetilde{q}_{C_+}(S_i, X_i, O) - q_{C_+}(S_i, X_i, O) \right\|_{P, q} \\
    & \quad +  
    \frac{(1-\epsilon)^2}{\epsilon^3}
    \left\|\widetilde{\mu}_{C_+, 1}(S_i, X_i, O) - \mu_{C_+, 1}(S_i, X_i, O)\right\|_{P, q} \\
    & \quad +  
    \frac{(1-\epsilon)}{\epsilon^3}
    \left\|[Y_i - \widetilde{q}_{C_+}(S_i, X_i, O)]_{+} - [Y_i - q_{C_+}(S_i, X_i, O)]_{+} \right\|_{P, q} \\
    & 
    \frac{(1-\epsilon)^2 + (1-\epsilon)}{\epsilon^3}
    \left\|\widetilde{q}_{C_+}(S_i, X_i, O) - q_{C_+}(S_i, X_i, O) \right\|_{P, q} \\
    & \quad +  
    \frac{(1-\epsilon)^2}{\epsilon^3}
    \left\|\widetilde{\mu}_{C_+, 1}(S_i, X_i, O) - \mu_{C_+, 1}(S_i, X_i, O)\right\|_{P, q} 
    \le
    T_{\epsilon},
\end{align*}
}
where $T_{\epsilon}$ is a positive constant which only depends on $T$ and $\epsilon$. Here, the second last equality holds because $r \mapsto [Y - r]_{+}$ is Lipchitz continuous. Similar calculation shows that $\|m_5(Z_i, \tau_{C_+}, \widetilde{\eta})\|_{P, q} \le T_{\epsilon}$. It concludes  $\|m_{C_+}(Z_i, \tau_{C_+}, \widetilde{\eta})\|_{P, q} \le T_{\epsilon}$.

\textbf{Part 3: Verification of Condition (b).}

\textbf{Part 3-1: Verification of Condition \eqref{eq:rate-diff-m_a}.}
\begin{align*}
        \sup_{\widetilde{\eta}  \in R_n} \| \mathbb{E}[m_{C_+}^{a}(Z_i, \widetilde{\eta})] - \mathbb{E}[m_{C_+}^{a}(Z_i, \eta)] \| 
        & = \mathbb{E}_P\left[\frac{\mathds{1}_{P_i=E}}{\varphi \widetilde{\varphi}}(\widetilde{\varphi} -\varphi)\right] 
        \le \epsilon^{-2} \| \hat{\varphi} -\hat{\varphi} \|_{P, 2}  \le \epsilon^{-2} \delta_n.
    \end{align*}

\textbf{Part 3-2: Verification of Condition \eqref{eq:rate-diff-L2}.}

Note that 
    \begin{align*}
        \|m_{C_+}(Z_i, \tau_{C_+}, \widetilde{\eta}) - m_{C_+}(Z_i, \tau_{C_+}, \eta) \|_{P, 2} \le \sum_{k=1}^{7} \| m_k(Z_i, \tau_{C_+}, \widetilde{\eta}) - m_k(Z_i, \tau_{C_+}, \eta) \|_{P, 2}
    \end{align*}
    We will show in Parts 3-2-1 to 3-2-4 that
    \begin{align*}
        \| m_k(Z_i, \tau_{C_+}, \widetilde{\eta}) - m_k(Z_i, \tau_{C_+}, \eta) \|_{P, 2} \le T_{\epsilon} \delta_n
    \end{align*}
    for some positive constant $T_{\epsilon}$ which depends on $T$ and $\epsilon$ only. This concludes Part 3-2. Before we derive the bounds for the right-hand side, we mention that the boundedness of $Y$ implies that 
    \begin{gather*}
        \mathbb{E}_P[((Y_i - q_{C_+}(S_i, X_i, O)]_{+})^2|S_i, X_i, P_i = 0] \le T_0, \\
        \mathbb{E}_P[((Y_i - q_{C_+}(S_i, X_i, O)]_{-})^2|S_i, X_i, P_i = 0] \le T_0, \\ \|\bar{\mu}_{C_+, 1}(1, X_i) - \bar{\mu}_{C_+, 0}(1, X_i) - \tau_{C_+}\|_{P, 2} \le T_0, \\ 
        \| \mu_{C_+, 1}(S_i, X_i, O) - \bar{\mu}_{C_+, 1}(1, X_i) \|_{P, 2} \le  C, \quad \| \mu_{C_+, 0}(S_i, X_i, O) - \bar{\mu}_{C_+, 0}(0, X_i) \|_{P, 2} \le  T_0.
    \end{gather*}
    for some positive constant $T_0$.

    \textbf{Part 3-2-1: Bounds of $\| m_k(Z_i, \tau_{C_+}, \widetilde{\eta}) - m_k(Z_i, \tau_{C_+}, \eta) \|_{P, 2}$ for $j=1, 2$.}

    We focus on $\| m_1(Z_i, \tau_{C_+}, \widetilde{\eta}) - m_1(Z_i, \tau_{C_+}, \eta) \|_{P, 2}$ since $\| m_2(Z_i, \tau_{C_+}, \widetilde{\eta}) - m_2(Z_i, \tau_{C_+}, \eta) \|_{P, 2}$ can be dealt in a similar way.
    {\footnotesize
    \begin{align*}
        & \| m_1(Z_i, \tau_{C_+}, \widetilde{\eta}) - m_1(Z_i, \tau_{C_+}, \eta) \|_{P, 2}
        \\
        & =
        \bigg\|\frac{\mathds{1}_{P_i = E}}{\hat{\varphi}} \frac{W_i}{\hat{\rho}(X_i)} (\tilde{\mu}_{C_+, 1}(S_i, X_i, O) - \hat{\bar{\mu}}_{C_+, 1}(1, X_i)) 
        - 
        \frac{\mathds{1}_{P_i = E}}{{\varphi}} \frac{W_i}{{\rho}(X_i)} (\mu_{C_+, 1}(S_i, X_i, O) - \bar{\mu}_{C_+, 1}(1, X_i))\bigg\|_{P, 2} \\
        & \le \epsilon^{-2}
        \bigg\| \varphi \frac{1}{\hat{\rho}(X_i)} (\tilde{\mu}_{C_+, 1}(S_i, X_i, O) - \hat{\bar{\mu}}_{C_+, 1}(1, X_i)) 
        - \hat{\varphi} \frac{1}{{\rho}(X_i)} (\mu_{C_+, 1}(S_i, X_i, O) - \bar{\mu}_{C_+, 1}(1, X_i))\bigg\|_{P, 2} \\
        & \le \epsilon^{-2}
        \bigg\| \frac{1}{\hat{\rho}(X_i)} (\tilde{\mu}_{C_+, 1}(S_i, X_i, O) - \tilde{\bar{\mu}}_{C+, 1}(1, X_i)) 
        - \frac{1}{{\rho}(X_i)} (\mu_{C_+, 1}(S_i, X_i, O) - \bar{\mu}_{C+, 1}(1, X_i))\bigg\|_{P, 2} \\
        & \quad + 
        \epsilon^{-2} \bigg\| (\varphi - \hat{\varphi}) \frac{1}{{\rho}(X_i)} (\mu_{C_+, 1}(S_i, X_i, O) - \bar{\mu}_{C_+, 1}(1, X_i))  \bigg\|_{P, 2} \\
        & \le \epsilon^{-4}
        \bigg\| \rho(X_i) (\widetilde{\mu}_{C_+, 1}(S_i, X_i, O) - \widetilde{\bar{\mu}}_{C_+, 1}(1, X_i)) 
        - \hat{\rho}(X_i) (\mu_{C_+, 1}(S_i, X_i, O) - \bar{\mu}_{C_+, 1}(1, X_i))\bigg\|_{P, 2} \\
        & \quad + 
        \epsilon^{-3}
        \bigg\| (\varphi - \hat{\varphi})  (\mu_{C_+, 1}(S_i, X_i, O) - \bar{\mu}_{C_+, 1}(1, X_i))  \bigg\|_{P, 2} \\
        & \le 
        \epsilon^{-4}
        \bigg\|  (\tilde{\mu}_{C_+, 1}(S_i, X_i, O) - \mu_{C_+, 1}(S_i, X_i, O)\bigg\|_{P, 2}
         + 
        \epsilon^{-4}
        \|  \widetilde{\bar{\mu}}_{C_+, 1}(1, X_i) - \bar{\mu}_{C_+, 1}(1, X_i)\bigg\|_{P, 2} \\
        & \quad 
        + \epsilon^{-4}
        \bigg\| 
        (\rho(X_i) - \hat{\rho}(X_i)) (\mu_{C_+, 1}(S_i, X_i, O) - \bar{\mu}_{C_+, 1}(1, X_i))\bigg\|_{P, 2} \\
        & \quad + 
        \epsilon^{-3}
        \Biggr(\mathbb{E}\Bigg[\bigg|   (\mu_{C_+, 1}(S_i, X_i, O) - \bar{\mu}_{C_+, 1}(1, X_i))  \bigg|^{2} \Bigg]\Biggr)^{1/2} |\varphi - \hat{\varphi}| \\
        & \le (2 \epsilon^{-4} + \sqrt{T_0} \epsilon^{-4} + \sqrt{T_0} \epsilon^{-3} ) \delta_n.
    \end{align*}
    }
    The first inequality comes from $|a/b - c/d| = |da - cb| / |bd|$. Similarly, we can show that $\|m_2(Z_i, \tau_{C_+}, \widetilde{\eta}) - m_2(Z_i, \tau_{C_+}, \eta) \|_{P, 2} \le T_2 \delta_n$ where $T_2$ depends on $T$ and $\epsilon$ only.

    \textbf{Part 3-2-2: Bounds of $\| m_3(Z_i, \tau_{C_+}, \widetilde{\eta}) - m_3(Z_i, \tau_{C_+}, \eta) \|_{P, 2}$.}
    {\footnotesize
    \begin{align*}
        &\| m_3(Z_i, \tau_{C_+}, \widetilde{\eta}) - m_3(Z_i, \tau_{C_+}, \eta) \|_{P, 2} \\
        & = 
        \left\|\frac{\mathds{1}_{P_i = E}}{\widetilde{\varphi}} (\widetilde{\bar{\mu}}_{C_+, 1}(1, X_i) - \widetilde{\bar{\mu}}_{C_+, 0}(0, X_i) - \tau_{C_+}) - \frac{\mathds{1}_{P_i = E}}{\varphi} (\bar{\mu}_{C_+, 1}(1, X_i) - \bar{\mu}_{C_+, 0}(0, X_i) - \tau_{C_+})\right\|_{P, 2} \\
        & \le
        \epsilon^{-2}
        \left\|\varphi (\widetilde{\bar{\mu}}_{C_+, 1}(1, X_i) - \widetilde{\bar{\mu}}_{C_+, 0}(0, X_i) - \tau_U) - \widetilde{\varphi} (\bar{\mu}_{C_+, 1}(1, X_i) - \bar{\mu}_{C_+, 0}(0, X_i) - \tau_{C_+})\right\|_{P, 2} \\
        & \le
        \epsilon^{-2}
        \left\|\varphi (\widetilde{\bar{\mu}}_{C_+, 1}(1, X_i) - \bar{\mu}_{C_+, 1}(1, X_i))\right\|_{P, 2} 
        +
        \epsilon^{-2}
        \left\|\varphi (\widetilde{\bar{\mu}}_{C_+, 0}(0, X_i) - \bar{\mu}_{C_+, 0}(0, X_i))\right\|_{P, 2}
        \\
        & \quad + \epsilon^{-2}
        \left\|(\varphi  - \tilde{\varphi})(\bar{\mu}_{C_+, 1}(1, X_i) - \bar{\mu}_{C_+, 0}(0, X_i) - \tau_{C_+})\right\|_{P, 2} \\
        & \le
        \epsilon^{-2}
        \left\|\widetilde{\bar{\mu}}_{C_+, 1}(1, X_i) - \bar{\mu}_{C_+, 1}(1, X_i)\right\|_{P, 2}
        +
        \epsilon^{-2}
        \left\|\widetilde{\bar{\mu}}_{C_+, 0}(0, X_i) - \bar{\mu}_{C_+, 0}(0, X_i)\right\|_{P, 2}
        \\
        & \quad + \epsilon^{-2}
        \left\|\bar{\mu}_{C_+, 1}(1, X_i) - \bar{\mu}_{C_+, 0}(0, X_i) - \tau_{C_+}\right\|_{P, 2} |\widetilde{\varphi} - \varphi|  \\
        & \le T_3 \delta_n
    \end{align*}
    }
    where $T_3$ only depends on $T$ and $\epsilon$.

    \textbf{Part 3-2-3: Bounds of $\| m_k(Z_i, \tau_{C_+}, \widetilde{\eta}) - m_k(Z_i, \tau_{C_+}, \eta) \|_{P, 2}$ for $j=4, 5$.}

    Let's define $V_1(S_i, X_i)$ and $V_2(S_i, X_i)$ by
    \begin{gather*}
        \frac{1}{V_1(S_i, X_i)} =  \frac{1}{{\varphi}} \frac{{\varphi}(S_i, X_i)}{1-{\varphi}(S_i, X_i)} \frac{1}{{\rho}(X_i)}, \quad  
        \frac{1}{\widetilde{V}_1(S_i, X_i)} =  \frac{1}{\tilde{\varphi}} \frac{\tilde{\varphi}(S_i, X_i)}{1-\tilde{\varphi}(S_i, X_i)} \frac{1}{\tilde{\rho}(X_i)} \\
        V_2(S_i, X_i) = V_1(S_i, X_i) / \rho(S_i, X_i), \quad
        \widetilde{V}_2(S_i, X_i) = \widetilde{V}_1(S_i, X_i) / \tilde{\rho}(S_i, X_i).
    \end{gather*}
    Under our assumption, we have $\min\{V_1(S_i, X_i), \widetilde{V}_1(S_i, X_i)\} \ge \epsilon^3 / (1 - \epsilon) := \epsilon_1$ and $\min\{V_2(S_i, X_i), \widetilde{V}_2(S_i, X_i)\} \ge \epsilon^3 / (1 - \epsilon)^2 := \epsilon_2$. The Mean-value theorem implies that 
    \begin{align*}
        | \widetilde{V}_1(S_i, X_i) - V_1(S_i, X_i) | & \le C_{\epsilon} \left(|\hat{\varphi} - \varphi| + |\tilde{\varphi}(S_i, X_i) - \varphi(S_i, X_i)| + |\tilde{\rho}(X_i) - \rho(X_i)|  \right) \\
         | \widetilde{V}_2(S_i, X_i) - V_2(S_i, X_i) | & \le T_{\epsilon} \Big(|\tilde{\varphi} - \varphi| + |\tilde{\varphi}(S_i, X_i) - \varphi(S_i, X_i)| \\
         & \quad \quad + |\tilde{\rho}(X_i) - \rho(X_i)| + |\tilde{\rho}(S_i, X_i) - \rho(S_i, X_i)| \Big),
    \end{align*}
    where $T_{\epsilon}$ is a positive number which only depends on $T$ and $\epsilon$.
    
    Therefore, the bound of $\| m_4(Z_i, \tau_{C_+}, \widetilde{\eta}) - m_4(Z_i, \tau_{C_+}, \eta) \|_{P, 2}$ can be derived as follows.
    \begin{align*}
        & \| m_4(Z_i, \tau_{C_+}, \widetilde{\eta}) - m_4(Z_i, \tau_{C_+}, \eta) \|_{P, 2} \\
        & = 
        \Bigg\|\frac{\mathds{1}_{P_i = O}}{\hat{\varphi}} \frac{\hat{\varphi}(S_i, X_i)}{1-\hat{\varphi}(S_i, X_i)} \frac{\hat{\rho}(S_i, X_i)}{\hat{\rho}(X_i)} (H_U(Y_i, \widetilde{q}_{C_+}(S_i, X_i, O), \tilde{\rho}(S_i, X_i)) -\tilde{\mu}_{C_+, 1}(S_i, X_i, O)) \\
        & \quad 
        -
        \frac{\mathds{1}_{P_i = O}}{{\varphi}} \frac{{\varphi}(S_i, X_i)}{1-{\varphi}(S_i, X_i)} \frac{{\rho}(S_i, X_i)}{{\rho}(X_i)} (H_U(Y_i, q_{C_+}(S_i, X_i, O), {\rho}(S_i, X_i)) - \mu_{C_+, 1}(S_i, X_i, O))
        \Bigg\|_{P, 2} \\
        & \le 
        \Bigg\|\frac{\mathds{1}_{P_i = O}}{\widetilde{V}_2(S_i, X_i)} (\widetilde{q}_{C_+}(S_i, X_i, O) -\tilde{\mu}_{C_+, 1}(S_i, X_i, O)) 
        -
        \frac{\mathds{1}_{P_i = O}}{{V}_2(S_i, X_i)} (q_{C_+}(S_i, X_i, O) -\mu_{C_+, 1}(S_i, X_i, O))
        \Bigg\|_{P, 2} \\
        & \quad +
        \Bigg\|\frac{\mathds{1}_{P_i = O}}{\widetilde{V}_1(S_i, X_i)} (Y_i - \widetilde{q}_{C_+}(S_i, X_i, O))_{+} 
        -
        \frac{\mathds{1}_{P_i = O}}{V_1(S_i, X_i)} (Y_i - {q}_{U}(S_i, X_i, O))_{+}
        \Bigg\|_{P, 2} \\
        & \le \epsilon_2^{-2}
        \Bigg\|V_2(S_i, X_i) (\widetilde{q}_{C_+}(S_i, X_i, O) -\tilde{\mu}_{C_+, 1}(S_i, X_i, O)) 
        -
        \widetilde{V}_2(S_i, X_i) (q_{C_+}(S_i, X_i, O) - \mu_{C_+, 1}(S_i, X_i, O))
        \Bigg\|_{P, 2} \\
        & \quad + \epsilon_1^{-2}
        \Bigg\| \mathds{1}_{P_i = O}\bigg|V_1(S_i, X_i) (Y_i - \widetilde{q}_{C_+}(S_i, X_i, O))_{+} 
        -
        \widehat{V}_1(S_i, X_i) (Y_i - q_{C_+}(S_i, X_i, O))_{+}\bigg|
        \Bigg\|_{P, 2} \\
        & \le
        \epsilon_2^{-2}
        \Bigg\|V_2(S_i, X_i) \bigg((\widetilde{q}_{C_+}(S_i, X_i, O) - q_{C_+}(S_i, X_i, O))-(\tilde{\mu}_{C_+, 1}(S_i, X_i, O)  - \mu_{C_+, 1}(S_i, X_i, O)\bigg) 
        \Bigg\|_{P, 2} \\
         & \quad +
        \epsilon_2^{-2}
        \Bigg\|(V_2(S_i, X_i)
        -
        \widetilde{V}_2(S_i, X_i)) (q_{C_+}(S_i, X_i, O) - \mu_{C_+, 1}(S_i, X_i, O))
        \Bigg\|_{P, 2} \\
        & \quad + \epsilon_1^{-2}
        \Bigg\|\mathds{1}_{P_i = O}\bigg|V_1(S_i, X_i) \bigg[(Y_i - \widetilde{q}_{C_+}(S_i, X_i, O))_{+} 
        -
         (Y_i - {q}_{U}(S_i, X_i, O))_{+}\bigg]
        \bigg|\Bigg\|_{P, 2} \\
         & \quad + \epsilon_1^{-2}
        \Bigg\| \mathds{1}_{P_i = O}\bigg|\bigg( V_1(S_i, X_i) -
        \widetilde{V}_1(S_i, X_i) \bigg) (Y_i - q_{C_+}(S_i, X_i, O))_{+}
        \bigg|^2\Bigg\|_{P, 2} \\
        & \le T_4 \delta_n
    \end{align*}
    where $T_4$ depends on $T$ and $\epsilon$ only. 
    We can similarly show that $\| m_5(Z_i, \tau_{C_+}, \widetilde{\eta}) - m_5(Z_i, \tau_{C_+}, \eta) \|_{P, 2} < T_5 \delta_n$ where $T_5 > 0$ depends only on $T$ and $\epsilon$. 

    \textbf{Part 3-2-4: Bounds of $\| m_k(Z_i, \tau_{C_+}, \widetilde{\eta}) - m_k(Z_i, \tau_{C_+}, \eta) \|_{P, 2}$ for $j=6, 7$.}
    \begin{align*}
        & \| m_6(Z_i, \tau_{C_+}, \widetilde{\eta}) - m_6(Z_i, \tau_{C_+}, \eta) \|_{P, 2}\\
        & =
        \Bigg\|\frac{\mathds{1}_{P_i = E}}{\tilde{\varphi}} \frac{1}{\hat{\rho}(X_i)} \left[\widetilde{q}_{C_+}(S_i, X_i, O) - \tilde{\mu}_{C_+, 1}(S_i, X_i, O)\right] \left(W_i - \tilde{\rho}(S_i, X_i)\right) \\
        & \quad \quad - \frac{\mathds{1}_{P_i = E}}{\varphi} \frac{1}{\rho(X_i)} \left[q_{C_+}(S_i, X_i, O) - \mu_{C_+, 1}(S_i, X_i, O)\right] \left(W_i - \rho(S_i, X_i)\right)\bigg\|_{P, 2} \\
        & \le 
        \Bigg\|\frac{\mathds{1}_{P_i = E}}{\hat{\varphi}} \frac{1}{\tilde{\rho}(X_i)} \left[\widetilde{q}_{C_+}(S_i, X_i, O) - \tilde{\mu}_{C_+, 1}(S_i, X_i, O)\right]  \\
        & \quad \quad - \frac{\mathds{1}_{P_i = E}}{\varphi} \frac{1}{\rho(X_i)} \left[q_{C_+}(S_i, X_i, O) - \mu_{C_+, 1}(S_i, X_i, O)\right] \Bigg\|_{P, 2} \\
        & \quad + 
        \Bigg\|\frac{\mathds{1}_{P_i = E}}{\tilde{\varphi}} \frac{\tilde{\rho}(S_i, X_i)}{\tilde{\rho}(X_i)} \left[\widetilde{q}_{C_+}(S_i, X_i, O) - \tilde{\mu}_{C_+, 1}(S_i, X_i, O)\right]  \\
        & \quad \quad - \frac{\mathds{1}_{P_i = E}}{\varphi} \frac{\rho(S_i, X_i)}{\rho(X_i)} \left[q_{C_+}(S_i, X_i, O) - \mu_{C_+, 1}(S_i, X_i, O)\right] \Bigg\|_{P, 2} \\
        & \le 
        \epsilon^{-4}
        \Bigg\| \varphi \rho(X_i) \left[\widetilde{q}_{C_+}(S_i, X_i, O) - \tilde{\mu}_{C_+, 1}(S_i, X_i, O)\right]  - \tilde{\varphi} \hat{\rho(X_i)} \left[q_{C_+}(S_i, X_i, O) - \mu_{C_+, 1}(S_i, X_i, O)\right] \Bigg\|_{P, 2}  \\
        & \quad + ((1-\epsilon)/\epsilon^2)^{2}
        \Bigg\| \frac{\varphi \rho(X_i)}{\rho(S_i, X_i)} \left[\widetilde{q}_{C_+}(S_i, X_i, O) - \tilde{\mu}_{C_+, 1}(S_i, X_i, O)\right]  \\
        & \quad \quad \quad \quad \quad \quad \quad \quad -  
        \frac{\tilde{\varphi}\tilde{\rho}(X_i)}{\tilde{\rho}(S_i, X_i)} \left[q_{C_+}(S_i, X_i, O) - \mu_{C_+, 1}(S_i, X_i, O)\right] \Bigg\|_{P, 2} \\
        & \le 
        \epsilon^{-4}
        \Bigg\| \varphi \rho(X_i) \left[(\widetilde{q}_{C_+}(S_i, X_i, O) - \tilde{\mu}_{C_+, 1}(S_i, X_i, O)) - (q_{C_+}(S_i, X_i, O) - \mu_{C_+, 1}(S_i, X_i, O))\right] \Bigg\|_{P, 2} \\
        & \quad + 
        \epsilon^{-4}
        \Bigg\| (\varphi \rho(X_i) - \tilde{\varphi} \tilde{\rho}(X_i)) \left[q_{C_+}(S_i, X_i, O) - \mu_{C_+, 1}(S_i, X_i, O)\right] \Bigg\|_{P, 2} \\
        & \quad + \frac{(1-\epsilon)^2}{\epsilon^4}
        \Bigg\| \frac{\varphi \rho(X_i)}{\rho(S_i, X_i)} \left[(\widetilde{q}_{C_+}(S_i, X_i, O) - \tilde{\mu}_{C_+, 1}(S_i, X_i, O)) - (q_{C_+}(S_i, X_i, O) - \mu_{C_+, 1}(S_i, X_i, O))\right] \Bigg|_{P, 2} \\
        & \quad + \frac{(1-\epsilon)^2}{\epsilon^4}
        \Bigg\| \left(\frac{\varphi \rho(X_i)}{\rho(S_i, X_i)} -  
        \frac{\hat{\varphi}\hat{\rho}(X_i)}{\hat{\rho}(S_i, X_i)} \right)\left[q_{C_+}(S_i, X_i, O) - \mu_{C_+, 1}(S_i, X_i, O)\right] \Bigg\|_{P, 2} \\
        & \le T_6 \delta_n
    \end{align*}
    where $T_6$ depends on $T$ and $\epsilon$ only. We can similarly show that $\| m_7(Z_i, \tau_{C_+}, \widetilde{\eta}) - m_7(Z_i, \tau_{C_+}, \eta) \|_{P, 2} \le T_7 \delta_n$ where $T_5 > 0$ depends only on $T$ and $\epsilon$.

\textbf{Part 3-3: Verification of Condition \eqref{eq:rate-diff-m}.}

\Cref{lemma:tech-remainer-wc} in the supplementary appendix implies that we have
\begin{align*}
    \mathbb{E}_P[m(W_i, \tau_{C_+}, \widetilde{\eta}) - m(W_i, \tau_{C_+}, \eta)] 
    & = 
    \sum_{j=1}^7 \mathbb{E}_P[m_j(W_i, \tau_{C_+}, \widetilde{\eta}) - m_j(W_i, \tau_{C_+}, \eta)] = \sum_{j=1}^8 \mathcal{J}_j,  
\end{align*} 
where
{\footnotesize
\begin{align*}
    \mathcal{J}_1 & = 
    - \mathbb{E}_P\Bigg[\bigg(\frac{\mathds{1}_{P_i = E}}{\widetilde{\varphi}} \frac{W_i}{\widetilde{\rho}(X_i)} - \frac{\mathds{1}_{P_i = E}}{\widetilde{\varphi}} \frac{W_i}{\rho(X_i)} \bigg)(\widetilde{\bar{\mu}}_{C_+, 1}(1, X_i) - \bar{\mu}_{C_+, 1}(1, X_i)) \Bigg] \\
    \mathcal{J}_2 & = 
    \mathbb{E}_P\Bigg[\bigg(\frac{\mathds{1}_{P_i = E}}{\widetilde{\varphi}} \frac{1-W_i}{1-\widetilde{\rho}(X_i)} - \frac{\mathds{1}_{P_i = E}}{\widetilde{\varphi}} \frac{1-W_i}{1-\rho(X_i)} \bigg) (\widetilde{\bar{\mu}}_{C_+, 0}(1, X_i) - \bar{\mu}_{C_+, 0}(1, X_i)) \Bigg] \\
    \mathcal{J}_3 & =
    \mathbb{E}_P\Bigg[
    \frac{\mathds{1}_{P_i = O}}{\widetilde{\varphi}} \frac{\widetilde{\varphi}(S_i, X_i)}{1-\widetilde{\varphi}(S_i, X_i)}\frac{1}{\widetilde{\rho}(X_i)} \\
    & \quad \quad \times \bigg[(\rho(S_i, X_i) \widetilde{q}_{C_+}(S_i, X_i, O) + [Y_i - \widetilde{q}_{C_+}(S_i, X_i, O)]_{+}) \\
    & \quad \quad \quad 
    - (\rho(S_i, X_i) q_{C_+}(S_i, X_i, O) + [Y_i - q_{C_+}(S_i, X_i, O)]_{+}) \bigg]  \Bigg] \\
    \mathcal{J}_4 & =
    -
    \mathbb{E}_P\Bigg[
    \left(\frac{\mathds{1}_{P_i = O}}{\widetilde{\varphi}} \frac{\widetilde{\varphi}(S_i, X_i)}{1-\widetilde{\varphi}(S_i, X_i)}\frac{1}{\widetilde{\rho}(X_i)} - \frac{\mathds{1}_{P_i = O}}{\widetilde{\varphi}} \frac{\varphi(S_i, X_i)}{1-\varphi(S_i, X_i)}\frac{1}{\widetilde{\rho}(X_i)}\right)
    \\
    & \quad \quad \times \bigg[\rho(S_i, X_i) (\widetilde{\mu}_{C_+, 1}(S_i, X_i, O)) - \mu_{C_+, 1}(S_i, X_i, O)) \bigg]  \Bigg] \\
    \mathcal{J}_5 & =
    \mathbb{E}_P\Bigg[
    \left(\frac{\mathds{1}_{P_i = O}}{\widetilde{\varphi}} \frac{\widetilde{\varphi}(S_i, X_i)}{1-\widetilde{\varphi}(S_i, X_i)}\frac{1}{\widetilde{\rho}(X_i)} - 
    \frac{\mathds{1}_{P_i = O}}{\widetilde{\varphi}} \frac{\varphi(S_i, X_i)}{1-\varphi(S_i, X_i)}\frac{1}{\widetilde{\rho}(X_i)} \right)\\
    & \quad \quad \times \bigg[(\widetilde{\rho}(S_i, X_i) - \rho(S_i, X_i)) (\widetilde{q}_{C_+}(S_i, X_i, O) -\widetilde{\mu}_{C_+, 1}(S_i, X_i, O))
     \Bigg] \\
    \mathcal{J}_6 & =
    -
    \mathbb{E}_P\Bigg[
    \frac{\mathds{1}_{P_i = O}}{\widetilde{\varphi}} \frac{\widetilde{\varphi}(S_i, X_i)}{1-\widetilde{\varphi}(S_i, X_i)}\frac{1}{1-\widetilde{\rho}(X_i)} \\
    & \quad \quad \times \bigg[\big((1-\rho(S_i, X_i)) \widetilde{q}_{C_+}(S_i, X_i, O) - [Y_i - \widetilde{q}_{C_+}(S_i, X_i, O)]_{-}) \\
    & \quad \quad \quad - \big((1-\rho(S_i, X_i)) q_{C_+}(S_i, X_i, O) - [Y_i - q_{C_+}(S_i, X_i, O)]_{-}\big) \bigg] \bigg| 
     \Bigg] \\
    \mathcal{J}_7 & =
    \mathbb{E}_P\Bigg[
    \left(\frac{\mathds{1}_{P_i = O}}{\widetilde{\varphi}} \frac{\widetilde{\varphi}(S_i, X_i)}{1-\widetilde{\varphi}(S_i, X_i)}\frac{1}{1-\widetilde{\rho}(X_i)} - \frac{\mathds{1}_{P_i = O}}{\widetilde{\varphi}} \frac{\varphi(S_i, X_i)}{1-\varphi(S_i, X_i)}\frac{1}{1-\widetilde{\rho}(X_i)}\right)
    \\
    & \quad \quad \times \bigg[(1-\rho(S_i, X_i)) (\widetilde{\mu}_{C_+, 0}(S_i, X_i, O)) - \mu_{C_+, 0}(S_i, X_i, O)) \bigg]  \Bigg] \\
    \mathcal{J}_8 & =
    \mathbb{E}_P\Bigg[
    \left(\frac{\mathds{1}_{P_i = O}}{\widetilde{\varphi}} \frac{\widetilde{\varphi}(S_i, X_i)}{1-\widetilde{\varphi}(S_i, X_i)}\frac{1}{1-\widetilde{\rho}(X_i)} - 
    \frac{\mathds{1}_{P_i = O}}{\widetilde{\varphi}} \frac{\varphi(S_i, X_i)}{1-\varphi(S_i, X_i)}\frac{1}{1-\widetilde{\rho}(X_i)} \right)\\
    & \quad \quad \times \bigg[(\widetilde{\rho}(S_i, X_i) - \rho(S_i, X_i)) (\widetilde{q}_{C_+}(S_i, X_i, O) -\widetilde{\mu}_{C_+, 0}(S_i, X_i, O))
     \Bigg].
    \end{align*}
}
Because $Y$ is bounded and $\mathcal{P}(\epsilon \le \varphi(S_i, X_i)\le 1-\epsilon) = 1$ and $\mathcal{P}(\epsilon \le \rho(S_i, X_i)\le 1-\epsilon) = 1$, we have the followings bounds for $\mathcal{J}_j$ for $j=1, \dotsc, 8$.

(1) Bounds for $\mathcal{J}_1$,  $\mathcal{J}_2$, $\mathcal{J}_4$, $\mathcal{J}_7$.
\begin{align*}
    |\mathcal{J}_1| & \le T_{\epsilon} \| \widetilde{\rho}(X_i) - \rho(X_i) \|_{P, 2} \times \| \widetilde{\bar{\mu}}_{C_+, 1}(1, X_i) - \bar{\mu}_{C_+, 1}(1, X_i) \|_{P, 2} \le T_{\epsilon} n^{-1/2} \delta_n, \\
    |\mathcal{J}_2| & \le T_{\epsilon} \| \widetilde{\rho}(X_i) - \rho(X_i) \|_{P, 2} \times \| \widetilde{\bar{\mu}}_{C_+, 0}(0, X_i) - \bar{\mu}_{C_+, 0}(0, X_i) \|_{P, 2} \le T_{\epsilon} n^{-1/2} \delta_n , \\
    |\mathcal{J}_4| & \le T_{\epsilon} \| \widetilde{\varphi}(S_i, X_i) - \varphi(S_i, X_i) \|_{P, 2} \times \| \widetilde{\mu}_{C_+, 1}(S_i, X_i, O) - \mu_{C_+, 1}(S_i, X_i, O) \|_{P, 2} \le T_{\epsilon} n^{-1/2} \delta_n , \\
    |\mathcal{J}_7| & \le T_{\epsilon} \| \widetilde{\varphi}(S_i, X_i) - \varphi(S_i, X_i) \|_{P, 2} \times \| \widetilde{\mu}_{C_+, 0}(S_i, X_i, O) - \mu_{C_+, 0}(S_i, X_i, O) \|_{P, 2} \le T_{\epsilon} n^{-1/2} \delta_n ,
\end{align*}
for some positive absolute constant $T_{\epsilon}$ that only depends on $\epsilon$ and $T$.

(2) Bounds for $\mathcal{J}_5$ and $\mathcal{J}_8$.
Note that
\begin{align*}
    \mathcal{J}_5 
    & =
    \mathbb{E}_P\Bigg[ \frac{\mathds{1}_{P_i = O}}{\widetilde{\varphi}}
    \left( \frac{\widetilde{\varphi}(S_i, X_i)}{1-\widetilde{\varphi}(S_i, X_i)}  - 
     \frac{\varphi(S_i, X_i)}{1-\varphi(S_i, X_i)}\right)\frac{1}{\widetilde{\rho}(X_i)}   \\
    & \quad \quad \times \bigg[(\widetilde{\rho}(S_i, X_i) - \rho(S_i, X_i)) (q_{C_+}(S_i, X_i, O) -\mu_{C_+, 1}(S_i, X_i, O))
      \Bigg] \\
     & \quad +
    \mathbb{E}_P\Bigg[ \frac{\mathds{1}_{P_i = O}}{\widetilde{\varphi}}
    \left( \frac{\widetilde{\varphi}(S_i, X_i)}{1-\widetilde{\varphi}(S_i, X_i)}  - 
     \frac{\varphi(S_i, X_i)}{1-\varphi(S_i, X_i)}\right)\frac{1}{\widetilde{\rho}(X_i)} (\widetilde{\rho}(S_i, X_i) - \rho(S_i, X_i))  \\
    & \quad \quad \quad \times \bigg[ \Big[(\widetilde{q}_{C_+}(S_i, X_i, O) -\widetilde{\mu}_{C_+, 1}(S_i, X_i, O)) - (q_{C_+}(S_i, X_i, O) -\mu_{C_+, 1}(S_i, X_i, O)) \Big]
     \Bigg].
\end{align*}
It implies that
\begin{align*}
    | \mathcal{J}_5 | 
    & \le 
    T_{\epsilon} \| \widetilde{\varphi}(S_i, X_i) - \varphi(S_i, X_i) \|_{P, 2} \times \| \widetilde{\rho}(S_i, X_i) - \rho(S_i, X_i) \|_{P, 2} \\
    & \quad + 
    T_{\epsilon} \| \widetilde{\varphi}(S_i, X_i) - \varphi(S_i, X_i) \|_{P, 2} \times \| \widetilde{q}_{C_+}(S_i, X_i, O) - q_{C_+}(S_i, X_i, O) \|_{P, 2} \\
    & \quad + 
    T_{\epsilon} \| \widetilde{\varphi}(S_i, X_i) - \varphi(S_i, X_i) \|_{P, 2} \times \| \widetilde{\mu}_{C_+, 1}(S_i, X_i, O) - \mu_{C_+, 1}(S_i, X_i, O) \|_{P, 2} \\ & \le T_{\epsilon} n^{-1/2} \delta_n ,
\end{align*}
where $T_{\epsilon}$ is a positive constant that depends only on $\epsilon$ and $T$.

Similarly, we have the following bounds for $\mathcal{J}_8$.
\begin{align*}
    | \mathcal{J}_8 | 
    & \le 
    T_{\epsilon} \| \widetilde{\varphi}(S_i, X_i) - \varphi(S_i, X_i) \|_{P, 2} \times \| \widetilde{\rho}(S_i, X_i) - \rho(S_i, X_i) \|_{P, 2} \\
    & \quad + 
    T_{\epsilon} \| \widetilde{\varphi}(S_i, X_i) - \varphi(S_i, X_i) \|_{P, 2} \times \| \widetilde{q}_{C_+}(S_i, X_i, O) - q_{C_+}(S_i, X_i, O) \|_{P, 2} \\
    & \quad + 
    T_{\epsilon} \| \widetilde{\varphi}(S_i, X_i) - \varphi(S_i, X_i) \|_{P, 2} \times \| \widetilde{\mu}_{C_+, 0}(S_i, X_i, O) - \mu_{C_+, 0}(S_i, X_i, O) \|_{P, 2} \\
    & \le T_{\epsilon} n^{-1/2} \delta_n ,
\end{align*}
where $T_{\epsilon}$ is a positive constant that depends only on $\epsilon$ and $T$.

(3) Bounds for $\mathcal{J}_3$ and $\mathcal{J}_6$. 
We focus on $\mathcal{J}_3$ since the bound of $\mathcal{J}_6$ can be derived in a similar way.
This part is motivated by the proof of Theorem 2 in \cite{dorn2024dvds}.

Note that
\begin{align*}
    \mathcal{J}_3 & = \mathbb{E}_P\Bigg[
    \frac{1 - \varphi}{\widetilde{\varphi}} \frac{\widetilde{\varphi}(S_i, X_i)}{1-\widetilde{\varphi}(S_i, X_i)}\frac{\rho(S_i, X_i)}{\widetilde{\rho}(X_i)} \mathcal{J}_{3}(S_i, X_i, O) \bigg| 
    P_i = O, \Bigg]. 
\end{align*}
where
\begin{align*}
   \mathcal{J}_{3}(S_i, X_i, O) & = \bigg[\left(\widetilde{q}_{C_+}(S_i, X_i, O) + \frac{\mathbb{E}_P[[Y_i - \widetilde{q}_{C_+}(S_i, X_i, O)]_{+}|S_i, X_i, P_i = O]}{\rho(S_i, X_i)} \right) \\
    & \quad \quad \quad - \left(q_{C_+}(S_i, X_i, O) + \frac{\mathbb{E}_P[[Y_i - q_{C_+}(S_i, X_i, O)]_{+})|S_i, X_i, P_i = O,]}{\rho(S_i, X_i)}\right) \bigg]  
\end{align*}

(Continuous Case:) Lemma 2 in \cite{dorn2024dvds} and $\mathbb{P}(\epsilon \le \rho(S_i, X_i) \le 1  - \epsilon)$ imply:
\begin{align*}
    |\mathcal{J}_3| \le T_{\epsilon} \| \widetilde{q}_{C_+}(S_i, X_i, O) - q_{C_+}(S_i, X_i, O) \|_{P, 2}^2 \le T_{\epsilon} n^{-1/2} \delta_n,
\end{align*}
where $T_{\epsilon}$ depends on $\epsilon$ and $T$ only.

(Binary Case:) 
Note that
\begin{align*}
    q_{C_+}(S_i, X_i, O) = 
    \begin{cases}
        1 & \text{ if } \mu(S_i, X_i, O) > \rho(S_i, X_i), \\
        0 & \text{ otherwise }
    \end{cases}
\end{align*}

The discussion in Section C.8.2 in the supplementary material of \cite{dorn2024dvds} implies that
\begin{align*}
    & \bigg|\left(\widetilde{q}_{C_+}(S_i, X_i, O) + \frac{\mathbb{E}_P[[Y_i - \widetilde{q}_{C_+}(S_i, X_i, O)]_{+}|S_i, X_i, P_i = O]}{\rho(S_i, X_i)} \right) \\
    & \quad - \left(q_{C_+}(S_i, X_i, O) + \frac{\mathbb{E}_P[[Y_i - q_{C_+}(S_i, X_i, O)]_{+})|S_i, X_i, P_i = O]}{\rho(S_i, X_i)}\right) \Bigg| \\
    & = \bigg|1 - \frac{\mu(S_i, X_i, O)}{\rho(S_i, X_i)}\bigg| \mathds{1}(\mu(S_i, X_i, O) \le \rho(S_i, X_i) < \widetilde{\mu}(S_i, X_i, O) \text{ or } \widetilde{\mu}(S_i, X_i, O) \le \rho(S_i, X_i) < \mu(S_i, X_i, O))
\end{align*}
If $\mu(S_i, X_i, O) \le \rho(S_i, X_i) \le \widetilde{\mu}(S_i, X_i, O) \text{ or } \widetilde{\mu}(S_i, X_i, O) \le \rho(S_i, X_i) \le \mu(S_i, X_i, O)$, then $|\mu(S_i, X_i, O) - \rho(S_i, X_i)| \le \| \widetilde{\mu}(S_i, X_i, O) - \mu(S_i, X_i, O)\|_{\infty}$. Therefore, we have
{\scriptsize
\begin{align*}
    |\mathcal{J}_3(S_i, X_i, O)| 
    & = \bigg|1 - \frac{\mu(S_i, X_i, O)}{\rho(S_i, X_i)}\bigg| \mathds{1}(\mu(S_i, X_i, O) \le \rho(S_i, X_i) \le \widetilde{\mu}(S_i, X_i, O) \text{ or } \widetilde{\mu}(S_i, X_i, O) \le \rho(S_i, X_i) \le \mu(S_i, X_i, O)) \\
    & \le \bigg|1 - \frac{\mu(S_i, X_i, O)}{\rho(S_i, X_i)}\bigg| \mathds{1}(|\mu(S_i, X_i, O) - \rho(S_i, X_i)| \le \| \widetilde{\mu}(S_i, X_i, O) - \mu(S_i, X_i, O)\|_{\infty}) \\
    & \le \frac{1}{\epsilon} | \widetilde{\mu}(S_i, X_i, O) - \mu(S_i, X_i, O)\|_{\infty} \mathds{1}(|\mu(S_i, X_i, O) - \rho(S_i, X_i)| \le \| \widetilde{\mu}(S_i, X_i, O) - \mu(S_i, X_i, O)\|_{\infty}).
\end{align*}
}
This implies that when $\mu(S_i, X_i, O) - \rho(S_i, X_i)$ has the bounded density, we have
{\footnotesize
\begin{align*}
    |\mathcal{J}_3| 
    & \le T_{\epsilon} | \widetilde{\mu}(S_i, X_i, O) - \mu(S_i, X_i, O)\|_{\infty} \mathbb{P}(\mathds{1}(|\mu(S_i, X_i, O) - \rho(S_i, X_i)| \le \| \widetilde{\mu}(S_i, X_i, O) - \mu(S_i, X_i, O)\|_{\infty}) \\
    & \le T_{\epsilon} | \widetilde{\mu}(S_i, X_i, O) - \mu(S_i, X_i, O)\|_{\infty}^2 \\
    & \le T_{\epsilon} n^{-1/2} \delta_n .
\end{align*}
}
Similar reasoning gives the bound for $\mathcal{J}_6$.
Note that 
\begin{align*}
    \mathcal{J}_6 = -
    \mathbb{E}_P\Bigg[\frac{\mathds{1}_{P_i = O}}{\widetilde{\varphi}} \frac{\widetilde{\varphi}(S_i, X_i)}{1-\widetilde{\varphi}(S_i, X_i)}\frac{1-\rho(S_i, X_i)}{1-\widetilde{\rho}(X_i)} \mathcal{J}_{6}(S_i,X_i, O) \bigg| P_i = O\Bigg]
\end{align*}
where
\begin{align*}
    \mathcal{J}_{6}(S_i,X_i, O) 
    & = \bigg[\big(\widetilde{q}_{C_+}(S_i, X_i, O) - \frac{\mathbb{E}_P[[Y_i - \widetilde{q}_{C_+}(S_i, X_i, O)]_{-}|S_i, X_i, P_i = O]}{1-\rho(S_i, X_i)}) \\
    & \quad \quad \quad - \big((1-\rho(S_i, X_i)) q_{C_+}(S_i, X_i, O) - \frac{\mathbb{E}_P[[Y_i - q_{C_+}(S_i, X_i, O)]_{-}|S_i, X_i, P_i = O]}{1-\rho(S_i, X_i)}\big) \bigg]  \Bigg].
\end{align*}
Similar to $|\mathcal{J}_3|$, $|\mathcal{J}_{6}(S_i,X_i, O)|$ have the following.
{\scriptsize
\begin{align*}
    |\mathcal{J}_{6}(S_i,X_i, O)| =
    \left|1 - \frac{1-\mu(S_i, X_i, O)}{1-\rho(S_i, X_i)}  \right| \mathds{1}(\mu(S_i, X_i, O) \le \rho(S_i, X_i) < \widetilde{\mu}(S_i, X_i, O) \text{ or } \widetilde{\mu}(S_i, X_i, O) \le \rho(S_i, X_i) < \mu(S_i, X_i, O)). 
\end{align*}
}
With similar approach, we have
\begin{align*}
    |\mathcal{J}_6| 
    & \le T_{\epsilon} | \widetilde{\mu}(S_i, X_i, O) - \mu(S_i, X_i, O)\|_{\infty}^2 \le T_{\epsilon} n^{-1/2} \delta_n .
\end{align*}
This implies that
\begin{align*}
    |\mathbb{E}_P[m(W_i, \tau_{C_+}, \widetilde{\eta}) - m(W_i, \tau_{C_+}, \eta)]| \le \sum_{j=1}^8 |\mathcal{J}_j| \le T_{\epsilon} n^{-1/2} \delta_n.
\end{align*}
where $T_{\epsilon}$ depends only on $T$ and $\epsilon$.

\textbf{Part 4: Verification of Condition \eqref{eq:pd-variance}.}
{\footnotesize
\begin{align*}
        & \mathbb{E}_P[m_{C_+}(Z_i, \tau_{C_+}, \eta)^2] \\
        & =
        \mathbb{E}_P\Bigg[\bigg(\sum_{j=1}^2 m_{j}(Z_i, \tau_{C_+}, \eta) + m_{3}(Z_i, \tau_{C_+}, \eta) + \sum_{j=6}^7 m_{j}(Z_i, \tau_{C_+}, \eta) \bigg)^2\Bigg] + 
        \mathbb{E}_P\Bigg[\bigg(\sum_{j=4}^5 m_{j}(Z_i, \tau_{C_+}, \eta)\bigg)^2\Bigg] \\
        & = 
        \mathbb{E}_P\Bigg[\bigg(\sum_{j=1}^2 m_{j}(Z_i, \tau_{C_+}, \eta)\bigg)^2\Bigg] + \mathbb{E}\Bigg[\bigg(\sum_{j=6}^7 m_{j}(Z_i, \tau_{C_+}, \eta) \bigg)^2\Bigg] \\
        & \quad + 2 \mathbb{E}\Bigg[\bigg(\sum_{j=1}^2 m_{j}(Z_i, \tau_{C_+}, \eta)\bigg)\bigg(\sum_{j=6}^7 m_{j}(Z_i, \tau_{C_+}, \eta) \bigg)\Bigg] + \mathbb{E}_P[m_{3}^2(Z_i, \tau_{C_+}, \eta)] + \mathbb{E}_P\Bigg[\bigg(\sum_{j=4}^5 m_{j}(Z_i, \tau_{C_+}, \eta)\bigg)^2\Bigg] \\
        & \ge \mathbb{E}_P[m_{3}^2(Z_i, \tau_{C_+}, \eta)] + \mathbb{E}_P\Bigg[\bigg(\sum_{j=4}^5 m_{j}(Z_i, \tau_{C_+}, \eta)\bigg)^2\Bigg] \\
        & \ge t.
    \end{align*}
}
    
\end{proof}

\section{Proofs for Section \ref{sec:est-copula}}

\label{sec:proofs-est-copula}

\subsection{ Proof of Dual Form of $\mu_{C_o, w}$} 
\label{sec:dicussion-dual-form}

\begin{proof}[Proof of \cref{lemma:dual-SpectralRiskMeasure}]
The proof of \cref{lemma:dual-SpectralRiskMeasure} comes from \cite{pichler2015} and Section 2.4.2 of  \cite{pflug2007ModelingMeasuringManagingRIsk} with the dual form of AVaR. The integration by parts can be applied for functions of bounded variation when one of them is continuous. (c.f. Theorem 5.3 of \cite{convertito2023StieltjesIntegral},  Theorem 6.2.2 of \cite{carter2000LebesgueStieltjesIntegral}).  Consequently, their proof works for our case. To be self-contained, we present it below.

Here, it is enough to show 
\begin{align}
    r 
    & = \sigma(0) \int_0^1 F_Y^{-1}(u) du + \int_0^1 \left(\int_{u}^1 F_Y^{-1}(s) ds \right) d \sigma(u) \label{eq:dual-SpectralRiskMeasure-sig0}\\
    & = \sigma(1) \int_0^1 F_Y^{-1}(u) du - \int_0^1 \left(\int_{0}^u F_Y^{-1}(s) ds \right) d \sigma(u). \label{eq:dual-SpectralRiskMeasure-sig1}
\end{align}

First, \cref{eq:dual-SpectralRiskMeasure-sig1} can be derived as follows.
\begin{align*}
    & \sigma(1) \int_0^1 F_Y^{-1}(s) ds - \int_0^1  \left(\int_0^{u} F_Y^{-1}(s) ds \right)  d \sigma(u) \\
    & = \sigma(1) \int_0^1 F_Y^{-1}(s) ds + \int_0^1 \sigma(u) d \left(\int_0^{u} F_Y^{-1}(s) ds \right) - 
        \sigma(1) \int_0^1 F_Y^{-1}(s) ds \\
    & = \int_0^1 F_Y^{-1}(u) \sigma(u) du.
\end{align*}
The first equality comes from the integration by parts with functions of bounded variation where one of them is continuous. (c.f., see Theorem 5.3 of \cite{convertito2023StieltjesIntegral}, Theorem 21.67 and Remark 21.68 of \cite{hewitt1965RealAbstractAnalysis}, or Theorem 6.2.2 of \cite{carter2000LebesgueStieltjesIntegral}) Note that $\sigma(u)$ is of bounded variation and $u \to \int_0^{u} F_Y^{-1}(s) ds$ is absolute continuous on $[0, 1]$.

\cref{eq:dual-SpectralRiskMeasure-sig0} and \cref{eq:dual-SpectralRiskMeasure-sig1} are identical because
\begin{align*}
    & \sigma(1) \int_0^1 F_Y^{-1}(s) ds - \int_0^1  \left(\int_0^{u} F_Y^{-1}(s) ds \right)  d \sigma(u) \\
    & =
    \sigma(1) \int_0^1 F_Y^{-1}(s) ds - \int_0^1  \left(\int_0^1 F_Y^{-1}(s) ds - \int_{u}^1 F_Y^{-1}(s) ds \right)  d \sigma(u) \\
    & = \sigma(1) \int_0^1 F_Y^{-1}(s) ds - (\sigma(1) - \sigma(0)) \int_0^1 F_Y^{-1}(s) ds + \int_0^1  \left(\int_{u}^{1} F_Y^{-1}(s) ds \right)  d \sigma(u) \\
    & = \sigma(0) \int_0^1 F_Y^{-1}(s) ds + \int_0^1  \left(\int_{u}^{1} F_Y^{-1}(s) ds \right)  d \sigma(u).
\end{align*}
It concludes the proof.
\end{proof}

\subsection{Main Proof for \cref{thm:CLT-for-copula}}

From \Cref{lemma:DML-asym-normality}, it is sufficient to verify that conditions in \Cref{lemma:DML-asym-normality} holds for $m_{C_o}(W_i, \tau_{C_o}, \eta)$. Note that we can write $m_{C_o}(W_i, \tau_{C_o}, \eta)$ in the following affine form.
\begin{align*}
    m_{C_o}(W_i, \tau, \eta) = m_{C_o}^a(Z_i, \eta) \tau + m_{C_o}^b(Z_i, \eta), \text{ where }
    m_{C_o}^a(Z_i, \eta) = 
    \mathds{1}_{P_i = E} / \varphi.
\end{align*}
\begin{lemma} \label{lemma:DML-verification-copula}
    Under Assumptions \ref{assumption: Random Sample}, \ref{assumption: Unconfounded}, \ref{assumption:Comparability}, \ref{assumption:knowncopula}, \ref{assumption:boundedvariation}, \ref{assumption:regularity-wc} (i), \ref{assumption:regularity-copula}, and  \ref{assumption:realization-set-copula}, Conditions in \Cref{lemma:DML-asym-normality} hold. In addition, $\mathbb{E}[m_{C_o}^2(Z_i, \tau_{C_o}, \eta)] > c_0$
    for some positive constant $c_0$ which only depends on $c$ and $\epsilon$.
\end{lemma}

The proof of \Cref{lemma:DML-verification-copula} is almost identical to the proof of \cref{lemma:DML-verification-wc}. The proof will be presented in \Cref{sec:supp-proofs-asym-copula} of the supplementary appendix.

\section{Cross-Fitting Algorithm}
\begin{algorithm}[H]
\small
\caption{Estimation of worst-case bounds for long-term treatment effect}
\label{alg:debiased}
    \begin{algorithmic}[1]
    \STATE\textbf{Input:} a \( K \)-fold random partition of the dataset $\{P_i, X_i, S_i,\mathds{1}_{P_i=E}W_i, \mathds{1}_{P_i=O}Y_i\}_{i=1}^n$, denoted as \( \cup_{k=1}^K \mathcal{F}_k \), where \( |\mathcal{F}_k| = n/K \).
   \FOR{$k\in[K]$}

   \STATE Construct $\widehat \rho^{-k(i)}(s,x)$ and $\widehat \rho^{-k(i)}(x)$ using $\{(X_i, S_i, W_i): i \in \mathcal{F}_k^c \land P_i=E\}$;
   \STATE Construct $\widehat \varphi^{-k(i)}(s,x)$ and $\widehat \varphi^{-k(i)}(x)$ using $\{(X_i, S_i, P_i): i \in \mathcal{F}_k^c\}$.
   \STATE For $\{(X_i, S_i): i \in \mathcal{F}_k^c \land P_i=E\}$, compute $\widehat{\mu}_{C_-, 0}(S_i, X_i, O)$, $\widehat{\mu}_{C_-, 1}(S_i, X_i, O)$, $\widehat{\mu}_{C_+, 0}(S_i, X_i, O)$, and $\widehat{\mu}_{C_+, 1}(S_i, X_i, O)$ using the two-stage estimator \citep{olma2021}:
   \begin{enumerate}
   [label=\arabic*), labelsep=0pt]
   \vspace{-0.2cm}
       \item Compute $\widehat\rho(S_i, X_i)$ based on leave-one-out within $\{(X_i, S_i, Y_i): i \in \mathcal{F}_k^c \land P_i=E\}$, and compute $\widehat F_O^{-1}(\cdot)$ based on leave-one-out within $\{(X_i, S_i, Y_i): i \in \mathcal{I}_k^c \land P_i=O\}$ using quantile forests.
       \item Construct a linear sieve model with the pseudo-outcome defined in Equation (3) of \cite{olma2021} and predict $\widehat{\mu}_{C_-, 0}(S_i, X_i, O)$, $\widehat{\mu}_{C_-, 1}(S_i, X_i, O)$, $\widehat{\mu}_{C_+, 0}(S_i, X_i, O)$, and $\widehat{\mu}_{C_+, 1}(S_i, X_i, O)$.
   \end{enumerate}
   \STATE Construct $\widehat{\bar{\mu}}^{-k(i)}_{C_-, 0}(s,x), \widehat{\bar{\mu}}_{C_-, 1}^{-k(i)}(s,x), \widehat{\bar{\mu}}_{C_+, 0}^{-k(i)}(s,x)$, and $\widehat{\bar{\mu}}_{C_+, 1}^{-k(i)}(s,x)$ using $\{(X_i, S_i): i \in \mathcal{F}_k^c \land P_i=E\}$ and the pseudo-outcomes from line 5.
   \FOR{$i\in \mathcal{F}_k$}
   \STATE Evaluate $\widehat \rho^{-k(i)}(S_i, X_i)$, $\widehat \rho^{-k(i)}(X_i)$, $\widehat \varphi^{-k(i)}(S_i,X_i)$, $\widehat \varphi^{-k(i)}(X_i)$, $\widehat{\varphi}^{-k(i)}=\frac{K|\{i\in\mathcal{F}_k^c: P_i=E\}|}{n(K-1)}$, $\widehat{\bar{\mu}}^{-k(i)}_{C_-, 0}(S_i, X_i)$, $\widehat{\bar{\mu}}^{-k(i)}_{C_-, 1}(S_i, X_i)$, $\widehat{\bar{\mu}}^{-k(i)}_{C_+, 0}(S_i, X_i)$, and $\widehat{\bar{\mu}}^{-k(i)}_{C_+, 1}(S_i, X_i)$, and compute
   \vspace{-0.2cm}
   \begin{enumerate}[label=\arabic*), labelsep=0pt]
   \setlength\itemsep{0.2em}
       \item $\widehat{\mu}_{C_-, 0}(S_i, X_i, O)$, $\widehat{\mu}_{C_-, 1}(S_i, X_i, O)$, $\widehat{\mu}_{C_+, 0}(S_i, X_i, O)$, and $\widehat{\mu}_{C_+, 1}(S_i, X_i, O)$ using the two-stage estimator \citep{olma2021}, where $\widehat \rho^{-k(i)}(S_i, X_i)$ is used, and $\widehat F_O^{-1}(\cdot)$ is estimated based on $\{(X_i, S_i, Y_i): i \in \mathcal{F}_k^c \land P_i=O\}$;
       \item $\widehat q_{C_+}(S_i, X_i, O)$ and $\widehat q_{C_-}(S_i, X_i, O)$ using quantile forests based on $\{(X_i, S_i, Y_i): i \in \mathcal{F}_k^c \land P_i=O\}$, where $\widehat \rho^{-k(i)}(S_i, X_i)$ is used;
   \end{enumerate}
   \ENDFOR
   \ENDFOR
   \STATE Compute $\widehat{\tau}_{C_-}$ and $\widehat{\tau}_{C_+}$ as closed-form solutions to $\frac{1}{n}\sum_{i=1}^n \widehat m_{C_-}(Z_i, \tau, \widehat{\eta})=0$ and $\frac{1}{n}\sum_{i=1}^n \widehat m_{C_+}(Z_i, \tau, \widehat{\eta})=0$, respectively.
   For $i\in[n]$, recompute vectors $\Gamma_{C_-}:=\widehat m_{C_-}(Z_i, \widehat\tau_{C_-}, \widehat\eta)$ and $\Gamma_{C_+}:=\widehat m_{C_+}(Z_i, \widehat\tau_{C_+}, \widehat\eta)$. 
   \STATE Return $(\widehat\tau_{C_-}, \widehat\tau_{C_+})$ and $\left(\left[\widehat\tau_{C_-}\pm \Phi^{-1}((1+\gamma)/2)\widehat{\mathrm{se}}_{C_-}\right], \left[\widehat\tau_{C_+}\pm \Phi^{-1}((1+\gamma)/2)\widehat{\mathrm{se}}_{C_+}\right]\right)$ as $\gamma$-CIs, where $\widehat{\mathrm{se}}_{C_-}:=\sqrt{\frac{1}{n(n-1)}\sum_{i=1}^n\Gamma_{C_-,i}^2}$ and $\widehat{\mathrm{se}}_{C_+}:=\sqrt{\frac{1}{n(n-1)}\sum_{i=1}^n\Gamma_{C_+,i}^2}$.
  \end{algorithmic}
\end{algorithm}


\newpage

\setcounter{equation}{1} 
\renewcommand{\theequation}{S\arabic{equation}} 
\setcounter{figure}{1} 
\renewcommand{\thefigure}{S\arabic{figure}}

\renewcommand{\thesubsection}{ S.\arabic{subsection}}

\counterwithout{lemma}{section}
\counterwithin{lemma}{subsection}

\section*{Supplement to ``A Sensitivity Analysis of the Surrogate Index Approach for Estimating Long-Term Treatment Effects''}

This manuscript presents the detailed explanations on numerical simulation, and the detailed calculation used in the proofs in the main text.

\subsection{Common Copula Families, Kendall's Tau, and Mathematical Details for ATE Computation}
\label{sec:numerical details}
This section supplements Section \ref{sec:numerical} by collecting explicit formulas for the copula families used: the Archimedean generators, the one-to-one relationship between Kendall's tau $\varrho_K$ and each family's parameter $\vartheta$, and the numerical steps for computing the ATE $\tau_{C_o}$.
\begin{itemize}
    \item Gaussian Copula: The copula parameter $\vartheta$ is the linear correlation coefficient of the underlying normal distribution. Kendall's tau is given by:
    \[
    \varrho_K = \frac{2}{\pi} \arcsin(\vartheta), \ \text{so}\ \vartheta = \sin\left(\frac{\pi}{2} \varrho_K\right).
    \]
    \item Clayton Copula: The generator function $\phi(t) = (1 + t)^{-1/\vartheta}$ leads to:
    \[
    \varrho_K = \frac{\vartheta}{\vartheta + 2}, \ \text{so}\ \vartheta = \frac{2\varrho_K}{1 - \varrho_K}.
    \]
    \item Gumbel Copula: $\phi(t) = \exp(-t^{1/\vartheta})$ yields:
    \[
    \varrho_K = 1 - \frac{1}{\vartheta}, \ \text{so}\ \vartheta = \frac{1}{1 - \varrho_K}.
    \]
    \item Frank Copula: With $\phi(t) = -\frac{1}{\vartheta} \log\left(1 - (1 - e^{-\vartheta}) e^{-t}\right)$,
    \[
    \varrho_K = 1 - \frac{4}{\vartheta} \left(1 - D_1(\vartheta)\right),
    \]
    where $D_1(\vartheta)$ is the Debye function:
    \[
    D_1(\vartheta) = \frac{1}{\vartheta} \int_0^\vartheta \frac{t}{e^t - 1} \, dt.
    \]
    Solving for $\vartheta$ given $\varrho_K$ requires numerical root-finding, as the relationship is not analytically invertible.
\end{itemize}

In Figure \ref{Figure: illustrative Archimedean conditional}, we approximate and plot how 
    \begin{align*}
        \tau_{C_o}
        &=
        \mathbb{E}\left[ \frac{\rho(X_i,S_i)}{\rho(X_i)[1-\rho(X_i)]}   \mu_{C_o, 1}(S_i,X_i,O)  \mid P_i = E \right] -  \mathbb{E}\left[ \frac{1}{1-\rho(X_i)} \mu(S_i,X_i,O) \mid P_i = E \right]
    \end{align*}
    changes with $\varrho_K$. Note that in the DGP in Section \ref{sec:numerical}$, \rho(X_i)=\rho$, and $\rho(S_i, X_i)=\mathbb{P}(W = 1 | S, X)$ can be derived as follows: Using Bayes' rule,
    \begin{equation}
    \label{eq:bayes}
        \mathbb{P}(W = 1 | S, X) = \frac{\mathbb{P}(S | W = 1, X) \mathbb{P}(W = 1 | X)}{\mathbb{P}(S | W = 1, X) \mathbb{P}(W = 1 | X) + \mathbb{P}(S | W = 0, X) \mathbb{P}(W = 0 | X)}.
    \end{equation}
Since $S | W = w, X \sim \mathcal{N}(X + w, 1)$,
the conditional densities are:
\[
\mathbb{P}(S | W = 1, X) = \frac{1}{\sqrt{2\pi}} \exp \left( -\frac{(S - (X+1))^2}{2} \right),
\]
\[
\text{and }\mathbb{P}(S | W = 0, X) = \frac{1}{\sqrt{2\pi}} \exp \left( -\frac{(S - X)^2}{2} \right).
\]
Since \( W \sim \text{Bernoulli}(\rho) \), we have: $\mathbb{P}(W = 1 | X) = \rho$ and $\mathbb{P}(W = 0 | X) = 1-\rho$.
Thus, \eqref{eq:bayes} simplifies to:
\begin{align*}
    \mathbb{P}(W = 1 | S, X) &= \frac{\rho\cdot\exp \left( -\frac{(S - (X+1))^2}{2} \right)}{\rho\cdot\exp \left( -\frac{(S - (X+1))^2}{2} \right) + (1-\rho)\cdot\exp \left( -\frac{(S - X)^2}{2} \right)} \\
    &= \frac{1}{1 + ((1-\rho)/\rho)\cdot\exp \left( -\left[ (S - X)^2 - (S - (X+1))^2 \right]/2 \right)} \\
    &= \frac{1}{1 + ((1-\rho)/\rho)\cdot e^{-(S - X - 0.5)}}.
\end{align*}
To approximate the expectation over $(S_i, X_i)$, we need the joint density of $(S_i,X_i)$:
    \[f_{S, X}(s, x) = f_X(x) \cdot f_{S|X}(s | x) =
\begin{cases}
(1-\rho) \cdot \phi(s - x) + \rho\cdot \phi(s - x - 1), & 0 \leq x \leq 1, \\
0, & \text{otherwise},
\end{cases}\]
where \( \phi(z) \) is the probability density function of the standard normal distribution. We truncate the range of $S$ to $(-4, 6)$, since the joint density of $(S,X)$ outside this range of $S$ is negligible. Then, 
\begin{align*}
     \tau_{C_o}
        &=
        \mathbb{E}\left[ \frac{\rho(X_i,S_i)}{\rho(X_i)[1-\rho(X_i)]}   \mu_{C_o, 1}(S_i,X_i,O)\right] -  \mathbb{E}\left[ \frac{1}{1-\rho(X_i)} \mu(S_i,X_i,O)\right]\\
        &\approx \int_0^1\int_{-4}^{6} \frac{\rho(x,s)}{\rho(1-\rho)} \left(\int_{0}^{1} F_Y^{-1}(u|s,x)\sigma_{C_o, 1}(u; 1- \rho(x,s)) d u\right)  f_{S, X}(s, x)dsdx\\
        &\quad- \int_0^1\int_{-4}^{6} \frac{\mu(s,x)}{1-\rho}  f_{S, X}(s, x)dsdx,
\end{align*}
with $F_Y^{-1}(u|s,x) = s+0.5x+\Phi^{-1}(u)$ and $\mu(s,x)=s+0.5x$ based on the DGP. Inside $\sigma_{C_o, 1}(u; 1- \rho(x,s))$, $C_o(1- \rho(x,s)|u)$ is calculated using the analytical form of the conditional copula for each of the four copula families:
\[
C_o(1-\rho(x, s) \mid u) = \frac{\partial C_o(u, 1-\rho(x, s))}{\partial u}.
\]
To compute \( \tau_{C_o} \), we use \texttt{adaptIntegrate} in \texttt{R} twice—first to evaluate $\mu_{C_o, 1}(S_i,X_i,O)$ for a given \( (S_i, X_i) \), and then to approximate the outer expectation over all \( (S_i, X_i) \).

\subsection{Technical Lemmas for \Cref{lemma:DML-verification-wc}}

We remind you of notation used in the proof of \Cref{lemma:DML-verification-wc}. Let $m_{C_+}^a(Z_i, \eta) = \mathds{1}_{P_i = E} / \varphi$, and
    {\footnotesize
\begin{align*}
    m_1(Z_i, \tau_{C_+}, \eta) & = \frac{\mathds{1}_{P_i = E}}{\varphi} \frac{W_i}{\rho(X_i)} (\mu_{C_+, 1}(S_i, X_i, O) - \bar{\mu}_{C_+, 1}(1, X_i)), \\
    m_2(Z_i, \tau_{C_+}, \eta) & = - \frac{\mathds{1}_{P_i = E}}{\varphi} \frac{1-W_i}{1-\rho(X_i)} (\mu_{C_+, 0}(S_i, X_i, O) - \bar{\mu}_{C_+, 0}(0, X_i)), \\
    m_3(Z_i, \tau_{C_+}, \eta) & =  \frac{\mathds{1}_{P_i = E}}{\varphi} (\bar{\mu}_{C_+, 1}(1, X_i) - \bar{\mu}_{C_+, 0}(0, X_i) - \tau_{C_+}), \\
    m_4(Z_i, \tau_{C_+}, \eta) & = \frac{\mathds{1}_{P_i = O}}{\varphi} \frac{\varphi(S_i, X_i)}{1-\varphi(S_i, X_i)}\frac{\rho(S_i, X_i)}{\rho(X_i)} (H_U(Y_i, q_{C_+}(S_i, X_i, O), \rho(S_i, X_i)) -\mu_{C_+, 1}(S_i, X_i, O)), \\
    m_5(Z_i, \tau_{C_+}, \eta) & = - \frac{\mathds{1}_{P_i = O}}{\varphi} \frac{\varphi(S_i, X_i)}{1-\varphi(S_i, X_i)} 
    \frac{1-\rho(S_i, X_i)}{1-\rho(X_i)} (H_L(Y_i, q_{C_+}(S_i, X_i, O), 1 - \rho(S_i, X_i)) - \mu_{C_+, 0}(S_i, X_i, O)), \\
    m_6(Z_i, \tau_{C_+}, \eta) & = \frac{\mathds{1}_{P_i = E}}{\varphi} \frac{1}{\rho(X_i)} \left[q_{C_+}(S_i, X_i, O) - \mu_{C_+, 1}(S_i, X_i, O)\right] \left(W_i - \rho(S_i, X_i)\right), \\
    m_7(Z_i, \tau_{C_+}, \eta) & =
    \frac{\mathds{1}_{P_i = E}}{\varphi} \frac{1}{1-\rho(X_i)} \left[q_{C_+}(S_i, X_i, O) - \mu_{C_+, 0}(S_i, X_i, O)\right] \left(W_i - \rho(S_i, X_i)\right).
\end{align*}
}
Note that 
\begin{align*}
    m_{C_+}(Z_i, \tau_{C_+}, \eta) = \sum_{j=1}^7 m_j(Z_i, \tau_{C_+}, \eta).
\end{align*}

We will show the following technical lemma.
\begin{lemma} 
    \label{lemma:tech-remainer-wc}
    \begin{align*}
    \mathbb{E}_P[m(W_i, \tau_{C_+}, \widetilde{\eta}) - m(W_i, \tau_{C_+}, \eta)] 
    & = 
    \sum_{j=1}^7 \mathbb{E}_P[m_j(W_i, \tau_{C_+}, \widetilde{\eta}) - m_j(W_i, \tau_{C_+}, \eta)] = \sum_{j=1}^8 \mathcal{J}_j,  
\end{align*} 
where
{\footnotesize
\begin{align*}
    \mathcal{J}_1 & = 
    - \mathbb{E}_P\Bigg[\bigg(\frac{\mathds{1}_{P_i = E}}{\widetilde{\varphi}} \frac{W_i}{\widetilde{\rho}(X_i)} - \frac{\mathds{1}_{P_i = E}}{\widetilde{\varphi}} \frac{W_i}{\rho(X_i)} \bigg)(\widetilde{\bar{\mu}}_{C_+, 1}(1, X_i) - \bar{\mu}_{C_+, 1}(1, X_i)) \Bigg] \\
    \mathcal{J}_2 & = 
    \mathbb{E}_P\Bigg[\bigg(\frac{\mathds{1}_{P_i = E}}{\widetilde{\varphi}} \frac{1-W_i}{1-\widetilde{\rho}(X_i)} - \frac{\mathds{1}_{P_i = E}}{\widetilde{\varphi}} \frac{1-W_i}{1-\rho(X_i)} \bigg) (\widetilde{\bar{\mu}}_{C_+, 0}(1, X_i) - \bar{\mu}_{C_+, 0}(1, X_i)) \Bigg] \\
    \mathcal{J}_3 & =
    \mathbb{E}_P\Bigg[
    \frac{\mathds{1}_{P_i = O}}{\widetilde{\varphi}} \frac{\widetilde{\varphi}(S_i, X_i)}{1-\widetilde{\varphi}(S_i, X_i)}\frac{1}{\widetilde{\rho}(X_i)} \\
    & \quad \quad \times \bigg[(\rho(S_i, X_i) \widetilde{q}_{C_+}(S_i, X_i, O) + [Y_i - \widetilde{q}_{C_+}(S_i, X_i, O)]_{+}) \\
    & \quad \quad \quad - (\rho(S_i, X_i) q_{C_+}(S_i, X_i, O) + [Y_i - q_{C_+}(S_i, X_i, O)]_{+}) \bigg]  \Bigg] \\
    \mathcal{J}_4 & =
    -
    \mathbb{E}_P\Bigg[
    \left(\frac{\mathds{1}_{P_i = O}}{\widetilde{\varphi}} \frac{\widetilde{\varphi}(S_i, X_i)}{1-\widetilde{\varphi}(S_i, X_i)}\frac{1}{\widetilde{\rho}(X_i)} - \frac{\mathds{1}_{P_i = O}}{\widetilde{\varphi}} \frac{\varphi(S_i, X_i)}{1-\varphi(S_i, X_i)}\frac{1}{\widetilde{\rho}(X_i)}\right)
    \\
    & \quad \quad \times \bigg[\rho(S_i, X_i) (\widetilde{\mu}_{C_+, 1}(S_i, X_i, O)) - \mu_{C_+, 1}(S_i, X_i, O)) \bigg]  \Bigg] \\
    \mathcal{J}_5 & =
    \mathbb{E}_P\Bigg[
    \left(\frac{\mathds{1}_{P_i = O}}{\widetilde{\varphi}} \frac{\widetilde{\varphi}(S_i, X_i)}{1-\widetilde{\varphi}(S_i, X_i)}\frac{1}{\widetilde{\rho}(X_i)} - 
    \frac{\mathds{1}_{P_i = O}}{\widetilde{\varphi}} \frac{\varphi(S_i, X_i)}{1-\varphi(S_i, X_i)}\frac{1}{\widetilde{\rho}(X_i)} \right)\\
    & \quad \quad \times \bigg[(\widetilde{\rho}(S_i, X_i) - \rho(S_i, X_i)) (\widetilde{q}_{C_+}(S_i, X_i, O) -\widetilde{\mu}_{C_+, 1}(S_i, X_i, O))
     \Bigg] \\
    \mathcal{J}_6 & =
    -
    \mathbb{E}_P\Bigg[
    \frac{\mathds{1}_{P_i = O}}{\widetilde{\varphi}} \frac{\widetilde{\varphi}(S_i, X_i)}{1-\widetilde{\varphi}(S_i, X_i)}\frac{1}{1-\widetilde{\rho}(X_i)} \\
    & \quad \quad \times \bigg[\big((1-\rho(S_i, X_i)) \widetilde{q}_{C_+}(S_i, X_i, O) - [Y_i - \widetilde{q}_{C_+}(S_i, X_i, O)]_{-}) \\
    & \quad \quad \quad - \big((1-\rho(S_i, X_i)) q_{C_+}(S_i, X_i, O) - [Y_i - q_{C_+}(S_i, X_i, O)]_{-}\big) \bigg] \bigg| 
     \Bigg] \\
    \mathcal{J}_7 & =
    \mathbb{E}_P\Bigg[
    \left(\frac{\mathds{1}_{P_i = O}}{\widetilde{\varphi}} \frac{\widetilde{\varphi}(S_i, X_i)}{1-\widetilde{\varphi}(S_i, X_i)}\frac{1}{1-\widetilde{\rho}(X_i)} - \frac{\mathds{1}_{P_i = O}}{\widetilde{\varphi}} \frac{\varphi(S_i, X_i)}{1-\varphi(S_i, X_i)}\frac{1}{1-\widetilde{\rho}(X_i)}\right)
    \\
    & \quad \quad \times \bigg[(1-\rho(S_i, X_i)) (\widetilde{\mu}_{C_+, 0}(S_i, X_i, O)) - \mu_{C_+, 0}(S_i, X_i, O)) \bigg]  \Bigg] \\
    \mathcal{J}_8 & =
    \mathbb{E}_P\Bigg[
    \left(\frac{\mathds{1}_{P_i = O}}{\widetilde{\varphi}} \frac{\widetilde{\varphi}(S_i, X_i)}{1-\widetilde{\varphi}(S_i, X_i)}\frac{1}{1-\widetilde{\rho}(X_i)} - 
    \frac{\mathds{1}_{P_i = O}}{\widetilde{\varphi}} \frac{\varphi(S_i, X_i)}{1-\varphi(S_i, X_i)}\frac{1}{1-\widetilde{\rho}(X_i)} \right)\\
    & \quad \quad \times \bigg[(\widetilde{\rho}(S_i, X_i) - \rho(S_i, X_i)) (\widetilde{q}_{C_+}(S_i, X_i, O) -\widetilde{\mu}_{C_+, 0}(S_i, X_i, O))
     \Bigg] \\
    \end{align*}
}
\end{lemma}

\begin{proof}[Proof of \Cref{lemma:tech-remainer-wc}]
    We show the result by computing $\mathbb{E}_P[m_j(W_i, \tau, \widetilde{\eta}) - m_j(W_i, \tau, \eta)]$ for $j = 1, \dotsc, 7$. 

    \textbf{Part 1: Calculation of $\mathbb{E}_P[m_j(W_i, \tau, \widetilde{\eta}) - m_j(W_i, \tau, \eta)]$ for $j = 1, \dotsc, 7$.}

(1) $\mathbb{E}_P[m_j(W_i, \tau_{C_+}, \widetilde{\eta}) - m_j(W_i, \tau_{C_+}, \eta)]$ for $j = 1, 2$.

\begin{align*}
    & \mathbb{E}_P[m_1(W_i, \tau_{C_+}, \widetilde{\eta}) - m_1(W_i, \tau_{C_+}, \eta)] \\
    & = 
    \mathbb{E}_P\Bigg[\frac{\mathds{1}_{P_i = E}}{\widetilde{\varphi}} \frac{W_i}{\widetilde{\rho}(X_i)} (\widetilde{\mu}_{C_+, 1}(S_i, X_i, O) - \widetilde{\bar{\mu}}_{C_+, 1}(1, X_i)) - \frac{\mathds{1}_{P_i = E}}{\varphi} \frac{W_i}{\rho(X_i)} (\mu_{C_+, 1}(S_i, X_i, O) - \bar{\mu}_{C_+, 1}(1, X_i))\Bigg] \\
    & = 
    \mathbb{E}_P\Bigg[\frac{\mathds{1}_{P_i = E}}{\widetilde{\varphi}} \frac{W_i}{\widetilde{\rho}(X_i)} \Big[(\widetilde{\mu}_{C_+, 1}(S_i, X_i, O) - \widetilde{\bar{\mu}}_{C_+, 1}(1, X_i)) - (\mu_{C_+, 1}(S_i, X_i, O) - \bar{\mu}_{C_+, 1}(1, X_i))\Big] \Bigg] \\
    & \quad + 
    \underbrace{\mathbb{E}_P\Bigg[\mathds{1}_{P_i = E}
    \bigg(
    \frac{1}{\widetilde{\varphi}} \frac{W_i}{\widetilde{\rho}(X_i)} - \frac{1}{\varphi} \frac{W_i}{\rho(X_i)} \bigg) (\mu_{C_+, 1}(S_i, X_i, O) - \bar{\mu}_{C_+, 1}(1, X_i)) 
    \Bigg]}_{=0 \text{ by the definition of $\bar{\mu}_{C_+, 1}(1, X_i)$.} } \\
    & = 
    \mathbb{E}_P\Bigg[\frac{\mathds{1}_{P_i = E}}{\widetilde{\varphi}} \frac{W_i}{\widetilde{\rho}(X_i)} (\widetilde{\mu}_{C_+, 1}(S_i, X_i, O) - \mu_{C_+, 1}(S_i, X_i, O)) \Bigg] \\
    & \quad
    -
    \mathbb{E}_P\Bigg[\frac{\mathds{1}_{P_i = E}}{\widetilde{\varphi}} \frac{W_i}{\widetilde{\rho}(X_i)} (\widetilde{\bar{\mu}}_{C_+, 1}(1, X_i)) -  \bar{\mu}_{C_+, 1}(1, X_i))\Bigg] \\
    & = 
    \mathbb{E}_P\Bigg[\frac{\mathds{1}_{P_i = E}}{\widetilde{\varphi}} \frac{W_i}{\widetilde{\rho}(X_i)} (\widetilde{\mu}_{C_+, 1}(S_i, X_i, O) - \mu_{C_+, 1}(S_i, X_i, O)) \Bigg] \\
    & \quad
    -
    \mathbb{E}_P\Bigg[\frac{\mathds{1}_{P_i = E}}{\widetilde{\varphi}} \frac{W_i}{\rho(X_i)} (\widetilde{\bar{\mu}}_{C_+, 1}(1, X_i)) -  \bar{\mu}_{C_+, 1}(1, X_i))\Bigg] \\
    &\quad -
    \mathbb{E}_P\Bigg[\bigg(\frac{\mathds{1}_{P_i = E}}{\widetilde{\varphi}} \frac{W_i}{\widetilde{\rho}(X_i)} - \frac{\mathds{1}_{P_i = E}}{\widetilde{\varphi}} \frac{W_i}{\rho(X_i)} \bigg)(\widetilde{\bar{\mu}}_{C_+, 1}(1, X_i) - \bar{\mu}_{C_+, 1}(1, X_i)) \Bigg].
\end{align*}
The first two terms will be canceled out when we sum all of $\mathbb{E}_P[m_j(W_i, \tau, \hat{\eta}_{k}) - m_j(W_i, \tau, \eta)]$ for all $j = 1, \dotsc, 7$. Therefore, the first term will be more important.
 
Similarly, we have the following result.
\begin{align*}
    & \mathbb{E}_P[m_2(W_i, \tau_{C_+}, \widetilde{\eta}) - m_2(W_i, \tau_{C_+}, \eta)] \\
    & =
    -
    \mathbb{E}_P\Bigg[\frac{\mathds{1}_{P_i = E}}{\widetilde{\varphi}} \frac{1-W_i}{1-\widetilde{\rho}(X_i)} (\widetilde{\mu}_{C_+, 0}(S_i, X_i, O) - \mu_{C_+, 0}(S_i, X_i, O)) \Bigg]
    \\
    & \quad +
    \mathbb{E}_P\Bigg[\frac{\mathds{1}_{P_i = E}}{\widetilde{\varphi}} \frac{1-W_i}{1-\rho(X_i)} (\widetilde{\bar{\mu}}_{C_+, 0}(0, X_i) - \bar{\mu}_{C_+, 0}(0, X_i)) \Bigg] \\
    & \quad +
    \mathbb{E}_P\Bigg[\bigg(\frac{\mathds{1}_{P_i = E}}{\widetilde{\varphi}} \frac{1-W_i}{1-\widetilde{\rho}(X_i)} - \frac{\mathds{1}_{P_i = E}}{\widetilde{\varphi}} \frac{1-W_i}{1-\rho(X_i)} \bigg) (\widetilde{\bar{\mu}}_{C_+, 0}(0, X_i) - \bar{\mu}_{C_+, 0}(0, X_i)) \Bigg] \\
\end{align*}

(2) $\mathbb{E}_P[m_j(W_i, \tau_{C_+}, \tilde{\eta}) - m_j(W_i, \tau_{C_+} \eta)]$ for $j = 3$.

\begin{align*}
    & \mathbb{E}_P[m_3(W_i, \tau_{C_+}, \widetilde{\eta}) - m_3(W_i, \tau_{C_+}, \eta)] \\
    & =
    \mathbb{E}_P\left[\frac{\mathds{1}_{P_i = E}}{\widetilde{\varphi}} (\widetilde{\bar{\mu}}_{C_+, 1}(1, X_i) - \widetilde{\bar{\mu}}_{C_+, 0}(0, X_i) - \tau_{C_+}) - \frac{\mathds{1}_{P_i = E}}{\varphi} (\bar{\mu}_{C_+, 1}(1, X_i) - \bar{\mu}_{C_+, 0}(0, X_i) - \tau_{C_+}) \right] \\
    & = 
    \mathbb{E}_P\left[\frac{\mathds{1}_{P_i = E}}{\widetilde{\varphi}} \left[(\widetilde{\bar{\mu}}_{C_+, 1}(1, X_i) - \widetilde{\bar{\mu}}_{C_+, 0}(0, X_i)) - (\bar{\mu}_{C_+, 1}(1, X_i) - \bar{\mu}_{C_+, 0}(0, X_i) )\right] \right] \\
    & \quad +
    \underbrace{\mathbb{E}_P\left[\left(\frac{\mathds{1}_{P_i = E}}{\widetilde{\varphi}}  
    - 
    \frac{\mathds{1}_{P_i = E}}{\varphi} \right) (\bar{\mu}_{C_+, 1}(1, X_i) - \bar{\mu}_{C_+, 0}(0, X_i) - \tau_{C_+}) \right]}_{
    =0 \text{, which comes from the definition of $\tau_{C_+}$, $\bar{\mu}_{C_+, 1}$, and $\bar{\mu}_{C_+, 1}$.} } \\
    & =
    \mathbb{E}_P\left[\frac{\mathds{1}_{P_i = E}}{\widetilde{\varphi}} \left[(\widetilde{\bar{\mu}}_{C_+, 1}(1, X_i) - \widetilde{\bar{\mu}}_{C_+, 0}(0, X_i)) - (\bar{\mu}_{C_+, 1}(1, X_i) - \bar{\mu}_{C_+, 0}(0, X_i) )\right] \right].
\end{align*}
Note that this term will be canceled out by the second terms in $\mathbb{E}_P[m_j(W_i, \tau_{C_+}, \widetilde{\eta}) - m_j(W_i, \tau_{C_+}, \eta)]$ for $j = 1, 2$.

(3) $\mathbb{E}_P[m_j(W_i, \tau_{C_+}, \hat{\eta}_{k}) - m_j(W_i, \tau_{C_+}, \eta)]$ for $j = 4, 5$.

Note that
\begin{align*}
    m_4(W_i, \tau_{C_+}, \eta) 
    & = \frac{\mathds{1}_{P_i = O}}{\varphi} \frac{\varphi(S_i, X_i)}{1-\varphi(S_i, X_i)}\frac{1}{\rho(X_i)} \\
    & \quad \times (\rho(S_i, X_i) q_{C_+}(S_i, X_i, O) + [Y_i - q_{C_+}(S_i, X_i, O)]_{+} - \rho(S_i, X_i) \mu_{C_+, 1}(S_i, X_i, O)).
\end{align*}

\begin{align*}
    & \mathbb{E}_P[m_4(W_i, \tau_{C_+}, \hat{\eta}_{k}) - m_4(W_i, \tau_{C_+}, \eta)] \\
    & =
    \mathbb{E}_P\Bigg[\frac{\mathds{1}_{P_i = O}}{\widetilde{\varphi}} \frac{\widetilde{\varphi}(S_i, X_i)}{1-\widetilde{\varphi}(S_i, X_i)}\frac{\widetilde{\rho}(S_i, X_i)}{\widetilde{\rho}(X_i)} \left( \widetilde{q}_{C_+}(S_i, X_i, O) + \frac{[Y_i - \widetilde{q}_{C_+}(S_i, X_i, O)]_{+}}{\widetilde{\rho}(S_i, X_i)} - \widetilde{\mu}_{C_+, 1}(S_i, X_i, O)\right)  \Bigg] \\
    & \quad -
    \mathbb{E}_P\Bigg[\frac{\mathds{1}_{P_i = O}}{\varphi} \frac{\varphi(S_i, X_i)}{1-\varphi(S_i, X_i)}\frac{\rho(S_i, X_i)}{\rho(X_i)} \left( q_{C_+}(S_i, X_i, O) + \frac{[Y_i - q_{C_+}(S_i, X_i, O)]_{+}}{\rho(S_i, X_i)} -  \mu_{C_+, 1}(S_i, X_i, O)\right)  \Bigg] \\
    & =
    \mathbb{E}_P\Bigg[\frac{\mathds{1}_{P_i = O}}{\widetilde{\varphi}} \frac{\widetilde{\varphi}(S_i, X_i)}{1-\widetilde{\varphi}(S_i, X_i)}\frac{1}{\widetilde{\rho}(X_i)} \\
    & \quad \quad \times \bigg[(\widetilde{\rho}(S_i, X_i) \widetilde{q}_{C_+}(S_i, X_i, O) + [Y_i - \widetilde{q}_{C_+}(S_i, X_i, O)]_{+} - \widetilde{\rho}(S_i, X_i) \widetilde{\mu}_{C_+, 1}(S_i, X_i, O)) \\
    & \quad \quad \quad - (\rho(S_i, X_i) q_{C_+}(S_i, X_i, O) + [Y_i - q_{C_+}(S_i, X_i, O)]_{+} - \rho(S_i, X_i) \mu_{C_+, 1}(S_i, X_i, O)) \bigg] \Bigg] \\
    & \quad + 
    \mathbb{E}_P\Bigg[\left(\frac{\mathds{1}_{P_i = O}}{\widetilde{\varphi}} \frac{\widetilde{\varphi}(S_i, X_i)}{1-\widetilde{\varphi}(S_i, X_i)}\frac{1}{\widetilde{\rho}(X_i)} - \frac{\mathds{1}_{P_i = O}}{\varphi} \frac{\varphi(S_i, X_i)}{1-\varphi(S_i, X_i)}\frac{1}{\rho(X_i)} \right) \\
    & \quad \quad \underbrace{\quad \times (\rho(S_i, X_i) q_{C_+}(S_i, X_i, O) + [Y_i - q_{C_+}(S_i, X_i, O)]_{+} - \rho(S_i, X_i) \mu_{C_+, 1}(S_i, X_i, O))  \Bigg]}_{=0 \text{ by definition of $\mu_{C_+, 1}$ and its dual representation.}} \\
    & = 
    \mathbb{E}_P\Bigg[\frac{\mathds{1}_{P_i = O}}{\widetilde{\varphi}} \frac{\widetilde{\varphi}(S_i, X_i)}{1-\widetilde{\varphi}(S_i, X_i)}\frac{1}{\widetilde{\rho}(X_i)} \\
    & \quad \quad \times \bigg[(\rho(S_i, X_i) \widetilde{q}_{C_+}(S_i, X_i, O) + [Y_i - \widetilde{q}_{C_+}(S_i, X_i, O)]_{+} - \rho(S_i, X_i) \widetilde{\mu}_{C_+, 1}(S_i, X_i, O)) \\
    & \quad \quad \quad - (\rho(S_i, X_i) q_{C_+}(S_i, X_i, O) + [Y_i - q_{C_+}(S_i, X_i, O)]_{+} - \rho(S_i, X_i) \mu_{C_+, 1}(S_i, X_i, O)) \bigg]  \Bigg] \\
    & \quad + 
    \mathbb{E}_P\Bigg[\frac{\mathds{1}_{P_i = O}}{\widetilde{\varphi}} \frac{\widetilde{\varphi}(S_i, X_i)}{1-\widetilde{\varphi}(S_i, X_i)}\frac{1}{\widetilde{\rho}(X_i)} (\widetilde{\rho}(S_i, X_i) - \rho(S_i, X_i)) (\widetilde{q}_{C_+}(S_i, X_i, O) -\widetilde{\mu}_{C_+, 1}(S_i, X_i, O))
     \Bigg] \\
    & = 
    \mathbb{E}_P\Bigg[
    \frac{\mathds{1}_{P_i = O}}{\widetilde{\varphi}} \frac{\widetilde{\varphi}(S_i, X_i)}{1-\widetilde{\varphi}(S_i, X_i)}\frac{1}{\widetilde{\rho}(X_i)} \\
    & \quad \quad \times \bigg[(\rho(S_i, X_i) \widetilde{q}_{C_+}(S_i, X_i, O) + [Y_i - \widetilde{q}_{C_+}(S_i, X_i, O)]_{+}) \\
    & \quad \quad \quad - (\rho(S_i, X_i) q_{C_+}(S_i, X_i, O) + [Y_i - q_{C_+}(S_i, X_i, O)]_{+}) \bigg]  \Bigg] \\
    & \quad -
    \mathbb{E}_P\Bigg[
    \left(\frac{\mathds{1}_{P_i = O}}{\widetilde{\varphi}} \frac{\widetilde{\varphi}(S_i, X_i)}{1-\widetilde{\varphi}(S_i, X_i)}\frac{1}{\widetilde{\rho}(X_i)} - \frac{\mathds{1}_{P_i = O}}{\widetilde{\varphi}} \frac{\varphi(S_i, X_i)}{1-\varphi(S_i, X_i)}\frac{1}{\widetilde{\rho}(X_i)}\right)
    \\
    &  \quad \quad \quad \quad \times \rho(S_i, X_i) (\widetilde{\mu}_{C_+, 1}(S_i, X_i, O)) - \mu_{C_+, 1}(S_i, X_i, O)) \Bigg] \\
    & \quad -
    \mathbb{E}_P\Bigg[
    \frac{\mathds{1}_{P_i = O}}{\widetilde{\varphi}} \frac{\varphi(S_i, X_i)}{1-\varphi(S_i, X_i)}\frac{1}{\widetilde{\rho}(X_i)}
     \rho(S_i, X_i) (\widetilde{\mu}_{C_+, 1}(S_i, X_i, O)) - \mu_{C_+, 1}(S_i, X_i, O)) \Bigg] \\
    & \quad + 
    \mathbb{E}_P\Bigg[
    \left(\frac{\mathds{1}_{P_i = O}}{\widetilde{\varphi}} \frac{\widetilde{\varphi}(S_i, X_i)}{1-\widetilde{\varphi}(S_i, X_i)}\frac{1}{\widetilde{\rho}(X_i)} - 
    \frac{\mathds{1}_{P_i = O}}{\widetilde{\varphi}} \frac{\varphi(S_i, X_i)}{1-\varphi(S_i, X_i)}\frac{1}{\widetilde{\rho}(X_i)} \right)\\
    &  \quad \quad \quad \quad \times \bigg[(\widetilde{\rho}(S_i, X_i) - \rho(S_i, X_i)) (\widetilde{q}_{C_+}(S_i, X_i, O) -\widetilde{\mu}_{C_+, 1}(S_i, X_i, O))
     \Bigg] \\
    & \quad + 
    \mathbb{E}_P\Bigg[
    \frac{\mathds{1}_{P_i = O}}{\widetilde{\varphi}} \frac{\varphi(S_i, X_i)}{1-\varphi(S_i, X_i)}\frac{1}{\widetilde{\rho}(X_i)} \bigg[(\widetilde{\rho}(S_i, X_i) - \rho(S_i, X_i)) (\widetilde{q}_{C_+}(S_i, X_i, O) -\widetilde{\mu}_{C_+, 1}(S_i, X_i, O))
     \Bigg] \\
\end{align*}

Note that
\begin{align*}
    m_5(W_i, \tau, \eta) 
    & = - \frac{\mathds{1}_{P_i = O}}{\varphi} \frac{\varphi(S_i, X_i)}{1-\varphi(S_i, X_i)}\frac{(1-\rho(S_i, X_i))}{1-\rho(X_i)} \\
    & \quad \times \left( q_{C_+}(S_i, X_i, O) - \frac{[Y_i - q_{C_+}(S_i, X_i, O)]_{-}}{(1-\rho(S_i, X_i))} - \mu_{C_+, 0}(S_i, X_i, O)\right).
\end{align*}

Then, we have the following result.
\begin{align*}
    & \mathbb{E}_P[m_5(W_i, \tau_{C_+}, \widetilde{\eta}) - m_5(W_i, \tau_{C_+}, \eta)] \\
    & =
    -
    \mathbb{E}_P\Bigg[\frac{\mathds{1}_{P_i = O}}{\widetilde{\varphi}} \frac{\widetilde{\varphi}(S_i, X_i)}{1-\widetilde{\varphi}(S_i, X_i)}\frac{1}{1-\widetilde{\rho
    }(X_i)} \\
    & \quad \quad \quad \times \big((1-\widetilde{\rho}(S_i, X_i)) \widetilde{q}_{C_+}(S_i, X_i, O) - [Y_i - \widetilde{q}_{C_+}(S_i, X_i, O)]_{-} - (1-\widetilde{\rho}(S_i, X_i)) \widetilde{\mu}_{C_+, 0}(S_i, X_i, O)\big) \Bigg] \\
    & \quad +
    \mathbb{E}_P\Bigg[\frac{\mathds{1}_{P_i = O}}{\varphi} \frac{\varphi(S_i, X_i)}{1-\varphi(S_i, X_i)}\frac{1}{1-\rho(X_i)} \\
    & \quad \quad \quad \times \big((1-\rho(S_i, X_i)) q_{C_+}(S_i, X_i, O) - [Y_i - q_{C_+}(S_i, X_i, O)]_{-} - (1-\rho(S_i, X_i)) \mu_{C_+, 0}(S_i, X_i, O)\big) \Bigg] \\
    & = -
    \mathbb{E}_P\Bigg[\frac{\mathds{1}_{P_i = O}}{\widetilde{\varphi}} \frac{\widetilde{\varphi}(S_i, X_i)}{1-\widetilde{\varphi}(S_i, X_i)}\frac{1}{1-\widetilde{\rho}(X_i)} \\
    & \quad \quad \times 
    \bigg[\big((1-\widetilde{\rho}(S_i, X_i)) \widetilde{q}_{C_+}(S_i, X_i, O) - [Y_i - \widetilde{q}_{C_+}(S_i, X_i, O)]_{-} - (1-\widetilde{\rho}(S_i, X_i)) \widetilde{\mu}_{C_+, 0}(S_i, X_i, O)\big) \\
    & \quad \quad \quad - \big((1-\rho(S_i, X_i))q_{C_+}(S_i, X_i, O) - [Y_i - q_{C_+}(S_i, X_i, O)]_{-} - (1-\rho(S_i, X_i)) \mu_{C_+, 0}(S_i, X_i, O)\big) \bigg] \Bigg] \\
    & \quad -
    \mathbb{E}_P\Bigg[\left(\frac{\mathds{1}_{P_i = O}}{\widetilde{\varphi}} \frac{\widetilde{\varphi}(S_i, X_i)}{1-\widetilde{\varphi}(S_i, X_i)}\frac{1}{1-\widetilde{\rho}(X_i)} - \frac{\mathds{1}_{P_i = O}}{\varphi} \frac{\varphi(S_i, X_i)}{1-\varphi(S_i, X_i)}\frac{1}{1-\rho(X_i)} \right) \\
    & \quad \quad \underbrace{\quad \times \big((1-\rho(S_i, X_i)) q_{C_+}(S_i, X_i, O) - [Y_i - q_{C_+}(S_i, X_i, O)]_{-} - (1- \rho(S_i, X_i)) \mu_{C_+, 0}(S_i, X_i, O)\big) \Bigg]}_{=0 \text{ by definition of $\mu_{C_+, 1}$ and its dual representation.}} \\
    & = 
    -
    \mathbb{E}_P\Bigg[\frac{\mathds{1}_{P_i = O}}{\widetilde{\varphi}} \frac{\widetilde{\varphi}(S_i, X_i)}{1-\widetilde{\varphi}(S_i, X_i)}\frac{1}{1-\widetilde{\rho}(X_i)} \\
    & \quad  \quad \quad \times \bigg[\big((1-\rho(S_i, X_i)) \widetilde{q}_{C_+}(S_i, X_i, O) - [Y_i - \widetilde{q}_{C_+}(S_i, X_i, O)]_{+} - (1-\rho(S_i, X_i)) \widetilde{\mu}_{C_+, 0}(S_i, X_i, O)) \\
    & \quad  \quad \quad \quad - \big((1-\rho(S_i, X_i)) q_{C_+}(S_i, X_i, O) - [Y_i - q_{C_+}(S_i, X_i, O)]_{-} - (1-\rho(S_i, X_i)) \mu_{C_+, 0}(S_i, X_i, O)) \bigg] \Bigg] \\
    & \quad + 
    \mathbb{E}_P\Bigg[\frac{\mathds{1}_{P_i = O}}{\widetilde{\varphi}} \frac{\widetilde{\varphi}(S_i, X_i)}{1-\widetilde{\varphi}(S_i, X_i)}\frac{1}{1-\widetilde{\rho}(X_i)} \bigg[(\widetilde{\rho}(S_i, X_i) - \rho(S_i, X_i)) (\widetilde{q}_{C_+}(S_i, X_i, O) -\widetilde{\mu}_{C_+, 0}(S_i, X_i, O))
     \Bigg] \\
    & = 
    -
    \mathbb{E}_P\Bigg[
    \frac{\mathds{1}_{P_i = O}}{\widetilde{\varphi}} \frac{\widetilde{\varphi}(S_i, X_i)}{1-\widetilde{\varphi}(S_i, X_i)}\frac{1}{1-\widetilde{\rho}(X_i)} \\
    & \quad \quad \quad \times \bigg[\big((1-\rho(S_i, X_i)) \widetilde{q}_{C_+}(S_i, X_i, O) - [Y_i - \widetilde{q}_{C_+}(S_i, X_i, O)]_{-}) \\
    & \quad \quad \quad \quad - \big((1-\rho(S_i, X_i)) q_{C_+}(S_i, X_i, O) - [Y_i - q_{C_+}(S_i, X_i, O)]_{-}\big) \bigg]  \Bigg] \\
    & \quad +
    \mathbb{E}_P\Bigg[
    \left(\frac{\mathds{1}_{P_i = O}}{\widetilde{\varphi}} \frac{\widetilde{\varphi}(S_i, X_i)}{1-\widetilde{\varphi}(S_i, X_i)}\frac{1}{1-\widetilde{\rho}(X_i)} - \frac{\mathds{1}_{P_i = O}}{\widetilde{\varphi}} \frac{\varphi(S_i, X_i)}{1-\varphi(S_i, X_i)}\frac{1}{1-\widetilde{\rho}(X_i)}\right)
    \\
    & \quad \quad \quad \quad \times \bigg[(1-\rho(S_i, X_i)) (\widetilde{\mu}_{C_+, 0}(S_i, X_i, O)) - \mu_{C_+, 0}(S_i, X_i, O)) \bigg]  \Bigg] \\
    & \quad +
    \mathbb{E}_P\Bigg[
    \frac{\mathds{1}_{P_i = O}}{\widetilde{\varphi}} \frac{\varphi(S_i, X_i)}{1-\varphi(S_i, X_i)}\frac{1}{1-\widetilde{\rho}(X_i)}
    (1-\rho(S_i, X_i)) (\widetilde{\mu}_{C_+, 0}(S_i, X_i, O)) - \mu_{C_+, 0}(S_i, X_i, O))  \Bigg] \\
    & \quad + 
    \mathbb{E}_P\Bigg[
    \left(\frac{\mathds{1}_{P_i = O}}{\widetilde{\varphi}} \frac{\widetilde{\varphi}(S_i, X_i)}{1-\widetilde{\varphi}(S_i, X_i)}\frac{1}{1-\widetilde{\rho}(X_i)} - 
    \frac{\mathds{1}_{P_i = O}}{\widetilde{\varphi}} \frac{\varphi(S_i, X_i)}{1-\varphi(S_i, X_i)}\frac{1}{1-\widetilde{\rho}(X_i)} \right)\\
    & \quad \quad \quad \times \bigg[(\widetilde{\rho}(S_i, X_i) - \rho(S_i, X_i)) (\widetilde{q}_{C_+}(S_i, X_i, O) -\widetilde{\mu}_{C_+, 0}(S_i, X_i, O))
     \Bigg] \\
    & \quad + 
    \mathbb{E}_P\Bigg[
    \frac{\mathds{1}_{P_i = O}}{\widetilde{\varphi}} \frac{\varphi(S_i, X_i)}{1-\varphi(S_i, X_i)}\frac{1}{1-\widetilde{\rho}(X_i)} \bigg[(\widetilde{\rho}(S_i, X_i) - \rho(S_i, X_i)) (\widetilde{q}_{C_+}(S_i, X_i, O) -\widetilde{\mu}_{C_+, 0}(S_i, X_i, O))
     \Bigg].
\end{align*}

(4) $\mathbb{E}[m_j(W_i, \tau_{C_+}, \widetilde{\eta}) - m_j(W_i, \tau_{C_+}, \eta)]$ for $j = 6, 7$.

\begin{align*}
& \mathbb{E}[m_6(W_i, \tau_{C_+}, \widetilde{\eta}) - m_6(W_i, \tau_{C_+}, \eta)] \\
& =
\mathbb{E}_P \Bigg[\frac{\mathds{1}_{P_i = E}}{\widetilde{\varphi}} \frac{1}{\widetilde{\rho}(X_i)} \left[\widetilde{q}_{C_+}(S_i, X_i, O) - \widetilde{\mu}_{C_+, 1}(S_i, X_i, O)\right] \left(W_i - \widetilde{\rho}(S_i, X_i)\right)  \Bigg] \\
& \quad - 
\mathbb{E}_P\Bigg[\frac{\mathds{1}_{P_i = E}}{\varphi} \frac{1}{\rho(X_i)} \left[q_{C_+}(S_i, X_i, O) - \mu_{C_+, 1}(S_i, X_i, O)\right] \left(W_i - \rho(S_i, X_i)\right) \Bigg] \\
& =
-
\mathbb{E}_P\Bigg[\frac{\mathds{1}_{P_i = E}}{\widetilde{\varphi}} \frac{1}{\widetilde{\rho}(X_i)} \left[\widetilde{q}_{C_+}(S_i, X_i, O) - \widetilde{\mu}_{C_+, 1}(S_i, X_i, O)\right] \left(\widetilde{\rho}(S_i, X_i) - \rho(S_i, X_i)\right) \Bigg] \\
& \quad - 
\mathbb{E}_P\Bigg[  
\bigg(\frac{\mathds{1}_{P_i = E}}{\widetilde{\varphi}} \frac{1}{\widetilde{\rho}(X_i)} \left[\widetilde{q}_{C_+}(S_i, X_i, O) - \widetilde{\mu}_{C_+, 1}(S_i, X_i, O)\right] \\
& \quad \quad \quad \quad \quad - \frac{\mathds{1}_{P_i = E}}{\varphi} \frac{1}{\rho(X_i)} \left[q_{C_+}(S_i, X_i, O) - \mu_{C_+, 1}(S_i, X_i, O)\right]\bigg) \\
& \quad \quad \underbrace{\quad \times \left(W_i - \rho(S_i, X_i)\right)  \Bigg]. \quad \quad  \quad \quad \quad  \quad\quad \quad  \quad\quad \quad  \quad }_{=0 \text{ by definition of $\rho(S_i, X_i)$}.} \\
& =
-
\mathbb{E}_P\Bigg[\frac{\mathds{1}_{P_i = E}}{\widetilde{\varphi}} \frac{1}{\widetilde{\rho}(X_i)} \left[\widetilde{q}_{C_+}(S_i, X_i, O) - \widetilde{\mu}_{C_+, 1}(S_i, X_i, O)\right] \left(\widetilde{\rho}(S_i, X_i) - \rho(S_i, X_i)\right) \Bigg]
\end{align*}

\begin{align*}
& \mathbb{E}_P[m_7(W_i,  \tau_{C_+}, \widetilde{\eta}) - m_7(W_i,  \tau_{C_+}, \eta)] \\
& =
\mathbb{E}_P\Bigg[\frac{\mathds{1}_{P_i = E}}{\widetilde{\varphi}} \frac{1}{1-\widetilde{\rho}(X_i)} \left[\widetilde{q}_{C_+}(S_i, X_i, O) - \widetilde{\mu}_{C_+, 0}(S_i, X_i, O)\right] \left(W_i - \widetilde{\rho}(S_i, X_i)\right)  \Bigg] \\
& \quad - 
\mathbb{E}_P\Bigg[\frac{\mathds{1}_{P_i = E}}{\varphi} \frac{1}{1-\rho(X_i)} \left[q_{C_+}(S_i, X_i, O) - \mu_{C_+, 0}(S_i, X_i, O)\right] \left(W_i - \rho(S_i, X_i)\right) \Bigg] \\
& =
-
\mathbb{E}_P\Bigg[\frac{\mathds{1}_{P_i = E}}{\widetilde{\varphi}} \frac{1}{1-\widetilde{\rho}(X_i)} \left[\widetilde{q}_{C_+}(S_i, X_i, O) - \widetilde{\mu}_{C_+, 0}(S_i, X_i, O)\right] \left(\widetilde{\rho}(S_i, X_i) - \rho(S_i, X_i)\right)  \Bigg] \\
& \quad - 
\mathbb{E}_P\Bigg[  
\bigg(\frac{\mathds{1}_{P_i = E}}{\widetilde{\varphi}} \frac{1}{1-\widetilde{\rho}(X_i)} \left[\widetilde{q}_{C_+}(S_i, X_i, O) - \widetilde{\mu}_{C_+, 0}(S_i, X_i, O)\right] \\
& \quad \quad \quad \quad \quad - \frac{\mathds{1}_{P_i = E}}{\varphi} \frac{1}{1-\rho(X_i)} \left[q_{C_+}(S_i, X_i, O) - \mu_{C_+, 0}(S_i, X_i, O)\right]\bigg) \\
& \quad \quad \underbrace{\quad \times \left(W_i - \rho(S_i, X_i)\right)  \Bigg]. \quad \quad  \quad \quad \quad  \quad\quad \quad  \quad\quad \quad  \quad }_{=0 \text{ by definition of $\rho(S_i, X_i)$}.} \\
& =
-
\mathbb{E}_P\Bigg[\frac{\mathds{1}_{P_i = E}}{\widetilde{\varphi}} \frac{1}{1-\widetilde{\rho}(X_i)} \left[\widetilde{q}_{C_+}(S_i, X_i, O) - \widetilde{\mu}_{C_+, 0}(S_i, X_i, O)\right] \left(\widetilde{\rho}(S_i, X_i) - \rho(S_i, X_i)\right)  \Bigg]
\end{align*}

\textbf{Part 2.} By combining all results in Part 1, we have
\begin{align*}
    \mathbb{E}_P[m(W_i, \tau_{C_+}, \widetilde{\eta}) - m(W_i, \tau_{C_+}, \eta)] 
    & = 
    \sum_{j=1}^7 \mathbb{E}_P[m_j(W_i, \tau_{C_+}, \widetilde{\eta}) - m_j(W_i, \tau_{C_+}, \eta)] = \sum_{j=1}^8 \mathcal{J}_j,  
\end{align*} 

where
\begin{align*}
    \mathcal{J}_1 & = 
    - \mathbb{E}_P\Bigg[\bigg(\frac{\mathds{1}_{P_i = E}}{\widetilde{\varphi}} \frac{W_i}{\widetilde{\rho}(X_i)} - \frac{\mathds{1}_{P_i = E}}{\widetilde{\varphi}} \frac{W_i}{\rho(X_i)} \bigg)(\widetilde{\bar{\mu}}_{C_+, 1}(1, X_i) - \bar{\mu}_{C_+, 1}(1, X_i)) \Bigg] \\
    \mathcal{J}_2 & = 
    \mathbb{E}_P\Bigg[\bigg(\frac{\mathds{1}_{P_i = E}}{\widetilde{\varphi}} \frac{1-W_i}{1-\widetilde{\rho}(X_i)} - \frac{\mathds{1}_{P_i = E}}{\widetilde{\varphi}} \frac{1-W_i}{1-\rho(X_i)} \bigg) (\widetilde{\bar{\mu}}_{C_+, 0}(1, X_i) - \bar{\mu}_{C_+, 0}(1, X_i)) \Bigg] \\
    \mathcal{J}_3 & =
    \mathbb{E}_P\Bigg[
    \frac{\mathds{1}_{P_i = O}}{\widetilde{\varphi}} \frac{\widetilde{\varphi}(S_i, X_i)}{1-\widetilde{\varphi}(S_i, X_i)}\frac{1}{\widetilde{\rho}(X_i)} \\
    & \quad \quad \times \bigg[(\rho(S_i, X_i) \widetilde{q}_{C_+}(S_i, X_i, O) + [Y_i - \widetilde{q}_{C_+}(S_i, X_i, O)]_{+}) \\
    & \quad \quad \quad - (\rho(S_i, X_i) q_{C_+}(S_i, X_i, O) + [Y_i - q_{C_+}(S_i, X_i, O)]_{+}) \bigg]  \Bigg] \\
    \mathcal{J}_4 & =
    -
    \mathbb{E}_P\Bigg[
    \left(\frac{\mathds{1}_{P_i = O}}{\widetilde{\varphi}} \frac{\widetilde{\varphi}(S_i, X_i)}{1-\widetilde{\varphi}(S_i, X_i)}\frac{1}{\widetilde{\rho}(X_i)} - \frac{\mathds{1}_{P_i = O}}{\widetilde{\varphi}} \frac{\varphi(S_i, X_i)}{1-\varphi(S_i, X_i)}\frac{1}{\widetilde{\rho}(X_i)}\right)
    \\
    & \quad \quad \times \bigg[\rho(S_i, X_i) (\widetilde{\mu}_{C_+, 1}(S_i, X_i, O)) - \mu_{C_+, 1}(S_i, X_i, O)) \bigg]  \Bigg] \\
    \mathcal{J}_5 & =
    \mathbb{E}_P\Bigg[
    \left(\frac{\mathds{1}_{P_i = O}}{\widetilde{\varphi}} \frac{\widetilde{\varphi}(S_i, X_i)}{1-\widetilde{\varphi}(S_i, X_i)}\frac{1}{\widetilde{\rho}(X_i)} - 
    \frac{\mathds{1}_{P_i = O}}{\widetilde{\varphi}} \frac{\varphi(S_i, X_i)}{1-\varphi(S_i, X_i)}\frac{1}{\widetilde{\rho}(X_i)} \right)\\
    & \quad \quad \times \bigg[(\widetilde{\rho}(S_i, X_i) - \rho(S_i, X_i)) (\widetilde{q}_{C_+}(S_i, X_i, O) -\widetilde{\mu}_{C_+, 1}(S_i, X_i, O))
     \Bigg] \\
    \mathcal{J}_6 & =
    -
    \mathbb{E}_P\Bigg[
    \frac{\mathds{1}_{P_i = O}}{\widetilde{\varphi}} \frac{\widetilde{\varphi}(S_i, X_i)}{1-\widetilde{\varphi}(S_i, X_i)}\frac{1}{1-\widetilde{\rho}(X_i)} \\
    & \quad \quad \times \bigg[\big((1-\rho(S_i, X_i)) \widetilde{q}_{C_+}(S_i, X_i, O) - [Y_i - \widetilde{q}_{C_+}(S_i, X_i, O)]_{-}) \\
    & \quad \quad \quad - \big((1-\rho(S_i, X_i)) q_{C_+}(S_i, X_i, O) - [Y_i - q_{C_+}(S_i, X_i, O)]_{-}\big) \bigg] \bigg| 
     \Bigg] \\
    \mathcal{J}_7 & =
    \mathbb{E}_P\Bigg[
    \left(\frac{\mathds{1}_{P_i = O}}{\widetilde{\varphi}} \frac{\widetilde{\varphi}(S_i, X_i)}{1-\widetilde{\varphi}(S_i, X_i)}\frac{1}{1-\widetilde{\rho}(X_i)} - \frac{\mathds{1}_{P_i = O}}{\widetilde{\varphi}} \frac{\varphi(S_i, X_i)}{1-\varphi(S_i, X_i)}\frac{1}{1-\widetilde{\rho}(X_i)}\right)
    \\
    & \quad \quad \times \bigg[(1-\rho(S_i, X_i)) (\widetilde{\mu}_{C_+, 0}(S_i, X_i, O)) - \mu_{C_+, 0}(S_i, X_i, O)) \bigg]  \Bigg] \\
    \mathcal{J}_8 & =
    \mathbb{E}_P\Bigg[
    \left(\frac{\mathds{1}_{P_i = O}}{\widetilde{\varphi}} \frac{\widetilde{\varphi}(S_i, X_i)}{1-\widetilde{\varphi}(S_i, X_i)}\frac{1}{1-\widetilde{\rho}(X_i)} - 
    \frac{\mathds{1}_{P_i = O}}{\widetilde{\varphi}} \frac{\varphi(S_i, X_i)}{1-\varphi(S_i, X_i)}\frac{1}{1-\widetilde{\rho}(X_i)} \right)\\
    & \quad \quad \times \bigg[(\widetilde{\rho}(S_i, X_i) - \rho(S_i, X_i)) (\widetilde{q}_{C_+}(S_i, X_i, O) -\widetilde{\mu}_{C_+, 0}(S_i, X_i, O))
     \Bigg] \\
    \end{align*}
\end{proof}

\subsection{Proof of \Cref{lemma:DML-verification-copula}}

\label{sec:supp-proofs-asym-copula}

\subsubsection{Technical Lemma}

With the abuse of the notation, we define 
{\footnotesize
\begin{align*}
        m_1(Z_i, \tau, \widetilde{\eta}) & = \frac{\mathds{1}_{P_i = E}}{\widetilde{\varphi}} \frac{W_i}{\widetilde{\rho}(X_i)} (\widetilde{\mu}_{C_o, 1}(S_i, X_i, O) - \widetilde{\bar{\mu}}_{C_0, 1}(1, X_i)), \\
        m_2(Z_i, \tau, \widetilde{\eta}) & = -\frac{\mathds{1}_{P_i = E}}{\widetilde{\varphi}} \frac{1-W_i}{1-\widetilde{\rho}(X_i)} (\widetilde{\mu}_{C_o, o}(S_i, X_i, O) - \widetilde{\bar{\mu}}_{C_o, 0}(0, X_i)), \\
        m_3(Z_i, \tau, \widetilde{\eta}) & = \frac{\mathds{1}_{P_i = E}}{\widetilde{\varphi}}(\widetilde{\bar{\mu}}_{C_o, 1}(1, X_i) - \widetilde{\bar{\mu}}_{C_o, 0}(0, X_i) - \tau) \\
        m_4(Z_i, \tau, \widetilde{\eta}) 
        & = \frac{\mathds{1}_{P_i = O}}{\widetilde{\varphi}} \frac{\widetilde{\varphi}(S_i, X_i)}{1-\widetilde{\varphi}(S_i, X_i)} \frac{\widetilde{\rho}(S_i, X_i)}{\widetilde{\rho}(X_i)} (h_{C_o, 1}(Y_i, \widetilde{F}_Y^{-1}(\cdot|S_i, X_i, O), 1 - \widetilde{\rho}(S_i, X_i)) -\widetilde{\mu}_{C_o, 1}(S_i, X_i, O)), \\
        & = \frac{\mathds{1}_{P_i = O}}{\widetilde{\varphi}} \frac{\widetilde{\varphi}(S_i, X_i)}{1-\widetilde{\varphi}(S_i, X_i)} \frac{1}{\widetilde{\rho}(X_i)} (\widetilde{h}_{C_o, 1}(Y_i, \widetilde{F}_Y^{-1}(\cdot|S_i, X_i, O), \widetilde{\rho}(S_i, X_i)) -\widetilde{\rho}(S_i, X_i) \widetilde{\mu}_{C_o, 1}(S_i, X_i, O)),
        \\
        m_5(Z_i, \tau, \widetilde{\eta}) & = - \frac{\mathds{1}_{P_i = O}}{\widetilde{\varphi}} \frac{\widetilde{\varphi}(S_i, X_i)}{1-\widetilde{\varphi}(S_i, X_i)}\frac{1-\widetilde{\rho}(S_i, X_i)}{1-\widetilde{\rho}(X_i)} (h_{C_o, 0}(Y_i, \widetilde{F}_Y^{-1}(\cdot|S_i, X_i, O), 1 - \widetilde{\rho}(S_i, X_i)) -\widetilde{\mu}_{C_o, 0}(S_i, X_i, O)), \\
        & = - \frac{\mathds{1}_{P_i = O}}{\widetilde{\varphi}} \frac{\widetilde{\varphi}(S_i, X_i)}{1-\widetilde{\varphi}(S_i, X_i)} \frac{1}{1-\widetilde{\rho}(X_i)}  \\
    & \quad \times (\tilde{h}_{C_o, 0}(Y_i, F_Y^{-1}(\cdot|S_i, X_i, O), 1-\rho(S_i, X_i)) -(1-\widetilde{\rho}(S_i, X_i))\widetilde{\mu}_{C_o, 0}(S_i, X_i, O)), \\
        m_6(Z_i, \tau, \widetilde{\eta}) & = \frac{\mathds{1}_{P_i = E}}{\widetilde{\varphi}} \frac{1}{\widetilde{\rho}(X_i)} (\widetilde{d}_{C_o, 1}(S_i, X_i) - \mu_{C_o, 1}(S_i, X_i, O)) \left(W_i - \widetilde{\rho}(S_i, X_i)\right), \\
        m_7(Z_i, \tau, \eta) & = \frac{\mathds{1}_{P_i = E}}{\widetilde{\varphi}} \frac{1}{1 - \widetilde{\rho}(X_i)} (\widetilde{d}_{C_o, 0}(S_i, X_i) - \mu_{C_o, 0}(S_i, X_i, O))\left(W_i - \widetilde{\rho}(S_i, X_i)\right).
    \end{align*}
}
where
{\footnotesize
\begin{align*}
    & \tilde{h}_{C_o, 1}(Y_i, F_Y^{-1}(\cdot|S_i, X_i, O), \widetilde{\rho}(S_i, X_i)) \\
    & = (1-C_o(1-\widetilde{\rho}(S_i, X_i)|0)) Y_i + \int_0^1 \left((1-u)F_Y^{-1}(u|S_i, X_i, O) + [Y_i - F_Y^{-1}(u|S_i, X_i, O)]_{+}\right)d(1-C_o(1-\widetilde{\rho}(S_i, X_i)|u), \\
    & \tilde{h}_{C_o, 0}(Y_i, F_Y^{-1}(\cdot|S_i, X_i, O), 1-\widetilde{\rho}(S_i, X_i)) \\
    & = C_o(1-\widetilde{\rho}(S_i, X_i)|1) Y_i  - \int_0^1 \left(u F_Y^{-1}(u|S_i, X_i, O) - [Y_i - F_Y^{-1}(u|S_i, X_i, O)]_{-}\right)d C_o(1-\widetilde{\rho}(S_i, X_i)|u)
\end{align*}
}
Note that
\begin{align*}
    m_{C_o}(Z_i, \tau_{C_o}, \eta) = \sum_{j=1}^7 m_{j}(Z_i, \tau_{C_o}, \eta).
\end{align*}

\begin{lemma}
    \label{lemma:tech-remainer-copula}
    \begin{align*}
    \mathbb{E}_P[m(W_i, \tau_{C_o}, \widetilde{\eta}) - m(W_i, \tau_{C_o}, \eta)] 
    & = 
    \sum_{j=1}^7 \mathbb{E}_P[m_j(W_i, \tau_{C_o}, \widetilde{\eta}) - m_j(W_i, \tau_{C_o}, \eta)] = \sum_{j=1}^{14} \mathcal{J}_j,  
\end{align*} 
where
{\footnotesize
\begin{align*}
    \mathcal{J}_1 & = 
    - \mathbb{E}_P\Bigg[\bigg(\frac{\mathds{1}_{P_i = E}}{\widetilde{\varphi}} \frac{W_i}{\widetilde{\rho}(X_i)} - \frac{\mathds{1}_{P_i = E}}{\widetilde{\varphi}} \frac{W_i}{\rho(X_i)} \bigg)(\widetilde{\bar{\mu}}_{C_+, 1}(1, X_i) - \bar{\mu}_{C_+, 1}(1, X_i)) \Bigg] \\
    \mathcal{J}_2 & = 
    \mathbb{E}_P\Bigg[\bigg(\frac{\mathds{1}_{P_i = E}}{\widetilde{\varphi}} \frac{1-W_i}{1-\widetilde{\rho}(X_i)} - \frac{\mathds{1}_{P_i = E}}{\widetilde{\varphi}} \frac{1-W_i}{1-\rho(X_i)} \bigg) (\widetilde{\bar{\mu}}_{C_+, 0}(1, X_i) - \bar{\mu}_{C_+, 0}(1, X_i)) \Bigg]
    \end{align*}

\begin{align*}
 \mathcal{J}_3   & =
\mathbb{E}_P\Bigg[\frac{\mathds{1}_{P_i = O}}{\widetilde{\varphi}} \frac{\widetilde{\varphi}(S_i, X_i)}{1-\widetilde{\varphi}(S_i, X_i)} \frac{1}{\widetilde{\rho}(X_i)}  \\
& \quad \quad \times (\tilde{h}_{C_o, 1}(Y_i, \widetilde{F}_Y^{-1}(\cdot|S_i, X_i, O), \widetilde{\rho}(S_i, X_i)) - \tilde{h}_{C_o, 1}(Y_i, F_Y^{-1}(\cdot|S_i, X_i, O), \widetilde{\rho}(S_i, X_i))\Bigg] 
\\
\mathcal{J}_4 &  =
\mathbb{E}_P\Bigg[\frac{\mathds{1}_{P_i = O}}{\widetilde{\varphi}} \bigg(\frac{\widetilde{\varphi}(S_i, X_i)}{1-\widetilde{\varphi}(S_i, X_i)} - \frac{\varphi(S_i, X_i)}{1-\varphi(S_i, X_i)} \bigg) \frac{1}{\widetilde{\rho}(X_i)}  \\
& \quad \quad \times (\tilde{h}_{C_o, 1}(Y_i, F_Y^{-1}(\cdot|S_i, X_i, O), \widetilde{\rho}(S_i, X_i)) - \tilde{h}_{C_o, 1}(Y_i, F_Y^{-1}(\cdot|S_i, X_i, O), \rho(S_i, X_i)))\Bigg] 
\\
\mathcal{J}_5 & =  - 
\mathbb{E}_P\Bigg[\frac{\mathds{1}_{P_i = O}}{\widetilde{\varphi}} \bigg(\frac{\widetilde{\varphi}(S_i, X_i)}{1-\widetilde{\varphi}(S_i, X_i)} - \frac{\varphi(S_i, X_i)}{1-\varphi(S_i, X_i)} \bigg) \frac{1}{\widetilde{\rho}(X_i)}  (\widetilde{\rho}(S_i, X_i)  - \rho(S_i, X_i))\widetilde{\mu}_{C_o, 1}(S_i, X_i, O)\Bigg]
\\
\mathcal{J}_6 & = - 
\mathbb{E}_P\Bigg[\frac{\mathds{1}_{P_i = O}}{\widetilde{\varphi}} \bigg(\frac{\widetilde{\varphi}(S_i, X_i)}{1-\widetilde{\varphi}(S_i, X_i)} - \frac{\varphi(S_i, X_i)}{1-\varphi(S_i, X_i)} \bigg) \frac{1}{\widetilde{\rho}(X_i)} \rho(S_i, X_i) (\widetilde{\mu}_{C_o, 1}(S_i, X_i, O) -  \mu_{C_o, 1}(S_i, X_i, O))\Bigg] \\
\mathcal{J}_7 & = 
\mathbb{E}_P\Bigg[\frac{\mathds{1}_{P_i = O}}{\widetilde{\varphi}} \frac{\varphi(S_i, X_i)}{1-\varphi(S_i, X_i)} \frac{1}{\widetilde{\rho}(X_i)}  \\
& \quad \quad \times (\tilde{h}_{C_o, 1}(Y_i, F_Y^{-1}(\cdot|S_i, X_i, O), \widetilde{\rho}(S_i, X_i)) - \tilde{h}_{C_o, 1}(Y_i, F_Y^{-1}(\cdot|S_i, X_i, O), \rho(S_i, X_i)))\Bigg] 
\end{align*}

\begin{align*}
\mathcal{J}_8 & = 
- \mathbb{E}_P\Bigg[\frac{\mathds{1}_{P_i = O}}{\widetilde{\varphi}} \frac{\widetilde{\varphi}(S_i, X_i)}{1-\widetilde{\varphi}(S_i, X_i)} \frac{1}{1-\widetilde{\rho}(X_i)}  \\
& \quad \quad \times (\tilde{h}_{C_o, 0}(Y_i, \widetilde{F}_Y^{-1}(\cdot|S_i, X_i, O), 1-\widetilde{\rho}(S_i, X_i)) - \tilde{h}_{C_o, 0}(Y_i, F_Y^{-1}(\cdot|S_i, X_i, O), 1-\widetilde{\rho}(S_i, X_i))\Bigg] 
\\
\mathcal{J}_9 & =  -
\mathbb{E}_P\Bigg[\frac{\mathds{1}_{P_i = O}}{\widetilde{\varphi}} \bigg(\frac{\widetilde{\varphi}(S_i, X_i)}{1-\widetilde{\varphi}(S_i, X_i)} - \frac{\varphi(S_i, X_i)}{1-\varphi(S_i, X_i)} \bigg) \frac{1}{1-\widetilde{\rho}(X_i)}  \\
& \quad \quad \times (\tilde{h}_{C_o, 0}(Y_i, F_Y^{-1}(\cdot|S_i, X_i, O), 1-\widetilde{\rho}(S_i, X_i)) - \tilde{h}_{C_o, 0}(Y_i, F_Y^{-1}(\cdot|S_i, X_i, O), 1-\rho(S_i, X_i)))\Bigg] 
\\
\mathcal{J}_{10} & =  -
\mathbb{E}_P\Bigg[\frac{\mathds{1}_{P_i = O}}{\widetilde{\varphi}} \bigg(\frac{\widetilde{\varphi}(S_i, X_i)}{1-\widetilde{\varphi}(S_i, X_i)} - \frac{\varphi(S_i, X_i)}{1-\varphi(S_i, X_i)} \bigg) \frac{1}{1-\widetilde{\rho}(X_i)}  (\widetilde{\rho}(S_i, X_i)  - \rho(S_i, X_i))\widetilde{\mu}_{C_o, 0}(S_i, X_i, O)\Bigg]
\\
\mathcal{J}_{11} & = 
\mathbb{E}_P\Bigg[\frac{\mathds{1}_{P_i = O}}{\widetilde{\varphi}} \bigg(\frac{\widetilde{\varphi}(S_i, X_i)}{1-\widetilde{\varphi}(S_i, X_i)} - \frac{\varphi(S_i, X_i)}{1-\varphi(S_i, X_i)} \bigg) \frac{1-\rho(S_i, X_i)}{1-\widetilde{\rho}(X_i)} (\widetilde{\mu}_{C_o, 0}(S_i, X_i, O) -  \mu_{C_o, 0}(S_i, X_i, O))\Bigg] \\
\mathcal{J}_{12} & =  -
\mathbb{E}_P\Bigg[\frac{\mathds{1}_{P_i = O}}{\widetilde{\varphi}} \frac{\varphi(S_i, X_i)}{1-\varphi(S_i, X_i)} \frac{1}{\widetilde{\rho}(X_i)}  \\
& \quad \quad \times (\tilde{h}_{C_o, 0}(Y_i, F_Y^{-1}(\cdot|S_i, X_i, O), 1-\widetilde{\rho}(S_i, X_i)) - \tilde{h}_{C_o, 0}(Y_i, F_Y^{-1}(\cdot|S_i, X_i, O), 1-\rho(S_i, X_i)))\Bigg] 
\end{align*}

\begin{align*}
    \mathcal{J}_{13} & = 
    - 
    \mathbb{E}_P\Bigg[\frac{\mathds{1}_{P_i = E}}{\widetilde{\varphi}} \frac{1}{\widetilde{\rho}(X_i)} \widetilde{d}_{C_o, 1}(S_i, X_i) \left(\widetilde{\rho}(S_i, X_i) - \rho(S_i, X_i)\right)\Bigg] \\
    \mathcal{J}_{14} & =  - 
    \mathbb{E}_P\Bigg[\frac{\mathds{1}_{P_i = E}}{\widetilde{\varphi}} \frac{1}{1-\widetilde{\rho}(X_i)} \widetilde{d}_{C_o, 0}(S_i, X_i) \left(\widetilde{\rho}(S_i, X_i) - \rho(S_i, X_i)\right)\Bigg].
\end{align*}
}

\end{lemma}

\begin{proof}[Proof of \Cref{lemma:tech-remainer-copula}]

We show the result by computing $\mathbb{E}_P[m_j(W_i, \tau, \widetilde{\eta}) - m_j(W_i, \tau, \eta)]$ for $j = 1, \dotsc, 7$.

\textbf{Part 1: Calculation of $\mathbb{E}_P[m_j(Z_i, \tau_{C_o}, \widetilde{\eta}) - m_j(Z_i, \tau_{C_o}, \eta)]$ for $j = 1, \dotsc, 7$}

(a) Calculation of $\mathbb{E}_P[m_j(Z_i, \tau_{C_o}, \widetilde{\eta}) - m_j(Z_i, \tau_{C_o}, \eta)]$ for $j = 1, 2, 3$.

It is identical to the computation in the proof of \Cref{lemma:tech-remainer-wc}.

(b) Calculation of $\mathbb{E}_P[m_j(Z_i, \tau_{C_o}, \widetilde{\eta}) - m_j(Z_i, \tau_{C_o}, \eta)]$ for $j = 4, 5$.

\begin{align*}
&\mathbb{E}_P[m_4(Z_i, \tau_{C_o}, \widetilde{\eta}) - m_4(Z_i, \tau_{C_o}, \eta)] \\
& = 
\mathbb{E}_P\Bigg[\frac{\mathds{1}_{P_i = O}}{\widetilde{\varphi}} \frac{\widetilde{\varphi}(S_i, X_i)}{1-\widetilde{\varphi}(S_i, X_i)} \frac{1}{\widetilde{\rho}(X_i)}  (\tilde{h}_{C_o, 1}(Y_i, \widetilde{F}_Y^{-1}(\cdot|S_i, X_i, O), \widetilde{\rho}(S_i, X_i)) -\widetilde{\rho}(S_i, X_i) \widetilde{\mu}_{C_o, 1}(S_i, X_i, O))\Bigg]
\\
& \quad - 
\mathbb{E}_P\Bigg[\frac{\mathds{1}_{P_i = O}}{\varphi} \frac{\varphi(S_i, X_i)}{1-\varphi(S_i, X_i)} \frac{1}{\rho(X_i)}  (\tilde{h}_{C_o, 1}(Y_i, F_Y^{-1}(\cdot|S_i, X_i, O), \rho(S_i, X_i)) - \rho(S_i, X_i) \mu_{C_o, 1}(S_i, X_i, O))\Bigg] \\
& = 
\mathbb{E}_P\Bigg[\frac{\mathds{1}_{P_i = O}}{\widetilde{\varphi}} \frac{\widetilde{\varphi}(S_i, X_i)}{1-\widetilde{\varphi}(S_i, X_i)} \frac{1}{\widetilde{\rho}(X_i)}  (\tilde{h}_{C_o, 1}(Y_i, \widetilde{F}_Y^{-1}(\cdot|S_i, X_i, O), \widetilde{\rho}(S_i, X_i)) -\widetilde{\rho}(S_i, X_i) \widetilde{\mu}_{C_o, 1}(S_i, X_i, O))\Bigg]
\\
& \quad - 
\mathbb{E}_P\Bigg[\frac{\mathds{1}_{P_i = O}}{\widetilde{\varphi}} \frac{\widetilde{\varphi}(S_i, X_i)}{1-\widetilde{\varphi}(S_i, X_i)} \frac{1}{\widetilde{\rho}(X_i)}  (\tilde{h}_{C_o, 1}(Y_i, F_Y^{-1}(\cdot|S_i, X_i, O), \rho(S_i, X_i)) - \rho(S_i, X_i) \mu_{C_o, 1}(S_i, X_i, O))\Bigg]
\\
& \quad + 
\mathbb{E}_P\Bigg[\bigg(\frac{\mathds{1}_{P_i = O}}{\widetilde{\varphi}} \frac{\widetilde{\varphi}(S_i, X_i)}{1-\widetilde{\varphi}(S_i, X_i)} \frac{1}{\widetilde{\rho}(X_i)} - \frac{\mathds{1}_{P_i = O}}{\varphi} \frac{\varphi(S_i, X_i)}{1-\varphi(S_i, X_i)} \frac{1}{\rho(X_i)}\bigg)  \\
& \quad \quad \underbrace{\quad \times
(\tilde{h}_{C_o, 1}(Y_i, F_Y^{-1}(\cdot|S_i, X_i, O), \rho(S_i, X_i)) - \rho(S_i, X_i) \mu_{C_o, 1}(S_i, X_i, O))\Bigg]}_{=0 \text{ because of the definition of $\mu_{C_o, 1}(S_i, X_i, O)$ and its dual representation.}} \\
& =
\mathbb{E}_P\Bigg[\frac{\mathds{1}_{P_i = O}}{\widetilde{\varphi}} \frac{\widetilde{\varphi}(S_i, X_i)}{1-\widetilde{\varphi}(S_i, X_i)} \frac{1}{\widetilde{\rho}(X_i)}  \\
& \quad \quad \times (\tilde{h}_{C_o, 1}(Y_i, \widetilde{F}_Y^{-1}(\cdot|S_i, X_i, O), \widetilde{\rho}(S_i, X_i)) - \tilde{h}_{C_o, 1}(Y_i, F_Y^{-1}(\cdot|S_i, X_i, O), \rho(S_i, X_i)))\Bigg]
\\
& \quad - 
\mathbb{E}_P\Bigg[\frac{\mathds{1}_{P_i = O}}{\widetilde{\varphi}} \frac{\widetilde{\varphi}(S_i, X_i)}{1-\widetilde{\varphi}(S_i, X_i)} \frac{1}{\widetilde{\rho}(X_i)}  (\widetilde{\rho}(S_i, X_i) \widetilde{\mu}_{C_o, 1}(S_i, X_i, O) - \rho(S_i, X_i) \mu_{C_o, 1}(S_i, X_i, O))\Bigg]
\\
& =
\mathbb{E}_P\Bigg[\frac{\mathds{1}_{P_i = O}}{\widetilde{\varphi}} \frac{\widetilde{\varphi}(S_i, X_i)}{1-\widetilde{\varphi}(S_i, X_i)} \frac{1}{\widetilde{\rho}(X_i)}  \\
& \quad \quad \times (\tilde{h}_{C_o, 1}(Y_i, \widetilde{F}_Y^{-1}(\cdot|S_i, X_i, O), \widetilde{\rho}(S_i, X_i)) - \tilde{h}_{C_o, 1}(Y_i, F_Y^{-1}(\cdot|S_i, X_i, O), \widetilde{\rho}(S_i, X_i))\Bigg] 
\\
& \quad +
\mathbb{E}_P\Bigg[\frac{\mathds{1}_{P_i = O}}{\widetilde{\varphi}} \frac{\widetilde{\varphi}(S_i, X_i)}{1-\widetilde{\varphi}(S_i, X_i)} \frac{1}{\widetilde{\rho}(X_i)}  \\
& \quad \quad \times (\tilde{h}_{C_o, 1}(Y_i, F_Y^{-1}(\cdot|S_i, X_i, O), \widetilde{\rho}(S_i, X_i)) - \tilde{h}_{C_o, 1}(Y_i, F_Y^{-1}(\cdot|S_i, X_i, O), \rho(S_i, X_i)))\Bigg] 
\\
& \quad - 
\mathbb{E}_P\Bigg[\frac{\mathds{1}_{P_i = O}}{\widetilde{\varphi}} \frac{\widetilde{\varphi}(S_i, X_i)}{1-\widetilde{\varphi}(S_i, X_i)} \frac{1}{\widetilde{\rho}(X_i)}  (\widetilde{\rho}(S_i, X_i)  - \rho(S_i, X_i))\widetilde{\mu}_{C_o, 1}(S_i, X_i, O)\Bigg]
\\
& \quad - 
\mathbb{E}_P\Bigg[\frac{\mathds{1}_{P_i = O}}{\widetilde{\varphi}} \frac{\widetilde{\varphi}(S_i, X_i)}{1-\widetilde{\varphi}(S_i, X_i)} \frac{1}{\widetilde{\rho}(X_i)} \rho(S_i, X_i) (\widetilde{\mu}_{C_o, 1}(S_i, X_i, O) -  \mu_{C_o, 1}(S_i, X_i, O))\Bigg] \\
& =
\mathbb{E}_P\Bigg[\frac{\mathds{1}_{P_i = O}}{\widetilde{\varphi}} \frac{\widetilde{\varphi}(S_i, X_i)}{1-\widetilde{\varphi}(S_i, X_i)} \frac{1}{\widetilde{\rho}(X_i)}  \\
& \quad \quad \times (\tilde{h}_{C_o, 1}(Y_i, \widetilde{F}_Y^{-1}(\cdot|S_i, X_i, O), \widetilde{\rho}(S_i, X_i)) - \tilde{h}_{C_o, 1}(Y_i, F_Y^{-1}(\cdot|S_i, X_i, O), \widetilde{\rho}(S_i, X_i))\Bigg] 
\\
& \quad +
\mathbb{E}_P\Bigg[\frac{\mathds{1}_{P_i = O}}{\widetilde{\varphi}} \bigg(\frac{\widetilde{\varphi}(S_i, X_i)}{1-\widetilde{\varphi}(S_i, X_i)} - \frac{\varphi(S_i, X_i)}{1-\varphi(S_i, X_i)} \bigg) \frac{1}{\widetilde{\rho}(X_i)}  \\
& \quad \quad \times (\tilde{h}_{C_o, 1}(Y_i, F_Y^{-1}(\cdot|S_i, X_i, O), \widetilde{\rho}(S_i, X_i)) - \tilde{h}_{C_o, 1}(Y_i, F_Y^{-1}(\cdot|S_i, X_i, O), \rho(S_i, X_i)))\Bigg] 
\\
& \quad - 
\mathbb{E}_P\Bigg[\frac{\mathds{1}_{P_i = O}}{\widetilde{\varphi}} \bigg(\frac{\widetilde{\varphi}(S_i, X_i)}{1-\widetilde{\varphi}(S_i, X_i)} - \frac{\varphi(S_i, X_i)}{1-\varphi(S_i, X_i)} \bigg) \frac{1}{\widetilde{\rho}(X_i)}  (\widetilde{\rho}(S_i, X_i)  - \rho(S_i, X_i))\widetilde{\mu}_{C_o, 1}(S_i, X_i, O)\Bigg]
\\
& \quad - 
\mathbb{E}_P\Bigg[\frac{\mathds{1}_{P_i = O}}{\widetilde{\varphi}} \bigg(\frac{\widetilde{\varphi}(S_i, X_i)}{1-\widetilde{\varphi}(S_i, X_i)} - \frac{\varphi(S_i, X_i)}{1-\varphi(S_i, X_i)} \bigg) \frac{1}{\widetilde{\rho}(X_i)} \rho(S_i, X_i) (\widetilde{\mu}_{C_o, 1}(S_i, X_i, O) -  \mu_{C_o, 1}(S_i, X_i, O))\Bigg] \\
& \quad +
\mathbb{E}_P\Bigg[\frac{\mathds{1}_{P_i = O}}{\widetilde{\varphi}} \frac{\varphi(S_i, X_i)}{1-\varphi(S_i, X_i)} \frac{1}{\widetilde{\rho}(X_i)}  \\
& \quad \quad \times (\tilde{h}_{C_o, 1}(Y_i, F_Y^{-1}(\cdot|S_i, X_i, O), \widetilde{\rho}(S_i, X_i)) - \tilde{h}_{C_o, 1}(Y_i, F_Y^{-1}(\cdot|S_i, X_i, O), \rho(S_i, X_i)))\Bigg] 
\\
& \quad - 
\mathbb{E}_P\Bigg[\frac{\mathds{1}_{P_i = O}}{\widetilde{\varphi}} \frac{\varphi(S_i, X_i)}{1-\varphi(S_i, X_i)} \frac{1}{\widetilde{\rho}(X_i)}  (\widetilde{\rho}(S_i, X_i)  - \rho(S_i, X_i))\widetilde{\mu}_{C_o, 1}(S_i, X_i, O)\Bigg]
\\
& \quad - 
\mathbb{E}_P\Bigg[\frac{\mathds{1}_{P_i = O}}{\widetilde{\varphi}} \frac{\varphi(S_i, X_i)}{1-\varphi(S_i, X_i)} \frac{1}{\widetilde{\rho}(X_i)} \rho(S_i, X_i) (\widetilde{\mu}_{C_o, 1}(S_i, X_i, O) -  \mu_{C_o, 1}(S_i, X_i, O))\Bigg].
\end{align*}

Similarly,

\begin{align*}
    & \mathbb{E}_P[m_5(Z_i, \tau_{C_o}, \widetilde{\eta}) - m_5(Z_i, \tau_{C_o}, \eta)] \\
    & =
- \mathbb{E}_P\Bigg[\frac{\mathds{1}_{P_i = O}}{\widetilde{\varphi}} \frac{\widetilde{\varphi}(S_i, X_i)}{1-\widetilde{\varphi}(S_i, X_i)} \frac{1}{1-\widetilde{\rho}(X_i)}  \\
& \quad \quad \times (\tilde{h}_{C_o, 0}(Y_i, \widetilde{F}_Y^{-1}(\cdot|S_i, X_i, O), 1-\widetilde{\rho}(S_i, X_i)) - \tilde{h}_{C_o, 0}(Y_i, F_Y^{-1}(\cdot|S_i, X_i, O), 1-\widetilde{\rho}(S_i, X_i))\Bigg] 
\\
& \quad -
\mathbb{E}_P\Bigg[\frac{\mathds{1}_{P_i = O}}{\widetilde{\varphi}} \bigg(\frac{\widetilde{\varphi}(S_i, X_i)}{1-\widetilde{\varphi}(S_i, X_i)} - \frac{\varphi(S_i, X_i)}{1-\varphi(S_i, X_i)} \bigg) \frac{1}{1-\widetilde{\rho}(X_i)}  \\
& \quad \quad \times (\tilde{h}_{C_o, 0}(Y_i, F_Y^{-1}(\cdot|S_i, X_i, O), 1-\widetilde{\rho}(S_i, X_i)) - \tilde{h}_{C_o, 0}(Y_i, F_Y^{-1}(\cdot|S_i, X_i, O), 1-\rho(S_i, X_i)))\Bigg] 
\\
& \quad -
\mathbb{E}_P\Bigg[\frac{\mathds{1}_{P_i = O}}{\widetilde{\varphi}} \bigg(\frac{\widetilde{\varphi}(S_i, X_i)}{1-\widetilde{\varphi}(S_i, X_i)} - \frac{\varphi(S_i, X_i)}{1-\varphi(S_i, X_i)} \bigg) \frac{1}{1-\widetilde{\rho}(X_i)}  (\widetilde{\rho}(S_i, X_i)  - \rho(S_i, X_i))\widetilde{\mu}_{C_o, 0}(S_i, X_i, O)\Bigg]
\\
& \quad +
\mathbb{E}_P\Bigg[\frac{\mathds{1}_{P_i = O}}{\widetilde{\varphi}} \bigg(\frac{\widetilde{\varphi}(S_i, X_i)}{1-\widetilde{\varphi}(S_i, X_i)} - \frac{\varphi(S_i, X_i)}{1-\varphi(S_i, X_i)} \bigg) \frac{1-\rho(S_i, X_i)}{1-\widetilde{\rho}(X_i)} (\widetilde{\mu}_{C_o, 0}(S_i, X_i, O) -  \mu_{C_o, 0}(S_i, X_i, O))\Bigg] \\
& \quad -
\mathbb{E}_P\Bigg[\frac{\mathds{1}_{P_i = O}}{\widetilde{\varphi}} \frac{\varphi(S_i, X_i)}{1-\varphi(S_i, X_i)} \frac{1}{\widetilde{\rho}(X_i)}  \\
& \quad \quad \times (\tilde{h}_{C_o, 0}(Y_i, F_Y^{-1}(\cdot|S_i, X_i, O), 1-\widetilde{\rho}(S_i, X_i)) - \tilde{h}_{C_o, 0}(Y_i, F_Y^{-1}(\cdot|S_i, X_i, O), 1-\rho(S_i, X_i)))\Bigg] 
\\
& \quad -
\mathbb{E}_P\Bigg[\frac{\mathds{1}_{P_i = O}}{\widetilde{\varphi}} \frac{\varphi(S_i, X_i)}{1-\varphi(S_i, X_i)} \frac{1}{1-\widetilde{\rho}(X_i)}  (\widetilde{\rho}(S_i, X_i)  - \rho(S_i, X_i))\widetilde{\mu}_{C_o, 0}(S_i, X_i, O)\Bigg]
\\
& \quad +
\mathbb{E}_P\Bigg[\frac{\mathds{1}_{P_i = O}}{\widetilde{\varphi}} \frac{\varphi(S_i, X_i)}{1-\varphi(S_i, X_i)} \frac{1-\rho(S_i, X_i)}{1-\widetilde{\rho}(X_i)}  (\widetilde{\mu}_{C_o, 0}(S_i, X_i, O) -  \mu_{C_o, 0}(S_i, X_i, O))\Bigg]
\end{align*}

(7) Calculation of $\mathbb{E}_P[m_j(Z_i, \tau_{C_o}, \widetilde{\eta}) - m_j(Z_i, \tau_{C_o}, \eta)]$ for $j = 6, 7$.

\begin{align*}
    & \mathbb{E}_P[m_6(Z_i, \tau_{C_o}, \widetilde{\eta})-m_6(Z_i, \tau_{C_o}, \eta)] \\
    & = 
    \mathbb{E}_P\Bigg[\frac{\mathds{1}_{P_i = E}}{\widetilde{\varphi}} \frac{1}{\widetilde{\rho}(X_i)} (\widetilde{d}_{C_o, 1}(S_i, X_i) - \widetilde{\mu}_{C_o, 1}(S_i, X_i, O))\left(W_i - \widetilde{\rho}(S_i, X_i)\right)\Bigg] \\
    & \quad - 
    \mathbb{E}_P\Bigg[\frac{\mathds{1}_{P_i = E}}{\varphi} \frac{1}{\rho(X_i)} (d_{C_o, 1}(S_i, X_i) - \mu_{C_o, 1}(S_i, X_i, O) )\left(W_i - \rho(S_i, X_i)\right)\Bigg] \\
    & = 
    - 
    \mathbb{E}_P\Bigg[\frac{\mathds{1}_{P_i = E}}{\widetilde{\varphi}} \frac{1}{\widetilde{\rho}(X_i)} (\widetilde{d}_{C_o, 1}(S_i, X_i) - \widetilde{\mu}_{C_o, 1}(S_i, X_i, O)) \left(\widetilde{\rho}(S_i, X_i) - \rho(S_i, X_i)\right)\Bigg] \\
    & \quad +
    \mathbb{E}_P\Bigg[\bigg(\frac{\mathds{1}_{P_i = E}}{\widetilde{\varphi}} \frac{1}{\widetilde{\rho}(X_i)} (\widetilde{d}_{C_o, 1}(S_i, X_i) - \widetilde{\mu}_{C_o, 1}(S_i, X_i, O)) - \frac{\mathds{1}_{P_i = E}}{\varphi} \frac{1}{\rho(X_i)} (d_{C_o, 1}(S_i, X_i)  - \mu_{C_o, 1}(S_i, X_i, O))\bigg) \\
    & \quad \quad \quad \underbrace{ \times \left(W_i - \rho(S_i, X_i)\right)\Bigg] \hspace{12cm} }_{=0 \text{ by the definition of $\rho(S_i, X_i)$.}} \\
    & = 
    - 
    \mathbb{E}_P\Bigg[\frac{\mathds{1}_{P_i = E}}{\widetilde{\varphi}} \frac{1}{\widetilde{\rho}(X_i)} \widetilde{d}_{C_o, 1}(S_i, X_i) \left(\widetilde{\rho}(S_i, X_i) - \rho(S_i, X_i)\right)\Bigg] \\
    & \quad +
    \mathbb{E}_P\Bigg[\frac{\mathds{1}_{P_i = E}}{\widetilde{\varphi}} \frac{1}{\widetilde{\rho}(X_i)} \widetilde{\mu}_{C_o, 1}(S_i, X_i, O) \left(\widetilde{\rho}(S_i, X_i) - \rho(S_i, X_i)\right) \Bigg]
\end{align*}  

Similarly, we have
\begin{align*}
    & \mathbb{E}_P[m_7(Z_i, \tau_{C_o}, \widetilde{\eta})-m_7(Z_i, \tau_{C_o}, \eta)] \\
    & = 
    - 
    \mathbb{E}_P\Bigg[\frac{\mathds{1}_{P_i = E}}{\widetilde{\varphi}} \frac{1}{1-\widetilde{\rho}(X_i)} \widetilde{d}_{C_o, 0}(S_i, X_i) \left(\widetilde{\rho}(S_i, X_i) - \rho(S_i, X_i)\right)\Bigg] \\
    & \quad +
    \mathbb{E}_P\Bigg[\frac{\mathds{1}_{P_i = E}}{\widetilde{\varphi}} \frac{1}{1-\widetilde{\rho}(X_i)} \widetilde{\mu}_{C_o, 0}(S_i, X_i, O) \left(\widetilde{\rho}(S_i, X_i) - \rho(S_i, X_i)\right) \Bigg]
\end{align*}

\textbf{Part 2.} By combining the results in Part 1, we have
\begin{align*}
    \mathbb{E}_P[m(W_i, \tau_{C_o}, \widetilde{\eta}) - m(W_i, \tau_{C_o}, \eta)] 
    & = 
    \sum_{j=1}^7 \mathbb{E}_P[m_j(W_i, \tau_{C_o}, \widetilde{\eta}) - m_j(W_i, \tau_{C_o}, \eta)] = \sum_{j=1}^{14} \mathcal{J}_j,  
\end{align*} 
where
\begin{align*}
    \mathcal{J}_1 & = 
    - \mathbb{E}_P\Bigg[\bigg(\frac{\mathds{1}_{P_i = E}}{\widetilde{\varphi}} \frac{W_i}{\widetilde{\rho}(X_i)} - \frac{\mathds{1}_{P_i = E}}{\widetilde{\varphi}} \frac{W_i}{\rho(X_i)} \bigg)(\widetilde{\bar{\mu}}_{C_+, 1}(1, X_i) - \bar{\mu}_{C_+, 1}(1, X_i)) \Bigg] \\
    \mathcal{J}_2 & = 
    \mathbb{E}_P\Bigg[\bigg(\frac{\mathds{1}_{P_i = E}}{\widetilde{\varphi}} \frac{1-W_i}{1-\widetilde{\rho}(X_i)} - \frac{\mathds{1}_{P_i = E}}{\widetilde{\varphi}} \frac{1-W_i}{1-\rho(X_i)} \bigg) (\widetilde{\bar{\mu}}_{C_+, 0}(1, X_i) - \bar{\mu}_{C_+, 0}(1, X_i)) \Bigg]
    \end{align*}

\begin{align*}
 \mathcal{J}_3   & =
\mathbb{E}_P\Bigg[\frac{\mathds{1}_{P_i = O}}{\widetilde{\varphi}} \frac{\widetilde{\varphi}(S_i, X_i)}{1-\widetilde{\varphi}(S_i, X_i)} \frac{1}{\widetilde{\rho}(X_i)}  \\
& \quad \quad \times (\tilde{h}_{C_o, 1}(Y_i, \widetilde{F}_Y^{-1}(\cdot|S_i, X_i, O), \widetilde{\rho}(S_i, X_i)) - \tilde{h}_{C_o, 1}(Y_i, F_Y^{-1}(\cdot|S_i, X_i, O), \widetilde{\rho}(S_i, X_i))\Bigg] 
\\
\mathcal{J}_4 &  =
\mathbb{E}_P\Bigg[\frac{\mathds{1}_{P_i = O}}{\widetilde{\varphi}} \bigg(\frac{\widetilde{\varphi}(S_i, X_i)}{1-\widetilde{\varphi}(S_i, X_i)} - \frac{\varphi(S_i, X_i)}{1-\varphi(S_i, X_i)} \bigg) \frac{1}{\widetilde{\rho}(X_i)}  \\
& \quad \quad \times (\tilde{h}_{C_o, 1}(Y_i, F_Y^{-1}(\cdot|S_i, X_i, O), \widetilde{\rho}(S_i, X_i)) - \tilde{h}_{C_o, 1}(Y_i, F_Y^{-1}(\cdot|S_i, X_i, O), \rho(S_i, X_i)))\Bigg] 
\\
\mathcal{J}_5 & =  - 
\mathbb{E}_P\Bigg[\frac{\mathds{1}_{P_i = O}}{\widetilde{\varphi}} \bigg(\frac{\widetilde{\varphi}(S_i, X_i)}{1-\widetilde{\varphi}(S_i, X_i)} - \frac{\varphi(S_i, X_i)}{1-\varphi(S_i, X_i)} \bigg) \frac{1}{\widetilde{\rho}(X_i)}  (\widetilde{\rho}(S_i, X_i)  - \rho(S_i, X_i))\widetilde{\mu}_{C_o, 1}(S_i, X_i, O)\Bigg]
\\
\mathcal{J}_6 & = - 
\mathbb{E}_P\Bigg[\frac{\mathds{1}_{P_i = O}}{\widetilde{\varphi}} \bigg(\frac{\widetilde{\varphi}(S_i, X_i)}{1-\widetilde{\varphi}(S_i, X_i)} - \frac{\varphi(S_i, X_i)}{1-\varphi(S_i, X_i)} \bigg) \frac{1}{\widetilde{\rho}(X_i)} \rho(S_i, X_i) (\widetilde{\mu}_{C_o, 1}(S_i, X_i, O) -  \mu_{C_o, 1}(S_i, X_i, O))\Bigg] \\
\mathcal{J}_7 & = 
\mathbb{E}_P\Bigg[\frac{\mathds{1}_{P_i = O}}{\widetilde{\varphi}} \frac{\varphi(S_i, X_i)}{1-\varphi(S_i, X_i)} \frac{1}{\widetilde{\rho}(X_i)}  \\
& \quad \quad \times (\tilde{h}_{C_o, 1}(Y_i, F_Y^{-1}(\cdot|S_i, X_i, O), \widetilde{\rho}(S_i, X_i)) - \tilde{h}_{C_o, 1}(Y_i, F_Y^{-1}(\cdot|S_i, X_i, O), \rho(S_i, X_i)))\Bigg] 
\end{align*}

\begin{align*}
\mathcal{J}_8 & = 
- \mathbb{E}_P\Bigg[\frac{\mathds{1}_{P_i = O}}{\widetilde{\varphi}} \frac{\widetilde{\varphi}(S_i, X_i)}{1-\widetilde{\varphi}(S_i, X_i)} \frac{1}{1-\widetilde{\rho}(X_i)}  \\
& \quad \quad \times (\tilde{h}_{C_o, 0}(Y_i, \widetilde{F}_Y^{-1}(\cdot|S_i, X_i, O), 1-\widetilde{\rho}(S_i, X_i)) - \tilde{h}_{C_o, 0}(Y_i, F_Y^{-1}(\cdot|S_i, X_i, O), 1-\widetilde{\rho}(S_i, X_i))\Bigg] 
\\
\mathcal{J}_9 & =  -
\mathbb{E}_P\Bigg[\frac{\mathds{1}_{P_i = O}}{\widetilde{\varphi}} \bigg(\frac{\widetilde{\varphi}(S_i, X_i)}{1-\widetilde{\varphi}(S_i, X_i)} - \frac{\varphi(S_i, X_i)}{1-\varphi(S_i, X_i)} \bigg) \frac{1}{1-\widetilde{\rho}(X_i)}  \\
& \quad \quad \times (\tilde{h}_{C_o, 0}(Y_i, F_Y^{-1}(\cdot|S_i, X_i, O), 1-\widetilde{\rho}(S_i, X_i)) - \tilde{h}_{C_o, 0}(Y_i, F_Y^{-1}(\cdot|S_i, X_i, O), 1-\rho(S_i, X_i)))\Bigg] 
\\
\mathcal{J}_{10} & =  -
\mathbb{E}_P\Bigg[\frac{\mathds{1}_{P_i = O}}{\widetilde{\varphi}} \bigg(\frac{\widetilde{\varphi}(S_i, X_i)}{1-\widetilde{\varphi}(S_i, X_i)} - \frac{\varphi(S_i, X_i)}{1-\varphi(S_i, X_i)} \bigg) \frac{1}{1-\widetilde{\rho}(X_i)}  (\widetilde{\rho}(S_i, X_i)  - \rho(S_i, X_i))\widetilde{\mu}_{C_o, 0}(S_i, X_i, O)\Bigg]
\\
\mathcal{J}_{11} & = 
\mathbb{E}_P\Bigg[\frac{\mathds{1}_{P_i = O}}{\widetilde{\varphi}} \bigg(\frac{\widetilde{\varphi}(S_i, X_i)}{1-\widetilde{\varphi}(S_i, X_i)} - \frac{\varphi(S_i, X_i)}{1-\varphi(S_i, X_i)} \bigg) \frac{1-\rho(S_i, X_i)}{1-\widetilde{\rho}(X_i)} (\widetilde{\mu}_{C_o, 0}(S_i, X_i, O) -  \mu_{C_o, 0}(S_i, X_i, O))\Bigg] \\
\mathcal{J}_{12} & =  -
\mathbb{E}_P\Bigg[\frac{\mathds{1}_{P_i = O}}{\widetilde{\varphi}} \frac{\varphi(S_i, X_i)}{1-\varphi(S_i, X_i)} \frac{1}{\widetilde{\rho}(X_i)}  \\
& \quad \quad \times (\tilde{h}_{C_o, 0}(Y_i, F_Y^{-1}(\cdot|S_i, X_i, O), 1-\widetilde{\rho}(S_i, X_i)) - \tilde{h}_{C_o, 0}(Y_i, F_Y^{-1}(\cdot|S_i, X_i, O), 1-\rho(S_i, X_i)))\Bigg] 
\end{align*}

\begin{align*}
    \mathcal{J}_{13} & = 
    - 
    \mathbb{E}_P\Bigg[\frac{\mathds{1}_{P_i = E}}{\widetilde{\varphi}} \frac{1}{\widetilde{\rho}(X_i)} \widetilde{d}_{C_o, 1}(S_i, X_i) \left(\widetilde{\rho}(S_i, X_i) - \rho(S_i, X_i)\right)\Bigg] \\
    \mathcal{J}_{14} & =  - 
    \mathbb{E}_P\Bigg[\frac{\mathds{1}_{P_i = E}}{\widetilde{\varphi}} \frac{1}{1-\widetilde{\rho}(X_i)} \widetilde{d}_{C_o, 0}(S_i, X_i) \left(\widetilde{\rho}(S_i, X_i) - \rho(S_i, X_i)\right)\Bigg].
\end{align*}
\end{proof}

\subsubsection{Proof of \Cref{lemma:DML-verification-copula}}

\begin{proof}[Proof of \Cref{lemma:DML-verification-copula}]
The proof is almost identical to the proof of \Cref{lemma:DML-verification-wc}. Therefore, we look at \cref{eq:Lq-m,eq:rate-diff-L2,eq:rate-diff-m} and asymptotic variance.

With abuse of the notation, we define 
{\footnotesize
\begin{align*}
        m_1(Z_i, \tau, \widetilde{\eta}) & = \frac{\mathds{1}_{P_i = E}}{\widetilde{\varphi}} \frac{W_i}{\widetilde{\rho}(X_i)} (\widetilde{\mu}_{C_o, 1}(S_i, X_i, O) - \widetilde{\bar{\mu}}_{C_0, 1}(1, X_i)), \\
        m_2(Z_i, \tau, \widetilde{\eta}) & = -\frac{\mathds{1}_{P_i = E}}{\widetilde{\varphi}} \frac{1-W_i}{1-\widetilde{\rho}(X_i)} (\widetilde{\mu}_{C_o, o}(S_i, X_i, O) - \widetilde{\bar{\mu}}_{C_o, 0}(0, X_i)), \\
        m_3(Z_i, \tau, \widetilde{\eta}) & = \frac{\mathds{1}_{P_i = E}}{\widetilde{\varphi}}(\widetilde{\bar{\mu}}_{C_o, 1}(1, X_i) - \widetilde{\bar{\mu}}_{C_o, 0}(0, X_i) - \tau) \\
        m_4(Z_i, \tau, \widetilde{\eta}) 
        & = \frac{\mathds{1}_{P_i = O}}{\widetilde{\varphi}} \frac{\widetilde{\varphi}(S_i, X_i)}{1-\widetilde{\varphi}(S_i, X_i)} \frac{\widetilde{\rho}(S_i, X_i)}{\widetilde{\rho}(X_i)} (h_{C_o, 1}(Y_i, \widetilde{F}_Y^{-1}(\cdot|S_i, X_i, O), 1 - \widetilde{\rho}(S_i, X_i)) -\widetilde{\mu}_{C_o, 1}(S_i, X_i, O)), \\
        & = \frac{\mathds{1}_{P_i = O}}{\widetilde{\varphi}} \frac{\widetilde{\varphi}(S_i, X_i)}{1-\widetilde{\varphi}(S_i, X_i)} \frac{1}{\widetilde{\rho}(X_i)} (\widetilde{h}_{C_o, 1}(Y_i, \widetilde{F}_Y^{-1}(\cdot|S_i, X_i, O), \widetilde{\rho}(S_i, X_i)) -\widetilde{\rho}(S_i, X_i) \widetilde{\mu}_{C_o, 1}(S_i, X_i, O)),
        \\
        m_5(Z_i, \tau, \widetilde{\eta}) & = - \frac{\mathds{1}_{P_i = O}}{\widetilde{\varphi}} \frac{\widetilde{\varphi}(S_i, X_i)}{1-\widetilde{\varphi}(S_i, X_i)}\frac{1-\widetilde{\rho}(S_i, X_i)}{1-\widetilde{\rho}(X_i)} (h_{C_o, 0}(Y_i, \widetilde{F}_Y^{-1}(\cdot|S_i, X_i, O), 1 - \widetilde{\rho}(S_i, X_i)) -\widetilde{\mu}_{C_o, 0}(S_i, X_i, O)), \\
        & = - \frac{\mathds{1}_{P_i = O}}{\widetilde{\varphi}} \frac{\widetilde{\varphi}(S_i, X_i)}{1-\widetilde{\varphi}(S_i, X_i)} \frac{1}{1-\widetilde{\rho}(X_i)}  \\
    & \quad \times (\tilde{h}_{C_o, 0}(Y_i, F_Y^{-1}(\cdot|S_i, X_i, O), 1-\rho(S_i, X_i)) -(1-\widetilde{\rho}(S_i, X_i))\widetilde{\mu}_{C_o, 0}(S_i, X_i, O)), \\
        m_6(Z_i, \tau, \widetilde{\eta}) & = \frac{\mathds{1}_{P_i = E}}{\widetilde{\varphi}} \frac{1}{\widetilde{\rho}(X_i)} (\widetilde{d}_{C_o, 1}(S_i, X_i) - \mu_{C_o, 1}(S_i, X_i, O)) \left(W_i - \widetilde{\rho}(S_i, X_i)\right), \\
        m_7(Z_i, \tau, \eta) & = \frac{\mathds{1}_{P_i = E}}{\widetilde{\varphi}} \frac{1}{1 - \widetilde{\rho}(X_i)} (\widetilde{d}_{C_o, 0}(S_i, X_i) - \mu_{C_o, 0}(S_i, X_i, O))\left(W_i - \widetilde{\rho}(S_i, X_i)\right).
    \end{align*}
}
where
{\footnotesize
\begin{align*}
    & \tilde{h}_{C_o, 1}(Y_i, F_Y^{-1}(\cdot|S_i, X_i, O), \widetilde{\rho}(S_i, X_i)) \\
    & = (1-C_o(1-\widetilde{\rho}(S_i, X_i)|0)) Y_i + \int_0^1 \left((1-u)F_Y^{-1}(u|S_i, X_i, O) + [Y_i - F_Y^{-1}(u|S_i, X_i, O)]_{+}\right)d(1-C_o(1-\widetilde{\rho}(S_i, X_i)|u), \\
    & \tilde{h}_{C_o, 0}(Y_i, F_Y^{-1}(\cdot|S_i, X_i, O), 1-\widetilde{\rho}(S_i, X_i)) \\
    & = C_o(1-\widetilde{\rho}(S_i, X_i)|1) Y_i  - \int_0^1 \left(u F_Y^{-1}(u|S_i, X_i, O) - [Y_i - F_Y^{-1}(u|S_i, X_i, O)]_{-}\right)d C_o(1-\widetilde{\rho}(S_i, X_i)|u)
\end{align*}
}
\textbf{Verification of \cref{eq:Lq-m}}

We focus on $\|m_j(Z_i, \tau_{C_o}, \widetilde{\eta}\|_{P, q} \le T_{\epsilon}$ for $j = 4, \dotsc, 7$ since \cref{lemma:DML-verification-wc} shows that $\|m_j(Z_i, \tau_{C_o}, \widetilde{\eta}\|_{P, q} \le T_{\epsilon}$ for $j = 1, 2, 3$. 

Under \cref{assumption:realization-set-copula,assumption:realization-set-copula} with boundedness of $Y$, we can show that
\begin{align*}
    & \| m_4(Z_i, \tau_{C_o}, \widetilde{\eta} \|_{P, q} \\
    & \le 
    T_{\epsilon} 
    \| \mathds{1}_{P_i = E} | \widetilde{h}_{C_o, 1}(Y_i, \widetilde{F}_Y^{-1}(\cdot|S_i, X_i, O), 1-\widetilde{\rho}(S_i, X_i)) - \widetilde{h}_{C_o, 1}(Y_i, F_Y^{-1}(\cdot|S_i, X_i, O), 1-\widetilde{\rho}(S_i, X_i)) |\|_{P, q} \\
    & \quad + T_{\epsilon} \| \mathds{1}_{P_i = E} \widetilde{h}_{C_o, 1}(Y_i, F_Y^{-1}(\cdot|S_i, X_i, O), 1-\widetilde{\rho}(S_i, X_i))| \|_{P, q} \\
    & \quad + T_{\epsilon} \| \widetilde{\mu}_{C_o, 1}(S_i, X_i, O) - \mu_{C_o, 1}(S_i, X_i, O) \|_{P, q} +  T_{\epsilon} \|  \mu_{C_o, 1}(S_i, X_i, O) \|_{P, q} \\
    & \le 
    T_{\epsilon} 
    \| \int_0^1 | \widetilde{F}^{-1}_Y(u|S_i, X_i, O) - \widetilde{F}^{-1}_Y(u|S_i, X_i, O) | d C_o(1-\widetilde{\rho}(S_i, X_i)|u) \|_{P, q} \\
    & \quad + T_{\epsilon} \| \widetilde{h}_{C_o, 1}(Y_i, F_Y^{-1}(\cdot|S_i, X_i, O), 1-\widetilde{\rho}(S_i, X_i))| \|_{P, q} \\
    & \quad +  T_{\epsilon} \| \widetilde{\mu}_{C_o, 1}(S_i, X_i, O) - \mu_{C_o, 1}(S_i, X_i, O) \|_{P, q} +  T_{\epsilon} \|  \mu_{C_o, 1}(S_i, X_i, O) \|_{P, q} \\
    & \le T_{\epsilon}.
\end{align*}
We can show $\|m_4(Z_i, \tau_{C_o}, \widetilde{\eta} \|_{P, q} \le T_{\epsilon}$ in a similar way.

\begin{align*}
    \| m_6(Z_i, \tau_{C_o}, \widetilde{\eta} \|_{P, q} 
    & \le T_{\epsilon} \| \widetilde{d}_{C_o, 1}(S_i, X_i) - \widetilde{\mu}_{C_o, 1}(S_i, X_i) \|_{P, q} \\
    & \le T_{\epsilon} \| \int_0^1 F_Y^{-1}(u|S_i, X_i, O) c_o(1-\widetilde{\rho}(S_i, X_i|u)du - \mu_{C_o, 1}(S_i, X_i) \|_{P, q} \\
    & \quad + T_{\epsilon} \| \int_0^1 \left| F_Y^{-1}(u|S_i, X_i, O) - F_Y^{-1}(u|S_i, X_i, O) \right| du  \|_{P, q} \\
    & \quad + T_{\epsilon} \|  \widetilde{\mu}_{C_o, 1}(S_i, X_i) - \mu_{C_o, 1}(S_i, X_i) \|_{P, q} \\
    & \le T_{\epsilon}.
\end{align*}    
We can show $\|m_7(Z_i, \tau_{C_o}, \widetilde{\eta} \|_{P, q} \le T_{\epsilon}$ in a similar way.

\textbf{Verification of \cref{eq:rate-diff-L2}}

We consider $\|m_j(Z_i, \tau_{C_o}, \widetilde{\eta}) - m_j(Z_i, \tau_{C_o}, \eta) \|_{P. 2}$ for $j=4,5,6,7$ because $\| \|_{P. 2}$ for $j=1, 2, 3$ are identical.

 Let's define $V_1(S_i, X_i)$ and $V_2(S_i, X_i)$ by
    \begin{gather*}
        \frac{1}{V_1(S_i, X_i)}  =  \frac{1}{{\varphi}} \frac{{\varphi}(S_i, X_i)}{1-{\varphi}(S_i, X_i)} \frac{1}{{\rho}(X_i)}, \quad 
        \frac{1}{\widehat{V}_1(S_i, X_i)}  =  \frac{1}{\hat{\varphi}} \frac{\hat{\varphi}(S_i, X_i)}{1-\hat{\varphi}(S_i, X_i)} \frac{1}{\hat{\rho}(X_i)} \\
        V_2(S_i, X_i) = V_1(S_i, X_i) / \rho(S_i, X_i), \quad 
        \widehat{V}_2(S_i, X_i) = \widehat{V}_1(S_i, X_i) / \hat{\rho}(S_i, X_i).
    \end{gather*}
    Under our assumption, we have $\min{V_1(S_i, X_i), \widehat{V}_1(S_i, X_i)} \ge \epsilon^3 / (1 - \epsilon) := \epsilon_1$ and $\min{V_2(S_i, X_i), \widehat{V}_2(S_i, X_i)} \ge \epsilon^3 / (1 - \epsilon)^2 := \epsilon_2$. Mean-value theorem implies that 
    \begin{align*}
        | \widehat{V}_1(S_i, X_i) - {V}_1(S_i, X_i) | & \le T_{\epsilon} \left(|\hat{\varphi} - \varphi| + |\hat{\varphi}(S_i, X_i) - \varphi(S_i, X_i)| + |\hat{\rho}(X_i) - \rho(X_i)|  \right) \\
         | \widehat{V}_2(S_i, X_i) - {V}_2(S_i, X_i) | & \le T_{\epsilon} \Big(|\hat{\varphi} - \varphi| + |\hat{\varphi}(S_i, X_i) - \varphi(S_i, X_i)| \\
         & \quad \quad + |\hat{\rho}(X_i) - \rho(X_i)| + |\hat{\rho}(S_i, X_i) - \rho(S_i, X_i)| \Big)
    \end{align*}
    where $T_{\epsilon}$ is a positive number which only depends on $\epsilon$.
    
    Then, the bound for $\|m_4(Z_i, \tau_{C_o}, \widetilde{\eta}) - m_4(Z_i, \tau_{C_o}, \eta)\|_{P, 2}$ is given as follows.
    {\footnotesize
    \begin{align*}
        & \|m_4(Z_i, \tau_{C_o}, \widetilde{\eta}) - m_4(Z_i, \tau_{C_o}, \eta)\|_{P, 2} \\
        & = 
        \Biggr\|\frac{\mathds{1}_{P_i = O}}{\widetilde{\varphi}} \frac{\widetilde{\varphi}(S_i, X_i)}{1-\widetilde{\varphi}(S_i, X_i)} \frac{\widetilde{\rho}(S_i, X_i)}{\widetilde{\rho}(X_i)} (h_{C_o, 1}(Y_i, \widetilde{F}_Y^{-1}(\cdot|S_i, X_i, O), 1 - \widetilde{\rho}(S_i, X_i)) -\widetilde{\mu}_{C_o, 1}(S_i, X_i, O)) \\
        & \quad 
        -
        \frac{\mathds{1}_{P_i = O}}{{\varphi}} \frac{{\varphi}(S_i, X_i)}{1-{\varphi}(S_i, X_i)} \frac{{\rho}(S_i, X_i)}{{\rho}(X_i)} (h_{C_o, 1}(Y_i, F_Y^{-1}(\cdot | S_i, X_i, O), 1 - {\rho}(S_i, X_i)) - \mu_{C_o, 1}(S_i, X_i, O))
        \Biggr\|_{P, 2}\\
        & \le
        \Biggr\|\frac{\mathds{1}_{P_i = O}}{\widetilde{V}_1(S_i, X_i)}  \widetilde{h}_{C_o, 1}(Y_i, \widetilde{F}_Y^{-1}(\cdot|S_i, X_i, O), \widetilde{\rho}(S_i, X_i)) - \frac{\mathds{1}_{P_i = O}}{V_1(S_i, X_i)}  h_{C_o, 1}(Y_i, F_Y^{-1}(\cdot|S_i, X_i, O), \rho(S_i, X_i)))  \Biggr\|_{P, 2} \\
        & \quad 
        +
        \Biggr\|
        \frac{\mathds{1}_{P_i = O}}{\widetilde{V}_2(S_i, X_i)} \widetilde{\mu}_{C_o, 1}(S_i, X_i, O)) -\frac{\mathds{1}_{P_i = O}}{V_2(S_i, X_i)} \mu_{C_o, 1}(S_i, X_i, O))
        \Biggr\|_{P, 2}\\
        & \le \epsilon_2^{-2}
        \Biggr(\mathbb{E}\Bigg[\mathds{1}_{P_i = O}\bigg|V_1(S_i, X_i) \widetilde{h}_{C_o, 1}(Y_i, \widetilde{F}_Y^{-1}(\cdot|S_i, X_i, O), \widetilde{\rho}(S_i, X_i) - 
        \widetilde{V}_1(S_i, X_i) \widetilde{h}_{C_o, 1}(Y_i, F_Y^{-1}(\cdot | S_i, X_i, O), \rho(S_i, X_i))
        \bigg|^2\Bigg]\Biggr)^{1/2} \\
        & \quad +
        \epsilon_2^{-2}
        \Biggr(\mathbb{E}\Bigg[\mathds{1}_{P_i = O}\bigg|V_2(S_i, X_i) \widetilde{\mu}_{C_o, 1}(S_i, X_i, O))  -
        \widetilde{V}_2(S_i, X_i) \mu_{C_o, 1}(S_i, X_i, O))
        \bigg|^2\Bigg]\Biggr)^{1/2} \\
        & \le
        \epsilon_2^{-2}
        \Biggr(\mathbb{E}\Bigg[\mathds{1}_{P_i = O} \bigg|V_1(S_i, X_i) \bigg(\widetilde{h}_{C_o, 1}(Y_i, \widetilde{F}_Y^{-1}(\cdot|S_i, X_i, O), \widetilde{\rho}(S_i, X_i)) 
        - \widetilde{h}_{C_o, 1}(Y_i, {F}_Y^{-1}(\cdot | S_i, X_i, O),{\rho}(S_i, X_i))\bigg) 
        \bigg|^2\Bigg]\Biggr)^{1/2} \\
         & \quad +
        \epsilon_2^{-2}
        \Biggr(\mathbb{E}\Bigg[\mathds{1}_{P_i = O} \bigg|(V_1(S_i, X_i)
        -
        \widehat{V}_1(S_i, X_i)) \widetilde{h}_{C_o, 1}(Y_i, F_{Y}^{-1}(\cdot | S_i, X_i, O), \rho(S_i, X_i))
        \bigg|^2\Bigg]\Biggr)^{1/2} \\
        & \quad +
        \epsilon_2^{-2}
        \Biggr(\mathbb{E}\Bigg[\bigg|V_2(S_i, X_i) \bigg(\widetilde{\mu}_{C_o, 1}(S_i, X_i, O)  - \mu_{C_o, 1}(S_i, X_i, O\bigg) 
        \bigg|^2\Bigg]\Biggr)^{1/2} \\
         & \quad +
        \epsilon_2^{-2}
        \Biggr(\mathbb{E}\Bigg[ \bigg|(V_2(S_i, X_i) -
        \widehat{V}_2(S_i, X_i)) \mu_{C_o, 1}(S_i, X_i, O))
        \bigg|^2\Bigg]\Biggr)^{1/2} \\
        & \le
        \epsilon_2^{-2}
        \Biggr(\mathbb{E}\Bigg[\mathds{1}_{P_i = O} \bigg|V_1(S_i, X_i) \bigg(\widetilde{h}_{C_o, 1}(Y_i, \widetilde{F}_Y^{-1}(\cdot|S_i, X_i, O), \widetilde{\rho}(S_i, X_i)) 
        - \widetilde{h}_{C_o, 1}(Y_i, F_Y^{-1}(\cdot | S_i, X_i, O), \widetilde{\rho}(S_i, X_i))\bigg) 
        \bigg|^2\Bigg]\Biggr)^{1/2} \\
        & \quad + 
        \epsilon_2^{-2}
        \Biggr(\mathbb{E}\Bigg[\mathds{1}_{P_i = O} \bigg|V_1(S_i, X_i) \bigg(\widetilde{h}_{C_o, 1}(Y_i, F_Y^{-1}(\cdot|S_i, X_i, O), \widetilde{\rho}(S_i, X_i)) 
        - \widetilde{h}_{C_o, 1}(Y_i, F_Y^{-1}(\cdot | S_i, X_i, O),\rho(S_i, X_i))\bigg) 
        \bigg|^2\Bigg]\Biggr)^{1/2} \\
         & \quad +
        \epsilon_2^{-2}
        \Biggr(\mathbb{E}\Bigg[\mathds{1}_{P_i = O} \bigg|(V_1(S_i, X_i)
        -
        \widehat{V}_1(S_i, X_i)) \widetilde{h}_{C_o, 1}(Y_i, F_Y^{-1}(\cdot | S_i, X_i, O), \rho(S_i, X_i))
        \bigg|^2\Bigg]\Biggr)^{1/2} \\
        & \quad +
        \epsilon_2^{-2}
        \Biggr(\mathbb{E}\Bigg[\bigg|V_2(S_i, X_i) \bigg(\widetilde{\mu}_{C_o, 1}(S_i, X_i, O)  - \mu_{C_o, 1}(S_i, X_i, O\bigg) 
        \bigg|^2\Bigg]\Biggr)^{1/2} \\
         & \quad +
        \epsilon_2^{-2}
        \Biggr(\mathbb{E}\Bigg[ \bigg|(V_2(S_i, X_i) -
        \widehat{V}_2(S_i, X_i)) \mu_{C_o, 1}(S_i, X_i, O))
        \bigg|^2\Bigg]\Biggr)^{1/2} \\
        & \le
        \epsilon_2^{-2}
        \Biggr(\mathbb{E}\Bigg[\mathds{1}_{P_i = O} \bigg|V_1(S_i, X_i) \bigg(\int_0^1 |\widetilde{F}_Y^{-1}(u|S_i, X_i, O)- F_Y^{-1}(u|S_i, X_i, O)| d (1 - C_o(1-\widetilde{\rho}(S_i, X_i)) )\bigg) 
        \bigg|^2\Bigg]\Biggr)^{1/2} \\
        & \quad + 
        \epsilon_2^{-2} C_{\epsilon}
        \Biggr(\mathbb{E}\Bigg[\mathds{1}_{P_i = O} \bigg|V_1(S_i, X_i) \bigg(|\widetilde{\rho}(S_i, X_i) - \rho(S_i, X_i)\bigg) 
        \bigg|^2\Bigg]\Biggr)^{1/2} \\
         & \quad +
        \epsilon_2^{-2}
        \Biggr(\mathbb{E}\Bigg[\mathds{1}_{P_i = O} \bigg|(V_1(S_i, X_i)
        -
        \widehat{V}_1(S_i, X_i)) \widetilde{h}_{C_o, 1}(Y_i, F_Y^{-1}(\cdot | S_i, X_i, O), \rho(S_i, X_i))
        \bigg|^2\Bigg]\Biggr)^{1/2} \\
        & \quad +
        \epsilon_2^{-2}
        \Biggr(\mathbb{E}\Bigg[\bigg|V_2(S_i, X_i) \bigg(\widetilde{\mu}_{C_o, 1}(S_i, X_i, O)  - \mu_{C_o, 1}(S_i, X_i, O\bigg) 
        \bigg|^2\Bigg]\Biggr)^{1/2} \\
         & \quad +
        \epsilon_2^{-2}
        \Biggr(\mathbb{E}\Bigg[ \bigg|(V_2(S_i, X_i) -
        \widehat{V}_2(S_i, X_i)) \mu_{C_o, 1}(S_i, X_i, O))
        \bigg|^2\Bigg]\Biggr)^{1/2} \\
        & \le T_4 \delta_n
    \end{align*}
    }
    where $T_4$ depends on $\epsilon$ and $T$ only. Note that in the last inequality, we can show that
    \begin{align*}
        |\widetilde{h}_{C_o, 1}(Y_i, F_Y^{-1}(\cdot|S_i, X_i, O), \widetilde{\rho}(S_i, X_i))
        - \widetilde{h}_{C_o, 1}(Y_i, F_Y^{-1}(\cdot | S_i, X_i, O),\rho(S_i, X_i))| \le T_{\epsilon} |\widetilde{\rho}(S_i, X_i) - \rho(S_i, X_i)|
    \end{align*}
    under \cref{assumption:regularity-copula}. We can similarly show that $\left(\mathbb{E}[|m_5(1) - m_5(0)|^2]\right)^{1/2} < T_5 \delta_n$ where $T_5 > 0$ depends only on $T$ and $\epsilon$. 

    The bound for $\| m_6(Z_i, \tau_{C_o}, \widetilde{\eta}) - m_6(Z_i, \tau_{C_o}, \eta) \|_{P, 2}$ is given as follows.
    \begin{align*}
        & \| m_6(Z_i, \tau_{C_o}, \widetilde{\eta}) - m_6(Z_i, \tau_{C_o}, \eta) \|_{P, 2} \\
        & =
        \Biggr(\mathbb{E}\Bigg[\bigg|\frac{\mathds{1}_{P_i = E}}{\widetilde{\varphi}} \frac{\widetilde{\rho}(S_i, X_i)}{\widetilde{\rho}(X_i)} (\widetilde{d}_{C_o}(S_i, X_i, O) - \widetilde{\mu}_{C_o, 1}(S_i, X_i, O)) \left(W_i - \widetilde{\rho}(S_i, X_i)\right) \\
        & \quad - \frac{\mathds{1}_{P_i = E}}{\varphi} \frac{\rho(S_i, X_i)}{\rho(X_i)} (d_{C_o, 1}(S_i, X_i, O) - \mu_{C_o, 1}(S_i, X_i, O)\left(W_i - \rho(S_i, X_i)\right)\bigg|^2\Bigg]\Biggr)^{1/2} \\
        & \le 
        \Biggr(\mathbb{E}\Bigg[\bigg|\frac{\mathds{1}_{P_i = E}}{\widetilde{\varphi}} \frac{\widetilde{\rho}(S_i, X_i)}{\widetilde{\rho}(X_i)} \widehat{d}_{C_o, 1}(S_i, X_i, O)  - \frac{\mathds{1}_{P_i = E}}{\varphi} \frac{\rho(S_i, X_i)}{\rho(X_i)} d_{C_o, 1}(S_i, X_i, O) \bigg|^2\Bigg]\Biggr)^{1/2} \\
        & \quad + 
        \Biggr(\mathbb{E}\Bigg[\bigg|\frac{\mathds{1}_{P_i = E}}{\widetilde{\varphi}} \frac{\widetilde{\rho}^2(S_i, X_i)}{\widetilde{\rho}(X_i)} \widehat{d}_{C_o, 1}(S_i, X_i, O)  - \frac{\mathds{1}_{P_i = E}}{\varphi} \frac{\rho^2(S_i, X_i)}{\rho(X_i)} d_{C_o, 1}(S_i, X_i, O) \bigg|^2\Bigg]\Biggr)^{1/2} \\
        & +
        \Biggr(\mathbb{E}\Bigg[\bigg|\frac{\mathds{1}_{P_i = E}}{\widetilde{\varphi}} \frac{\widetilde{\rho}(S_i, X_i)}{\widetilde{\rho}(X_i)} \widetilde{\mu}_{C_o, 1}(S_i, X_i, O)  - \frac{\mathds{1}_{P_i = E}}{\varphi} \frac{\rho(S_i, X_i)}{\rho(X_i)} \mu_{C_o, 1}(S_i, X_i, O) \bigg|^2\Bigg]\Biggr)^{1/2} \\
        & \quad + 
        \Biggr(\mathbb{E}\Bigg[\bigg|\frac{\mathds{1}_{P_i = E}}{\widetilde{\varphi}} \frac{\widetilde{\rho}^2(S_i, X_i)}{\widetilde{\rho}(X_i)} \widetilde{\mu}_{C_o, 1}(S_i, X_i, O)  - \frac{\mathds{1}_{P_i = E}}{\varphi} \frac{\rho^2(S_i, X_i)}{\rho(X_i)} \mu_{C_o, 1}(S_i, X_i, O) \bigg|^2\Bigg]\Biggr)^{1/2}
        \\
        & \le
         ((1-\epsilon)/\epsilon^2)^{2}
        \Biggr(\mathbb{E}\Bigg[\bigg|\ \frac{\varphi \rho(X_i)}{\rho(S_i, X_i)} \widetilde{d}_{C_o, 1}(S_i, X_i, O)   -  
        \frac{\widetilde{\varphi}\widetilde{\rho}(X_i)}{\widetilde{\rho}(S_i, X_i)} d_{C_o, 1}(S_i, X_i, O)\bigg|^2\Bigg]\Biggr)^{1/2} \\
        & \quad + ((1-\epsilon)^2/\epsilon^2)^{2}
        \Biggr(\mathbb{E}\Bigg[\bigg|\ \frac{\varphi \rho(X_i)}{\rho^2(S_i, X_i)} \widetilde{d}_{C_o, 1}(S_i, X_i, O)   -  
        \frac{\widetilde{\varphi}\widetilde{\rho}(X_i)}{\widetilde{\rho}^2(S_i, X_i)} d_{C_o, 1}(S_i, X_i, O)\bigg|^2\Bigg]\Biggr)^{1/2} \\
         & +
         ((1-\epsilon)/\epsilon^2)^{2}
        \Biggr(\mathbb{E}\Bigg[\bigg|\ \frac{\varphi \rho(X_i)}{\rho(S_i, X_i)} \widetilde{\mu}_{C_o, 1}(S_i, X_i, O)   -  
        \frac{\widetilde{\varphi}\widetilde{\rho}(X_i)}{\widetilde{\rho}(S_i, X_i)} \mu_{C_o, 1}(S_i, X_i, O)\bigg|^2\Bigg]\Biggr)^{1/2} \\
        & \quad + ((1-\epsilon)^2/\epsilon^2)^{2}
        \Biggr(\mathbb{E}\Bigg[\bigg|\ \frac{\varphi \rho(X_i)}{\rho^2(S_i, X_i)} \widetilde{\mu}_{C_o, 1}(S_i, X_i, O)   -  
        \frac{\widetilde{\varphi}\widetilde{\rho}(X_i)}{\widetilde{\rho}^2(S_i, X_i)} \mu_{C_o, 1}(S_i, X_i, O)\bigg|^2\Bigg]\Biggr)^{1/2} \\
        & \le 
         ((1-\epsilon)/\epsilon^2)^{2}
        \Biggr(\mathbb{E}\Bigg[\bigg| \frac{\varphi \rho(X_i)}{\rho(S_i, X_i)} \left[\widetilde{d}_{C_o,1}(S_i, X_i, O) -  d_{C_o,1}(S_i, X_i, O) \right] \bigg|^2\Bigg]\Biggr)^{1/2} \\
        & \quad + ((1-\epsilon)/\epsilon^2)^{2}
        \Biggr(\mathbb{E}\Bigg[\bigg| \left(\frac{\varphi \rho(X_i)}{\rho(S_i, X_i)} -  
        \frac{\widetilde{\varphi}\widetilde{\rho}(X_i)}{\widetilde{\rho}(S_i, X_i)} \right) d_{C_o, 1}(S_i, X_i, O) \bigg|^2\Bigg]\Biggr)^{1/2} \\
        & \quad + 
        ((1-\epsilon)^2/\epsilon^2)^{2}
        \Biggr(\mathbb{E}\Bigg[\bigg| \frac{\varphi \rho(X_i)}{\rho^2(S_i, X_i)} \left[\widetilde{d}_{C_o,1}(S_i, X_i, O) -  d_{C_o,1}(S_i, X_i, O) \right] \bigg|^2\Bigg]\Biggr)^{1/2} \\
        & \quad + ((1-\epsilon)/\epsilon^2)^{2}
        \Biggr(\mathbb{E}\Bigg[\bigg| \left(\frac{\varphi \rho(X_i)}{\rho^2(S_i, X_i)} -  
        \frac{\widetilde{\varphi}\widetilde{\rho}(X_i)}{\widetilde{\rho}^2(S_i, X_i)} \right) d_{C_o, 1}(S_i, X_i, O) \bigg|^2\Bigg]\Biggr)^{1/2} \\
        & +
         ((1-\epsilon)/\epsilon^2)^{2}
        \Biggr(\mathbb{E}\Bigg[\bigg| \frac{\varphi \rho(X_i)}{\rho(S_i, X_i)} \left[\widetilde{\mu}_{C_o,1}(S_i, X_i, O) -  \mu_{C_o,1}(S_i, X_i, O) \right] \bigg|^2\Bigg]\Biggr)^{1/2} \\
        & \quad + ((1-\epsilon)/\epsilon^2)^{2}
        \Biggr(\mathbb{E}\Bigg[\bigg| \left(\frac{\varphi \rho(X_i)}{\rho(S_i, X_i)} -  
        \frac{\widetilde{\varphi}\widetilde{\rho}(X_i)}{\widetilde{\rho}(S_i, X_i)} \right) \mu_{C_o, 1}(S_i, X_i, O) \bigg|^2\Bigg]\Biggr)^{1/2} \\
        & \quad + 
        ((1-\epsilon)^2/\epsilon^2)^{2}
        \Biggr(\mathbb{E}\Bigg[\bigg| \frac{\varphi \rho(X_i)}{\rho^2(S_i, X_i)} \left[\widetilde{\mu}_{C_o,1}(S_i, X_i, O) -  \mu_{C_o,1}(S_i, X_i, O) \right] \bigg|^2\Bigg]\Biggr)^{1/2} \\
        & \quad + ((1-\epsilon)/\epsilon^2)^{2}
        \Biggr(\mathbb{E}\Bigg[\bigg| \left(\frac{\varphi \rho(X_i)}{\rho^2(S_i, X_i)} -  
        \frac{\widetilde{\varphi}\widetilde{\rho}(X_i)}{\widetilde{\rho}^2(S_i, X_i)} \right) \mu_{C_o, 1}(S_i, X_i, O) \bigg|^2\Bigg]\Biggr)^{1/2} \\
        & \le T_6 \delta_n
    \end{align*}
    where $T_6$ depends on $\epsilon$ and $T$ only. We can similarly show that $\| m_7(Z_i, \tau_{C_o}, \widetilde{\eta}) - m_7(Z_i, \tau_{C_o}, \eta) \|_{P, 2} < T_7 \delta_n$ where $T_7 > 0$ depends only on $T$ and $\epsilon$.

\textbf{Part 3: Verification of \cref{eq:rate-diff-m}}

\Cref{lemma:tech-remainer-copula} in the supplementary appendix implies that we have
\begin{align*}
    \mathbb{E}_P[m(W_i, \tau_{C_o}, \widetilde{\eta}) - m(W_i, \tau_{C_o}, \eta)] 
    & = 
    \sum_{j=1}^7 \mathbb{E}_P[m_j(W_i, \tau_{C_o}, \widetilde{\eta}) - m_j(W_i, \tau_{C_o}, \eta)] = \sum_{j=1}^{14} \mathcal{J}_j,  
\end{align*} 
where
{\footnotesize
\begin{align*}
    \mathcal{J}_1 & = 
    - \mathbb{E}_P\Bigg[\bigg(\frac{\mathds{1}_{P_i = E}}{\widetilde{\varphi}} \frac{W_i}{\widetilde{\rho}(X_i)} - \frac{\mathds{1}_{P_i = E}}{\widetilde{\varphi}} \frac{W_i}{\rho(X_i)} \bigg)(\widetilde{\bar{\mu}}_{C_+, 1}(1, X_i) - \bar{\mu}_{C_+, 1}(1, X_i)) \Bigg] \\
    \mathcal{J}_2 & = 
    \mathbb{E}_P\Bigg[\bigg(\frac{\mathds{1}_{P_i = E}}{\widetilde{\varphi}} \frac{1-W_i}{1-\widetilde{\rho}(X_i)} - \frac{\mathds{1}_{P_i = E}}{\widetilde{\varphi}} \frac{1-W_i}{1-\rho(X_i)} \bigg) (\widetilde{\bar{\mu}}_{C_+, 0}(1, X_i) - \bar{\mu}_{C_+, 0}(1, X_i)) \Bigg]
    \end{align*}

\begin{align*}
 \mathcal{J}_3   & =
\mathbb{E}_P\Bigg[\frac{\mathds{1}_{P_i = O}}{\widetilde{\varphi}} \frac{\widetilde{\varphi}(S_i, X_i)}{1-\widetilde{\varphi}(S_i, X_i)} \frac{1}{\widetilde{\rho}(X_i)}  \\
& \quad \quad \times (\tilde{h}_{C_o, 1}(Y_i, \widetilde{F}_Y^{-1}(\cdot|S_i, X_i, O), \widetilde{\rho}(S_i, X_i)) - \tilde{h}_{C_o, 1}(Y_i, F_Y^{-1}(\cdot|S_i, X_i, O), \widetilde{\rho}(S_i, X_i))\Bigg] 
\\
\mathcal{J}_4 &  =
\mathbb{E}_P\Bigg[\frac{\mathds{1}_{P_i = O}}{\widetilde{\varphi}} \bigg(\frac{\widetilde{\varphi}(S_i, X_i)}{1-\widetilde{\varphi}(S_i, X_i)} - \frac{\varphi(S_i, X_i)}{1-\varphi(S_i, X_i)} \bigg) \frac{1}{\widetilde{\rho}(X_i)}  \\
& \quad \quad \times (\tilde{h}_{C_o, 1}(Y_i, F_Y^{-1}(\cdot|S_i, X_i, O), \widetilde{\rho}(S_i, X_i)) - \tilde{h}_{C_o, 1}(Y_i, F_Y^{-1}(\cdot|S_i, X_i, O), \rho(S_i, X_i)))\Bigg] 
\\
\mathcal{J}_5 & =  - 
\mathbb{E}_P\Bigg[\frac{\mathds{1}_{P_i = O}}{\widetilde{\varphi}} \bigg(\frac{\widetilde{\varphi}(S_i, X_i)}{1-\widetilde{\varphi}(S_i, X_i)} - \frac{\varphi(S_i, X_i)}{1-\varphi(S_i, X_i)} \bigg) \frac{1}{\widetilde{\rho}(X_i)}  (\widetilde{\rho}(S_i, X_i)  - \rho(S_i, X_i))\widetilde{\mu}_{C_o, 1}(S_i, X_i, O)\Bigg]
\\
\mathcal{J}_6 & = - 
\mathbb{E}_P\Bigg[\frac{\mathds{1}_{P_i = O}}{\widetilde{\varphi}} \bigg(\frac{\widetilde{\varphi}(S_i, X_i)}{1-\widetilde{\varphi}(S_i, X_i)} - \frac{\varphi(S_i, X_i)}{1-\varphi(S_i, X_i)} \bigg) \frac{1}{\widetilde{\rho}(X_i)} \rho(S_i, X_i) (\widetilde{\mu}_{C_o, 1}(S_i, X_i, O) -  \mu_{C_o, 1}(S_i, X_i, O))\Bigg] \\
\mathcal{J}_7 & = 
\mathbb{E}_P\Bigg[\frac{\mathds{1}_{P_i = O}}{\widetilde{\varphi}} \frac{\varphi(S_i, X_i)}{1-\varphi(S_i, X_i)} \frac{1}{\widetilde{\rho}(X_i)} (\tilde{h}_{C_o, 1}(Y_i, F_Y^{-1}(\cdot|S_i, X_i, O), \widetilde{\rho}(S_i, X_i)) - \tilde{h}_{C_o, 1}(Y_i, F_Y^{-1}(\cdot|S_i, X_i, O), \rho(S_i, X_i)))\Bigg] 
\end{align*}

{\scriptsize
\begin{align*}
\mathcal{J}_8 & = 
- \mathbb{E}_P\Bigg[\frac{\mathds{1}_{P_i = O}}{\widetilde{\varphi}} \frac{\widetilde{\varphi}(S_i, X_i)}{1-\widetilde{\varphi}(S_i, X_i)} \frac{1}{1-\widetilde{\rho}(X_i)}  \\
& \quad \quad \times (\tilde{h}_{C_o, 0}(Y_i, \widetilde{F}_Y^{-1}(\cdot|S_i, X_i, O), 1-\widetilde{\rho}(S_i, X_i)) - \tilde{h}_{C_o, 0}(Y_i, F_Y^{-1}(\cdot|S_i, X_i, O), 1-\widetilde{\rho}(S_i, X_i))\Bigg] 
\\
\mathcal{J}_9 & =  -
\mathbb{E}_P\Bigg[\frac{\mathds{1}_{P_i = O}}{\widetilde{\varphi}} \bigg(\frac{\widetilde{\varphi}(S_i, X_i)}{1-\widetilde{\varphi}(S_i, X_i)} - \frac{\varphi(S_i, X_i)}{1-\varphi(S_i, X_i)} \bigg) \frac{1}{1-\widetilde{\rho}(X_i)}  \\
& \quad \quad \times (\tilde{h}_{C_o, 0}(Y_i, F_Y^{-1}(\cdot|S_i, X_i, O), 1-\widetilde{\rho}(S_i, X_i)) - \tilde{h}_{C_o, 0}(Y_i, F_Y^{-1}(\cdot|S_i, X_i, O), 1-\rho(S_i, X_i)))\Bigg] 
\\
\mathcal{J}_{10} & =  -
\mathbb{E}_P\Bigg[\frac{\mathds{1}_{P_i = O}}{\widetilde{\varphi}} \bigg(\frac{\widetilde{\varphi}(S_i, X_i)}{1-\widetilde{\varphi}(S_i, X_i)} - \frac{\varphi(S_i, X_i)}{1-\varphi(S_i, X_i)} \bigg) \frac{1}{1-\widetilde{\rho}(X_i)}  (\widetilde{\rho}(S_i, X_i)  - \rho(S_i, X_i))\widetilde{\mu}_{C_o, 0}(S_i, X_i, O)\Bigg]
\\
\mathcal{J}_{11} & = 
\mathbb{E}_P\Bigg[\frac{\mathds{1}_{P_i = O}}{\widetilde{\varphi}} \bigg(\frac{\widetilde{\varphi}(S_i, X_i)}{1-\widetilde{\varphi}(S_i, X_i)} - \frac{\varphi(S_i, X_i)}{1-\varphi(S_i, X_i)} \bigg) \frac{1-\rho(S_i, X_i)}{1-\widetilde{\rho}(X_i)} (\widetilde{\mu}_{C_o, 0}(S_i, X_i, O) -  \mu_{C_o, 0}(S_i, X_i, O))\Bigg] \\
\mathcal{J}_{12} & =  -
\mathbb{E}_P\Bigg[\frac{\mathds{1}_{P_i = O}}{\widetilde{\varphi}} \frac{\varphi(S_i, X_i)}{1-\varphi(S_i, X_i)} \frac{1}{\widetilde{\rho}(X_i)} (\tilde{h}_{C_o, 0}(Y_i, F_Y^{-1}(\cdot|S_i, X_i, O), 1-\widetilde{\rho}(S_i, X_i)) - \tilde{h}_{C_o, 0}(Y_i, F_Y^{-1}(\cdot|S_i, X_i, O), 1-\rho(S_i, X_i)))\Bigg] 
\end{align*}
}

\begin{align*}
    \mathcal{J}_{13} & = 
    - 
    \mathbb{E}_P\Bigg[\frac{\mathds{1}_{P_i = E}}{\widetilde{\varphi}} \frac{1}{\widetilde{\rho}(X_i)} \widetilde{d}_{C_o, 1}(S_i, X_i) \left(\widetilde{\rho}(S_i, X_i) - \rho(S_i, X_i)\right)\Bigg] \\
    \mathcal{J}_{14} & =  - 
    \mathbb{E}_P\Bigg[\frac{\mathds{1}_{P_i = E}}{\widetilde{\varphi}} \frac{1}{1-\widetilde{\rho}(X_i)} \widetilde{d}_{C_o, 0}(S_i, X_i) \left(\widetilde{\rho}(S_i, X_i) - \rho(S_i, X_i)\right)\Bigg].
\end{align*}
}

Because $Y$ is bounded and $\mathcal{P}(\epsilon \le \varphi(S_i, X_i)\le 1-\epsilon) = 1$ and $\mathcal{P}(\epsilon \le \rho(S_i, X_i)\le 1-\epsilon) = 1$, we have the followings bounds for $\mathcal{J}_j$ for $j=1, \dotsc, 14$.

(1) Bounds for $\mathcal{J}_1$,  $\mathcal{J}_2$, $\mathcal{J}_6$, $\mathcal{J}_{11}$.

\begin{align*}
    |\mathcal{J}_1| & \le M \| \widetilde{\rho}(X_i) - \rho(X_i) \|_{P, 2} \times \| \widetilde{\bar{\mu}}_{C_+, 1}(1, X_i) - \bar{\mu}_{C_+, 1}(1, X_i) \|_{P, 2} \le M n^{-1/2} \delta_n , \\
    |\mathcal{J}_2| & \le M \| \widetilde{\rho}(X_i) - \rho(X_i) \|_{P, 2} \times \| \widetilde{\bar{\mu}}_{C_+, 0}(0, X_i) - \bar{\mu}_{C_+, 0}(0, X_i) \|_{P, 2} \le M n^{-1/2} \delta_n , \\
    |\mathcal{J}_6| & \le M \| \widetilde{\varphi}(S_i, X_i) - \varphi(S_i, X_i) \|_{P, 2} \times \| \widetilde{\mu}_{C_+, 1}(S_i, X_i, O) - \mu_{C_+, 1}(S_i, X_i, O) \|_{P, 2} \le M n^{-1/2} \delta_n , \\
    |\mathcal{J}_{11}| & \le M \| \widetilde{\varphi}(S_i, X_i) - \varphi(S_i, X_i) \|_{P, 2} \times \| \widetilde{\mu}_{C_+, 0}(S_i, X_i, O) - \mu_{C_+, 0}(S_i, X_i, O) \|_{P, 2} \le M n^{-1/2} \delta_n ,
\end{align*}
for some absolute constant $M$ that only depends on $\epsilon$ and $C$.

(2) Bounds for $\mathcal{J}_5$,  $\mathcal{J}_{10}$.

Note that
{\footnotesize
\begin{align*}
    \mathcal{J}_5 & = -  
\mathbb{E}_P\Bigg[\frac{\mathds{1}_{P_i = O}}{\widetilde{\varphi}} \bigg(\frac{\widetilde{\varphi}(S_i, X_i)}{1-\widetilde{\varphi}(S_i, X_i)} - \frac{\varphi(S_i, X_i)}{1-\varphi(S_i, X_i)} \bigg) \frac{1}{\widetilde{\rho}(X_i)}  (\widetilde{\rho}(S_i, X_i)  - \rho(S_i, X_i)) \mu_{C_o, 1}(S_i, X_i, O)\Bigg] \\
& \quad  - 
\mathbb{E}_P\Bigg[\frac{\mathds{1}_{P_i = O}}{\widetilde{\varphi}} \bigg(\frac{\widetilde{\varphi}(S_i, X_i)}{1-\widetilde{\varphi}(S_i, X_i)} - \frac{\varphi(S_i, X_i)}{1-\varphi(S_i, X_i)} \bigg) \frac{1}{\widetilde{\rho}(X_i)}  (\widetilde{\rho}(S_i, X_i)  - \rho(S_i, X_i))(\widetilde{\mu}_{C_o, 1}(S_i, X_i, O) - \mu_{C_o, 1}(S_i, X_i, O))\Bigg].
\end{align*}
}
Therefore, we have
\begin{align*}
    |\mathcal{J}_5| 
    & \le T_{\epsilon} \| \widetilde{\varphi}(S_i, X_i) - \varphi(S_i, X_i) \|_{P, 2} \times \| \widetilde{\mu}_{C_+, 1}(S_i, X_i, O) - \mu_{C_+, 1}(S_i, X_i, O) \|_{P, 2}  \\
    & \quad +  T_{\epsilon} \| \widetilde{\varphi}(S_i, X_i) - \varphi(S_i, X_i) \|_{P, 2} \times \| \widetilde{\rho}(S_i, X_i) - \rho(S_i, X_i) \|_{P, 2} \\
    & \le T_{\epsilon} n^{-1/2} \delta_n,
\end{align*}
where $T_{\epsilon}$ is a generic absolute constant which only depends on $T$ and $\epsilon$. 

Similarly, we can show that 
\begin{align*}
    |\mathcal{J}_{10}| 
    & \le T_{\epsilon} \| \widetilde{\varphi}(S_i, X_i) - \varphi(S_i, X_i) \|_{P, 2} \times \| \widetilde{\mu}_{C_+, 0}(S_i, X_i, O) - \mu_{C_+, 0}(S_i, X_i, O) \|_{P, 2}  \\
    & \quad +  T_{\epsilon} \| \widetilde{\varphi}(S_i, X_i) - \varphi(S_i, X_i) \|_{P, 2} \times \| \widetilde{\rho}(S_i, X_i) - \rho(S_i, X_i) \|_{P, 2} \\
    & \le T_{\epsilon} n^{-1/2} \delta_n .
\end{align*}

(3) Bounds for $\mathcal{J}_3$,  $\mathcal{J}_{8}$.

Similar to the proof in Part 3-3-2-(3) \cref{lemma:DML-verification-wc}, we can show that
\begin{align*}
    |\mathcal{J}_3| 
    & \le T_{\epsilon} \mathbb{E}\Bigg[\int_0^1 \left(\widetilde{F}_Y^{-1}(u|S_i, X_i, O) - F_Y^{-1}(u|S_i, X_i, O)\right)^2 d(1-C_o(1-\widetilde{\rho}(S_i, X_i)|u))  \Bigg] \\
    & \le T_{\epsilon} \mathbb{E}\Bigg[\int_0^1 \left(\widetilde{F}_Y^{-1}(u|S_i, X_i, O) - F_Y^{-1}(u|S_i, X_i, O)\right)^2 du  \Bigg] \\
    \le T_{\epsilon} n^{-1/2} \delta_n, 
\end{align*}
where the last inequality works since $\sup_{\alpha \in [\epsilon, 1-\epsilon], u \in [0, 1]}\left|\frac{\partial C_o(\alpha|u)}{\partial u}\right|$ are bounded by \cref{assumption:regularity-copula} (3).

Similarly, we can show that
\begin{align*}
    |\mathcal{J}_8| 
    & \le T_{\epsilon} \mathbb{E}\Bigg[\int_0^1 \left(\widetilde{F}_Y^{-1}(u|S_i, X_i, O) - F_Y^{-1}(u|S_i, X_i, O)\right)^2 d(1-C_o(1-\widetilde{\rho}(S_i, X_i)|u))  \Bigg] \\
    & \le T_{\epsilon} \mathbb{E}\Bigg[\int_0^1 \left(\widetilde{F}_Y^{-1}(u|S_i, X_i, O) - F_Y^{-1}(u|S_i, X_i, O)\right)^2 du  \Bigg] \\
    & \le T_{\epsilon} n^{-1/2} \delta_n 
\end{align*}

(4) Bounds for $\mathcal{J}_4$,  $\mathcal{J}_{9}$.

Note that we can write $\mathcal{J}_4$ as follows.
{\footnotesize
\begin{align*}
    \mathcal{J}_4 
    &  =
\mathbb{E}_P\Bigg[\frac{\mathds{1}_{P_i = O}}{\widetilde{\varphi}} \bigg(\frac{\widetilde{\varphi}(S_i, X_i)}{1-\widetilde{\varphi}(S_i, X_i)} - \frac{\varphi(S_i, X_i)}{1-\varphi(S_i, X_i)} \bigg) \frac{1}{\widetilde{\rho}(X_i)}  \\
& \quad \quad \times \bigg(\int_0^1 F_Y^{-1}(u|S_i, X_i, O)(1-C_o(1-\widetilde{\rho}(S_i, X_i)|u))du - \int_0^1 F_Y^{-1}(u|S_i, X_i, O)(1- C_o(1-\rho(S_i, X_i)|u)) du\bigg)\Bigg].
\end{align*}
}
by \cref{lemma:dual-SpectralRiskMeasure}. 

Then, under \cref{assumption:regularity-copula}, we have the following.
\begin{align*}
    \mathcal{J}_4 & \le T_{\epsilon} \| \widetilde{\varphi}(S_i, X_i) - \varphi(S_i, X_i) \|_{P, 2} \times \| \widetilde{\rho}(S_i, X_i) - \rho(S_i, X_i) \|_{P, 2} \le M n^{-1/2} \delta_n. 
\end{align*}

(5) Bounds for $\mathcal{J}_7+\mathcal{J}_{13}$,  $\mathcal{J}_{12}+\mathcal{J}_{14}$.

We can write $\mathcal{J}_7+\mathcal{J}_{13}$ as follows.
{\footnotesize
\begin{align*}
    & \mathcal{J}_7+\mathcal{J}_{13} \\
    & =
    \mathbb{E}_P \Bigg[\frac{\mathds{1}_{P_i = E}}{\widehat{\varphi}}\frac{1}{\widetilde{\rho}(S_i, X_i)} \\
    & \quad \times \bigg(- \int_0^1 F_Y^{-1}(u|S_i, X_i, O) C_o(1-\widetilde{\rho}(S_i, X_i)|u)du  + \int_0^1 F_Y^{-1}(u|S_i, X_i, O) C_o(1-\rho(S_i, X_i)|u) du \\
    & \quad \quad \quad - (\widetilde{\rho}(S_i, X_i) - \rho(S_i, X_i)) \int_0^1 F_Y^{-1}(u|S_i, X_i, O) c_o(1-\widetilde{\rho}(S_i, X_i)|u) du  \bigg) \\
    & \quad \quad \quad - (\widetilde{\rho}(S_i, X_i) - \rho(S_i, X_i)) \int_0^1 (\widetilde{F}_Y^{-1}(u|S_i, X_i, O) - F_Y^{-1}(u|S_i, X_i, O)) c_o(1-\widetilde{\rho}(S_i, X_i)|u) du \bigg)
    \Bigg].
\end{align*}
}
Under \cref{assumption:regularity-copula}, we can see that
\begin{align*}
    |\mathcal{J}_7+\mathcal{J}_{13}|
    & \le T_{\epsilon} \| \widehat{\rho}(S_i, X_i) - \rho(S_i, X_i) \|_{P, 2}^2 \\
    & \quad + T_{\epsilon} \| \widehat{\rho}(S_i, X_i) - \rho(S_i, X_i) \|_{P, 2} \left(\mathbb{E}\left[\int_0^1 (\widetilde{F}_Y^{-1}(u|S_i, X_i, O) - F_Y^{-1}(u|S_i, X_i, O))^2 du\right]\right)^{1/2} \\
    & \le T_{\epsilon} n^{-1/2} \delta_n 
\end{align*}
Similarly, we can show that
\begin{align*}
    |\mathcal{J}_{12}+\mathcal{J}_{14}|
    & \le M \| \widehat{\rho}(S_i, X_i) - \rho(S_i, X_i) \|_{P, 2}^2 \\
    & \quad + M \| \widehat{\rho}(S_i, X_i) - \rho(S_i, X_i) \|_{P, 2} \left(\mathbb{E}\left[\int_0^1 (\widetilde{F}_Y^{-1}(u|S_i, X_i, O) - F_Y^{-1}(u|S_i, X_i, O))^2 du\right]\right)^{1/2} \\
    & \le T_{\epsilon} n^{-1/2} \delta_n 
\end{align*}

All computations in (1)-(5) in Part 3-2 implies that
\begin{align*}
    \| \mathbb{E}_P[m_{C_o}(Z_i, \tau_{C_o}, \widetilde{\eta}) - m_{C_o}(Z_i, \tau_{C_o}, \eta)] \| \le \left\| \sum_{j=1}^{14} \mathcal{J}_{j} \right\| \le T_{\epsilon} n^{-1/2} \delta_n  
\end{align*}
for some positive constant $T_{\epsilon}$ which only depends on $T$ and $\epsilon$.

\textbf{Step 4: Asymptotic Variance}

With the similar calculation as in Part 4 in the proof of \Cref{lemma:DML-verification-wc}, we can show that
\begin{align*}
    \mathbb{E}[m_{C_o}^2(Z_i, \tau, \eta)] \ge C_{\epsilon} \mathbb{E}[(\bar{\mu}_{C_o, 1}(1, X_i) - \bar{\mu}_{C_o, 0}(0, X_i) - \tau_{C_o})^2] + C_{\epsilon}
     \mathbb{E}[V(S_i, X_i, O)] > 0,
\end{align*}
where $C_{\epsilon}$ depends on $\epsilon$ and
\begin{align*}
    V(S_i, X_i, O) = \mathrm{Var}\left(h_{C_o, 1}(Y_i, F_Y^{-1}, 1-\rho(S_i, X_i)) - h_{C_o, 0}(Y_i, F_Y^{-1}, 1-\rho(S_i, X_i)) \right).
\end{align*}

\end{proof}


\end{document}